\documentclass[acmsmall,screen,nonacm]{acmart}
\newif\ifarXiv
 \arXivfalse

\newif\ifdraft%
\draftfalse%

\newif\ifCR%
\CRfalse%

\newif\ifmargins%
\marginsfalse%

\newif\ifanon
\anontrue

\usepackage[T1]{fontenc}
\usepackage[utf8]{inputenc}
\usepackage{thmtools}
\usepackage{enumerate}
\usepackage[compat=0.6]{yquant}
\usepackage{quantikz}
\usepackage{breqn}
\usepackage{bussproofs}
\usepackage{stmaryrd}
\usepackage[inkscapeformat=png]{svg}
\usepackage{framed}
\usepackage{wrapfig}
\usepackage{currfile}
\usepackage{subcaption}
\usepackage{float}
\usepackage{arydshln}
\usepackage{dsfont} 
\usepackage{multirow}
\usepackage{lipsum}
\usepackage{xparse}
\ifdraft%
    \usepackage[draft]{fixme}
\else
    \usepackage[final]{fixme}
\fi
\fxusetheme{colorsig} 
\FXRegisterAuthor{cc}{cci}{Christophe}
\FXRegisterAuthor{rp}{rpi}{Romain}
\FXRegisterAuthor{ji}{jii}{Jad}
\makeatletter
\@ifpackageloaded{listings}{}{
  \usepackage{listings}
}
\makeatother
\usepackage{xcolor}
\usepackage{graphicx} 
\usepackage[capitalise]{cleveref}
\crefname{figure}{Fig.}{Figs.}
\Crefname{figure}{Figure}{Figures}
\crefname{table}{Table}{Tables}
\Crefname{table}{Table}{Tables}
\crefname{section}{\S}{\S}
\Crefname{section}{Section}{Sections}
\crefname{definition}{Def.}{Defs.}
\Crefname{definition}{Definition}{Definitions}
\makeatletter
\AddToHook{cmd/appendix/before}{\def\cref@section@alias{appendix}}
\AddToHook{cmd/appendix/before}{\def\cref@subsection@alias{appendix}}
\makeatother

\newcommand{\newthmparent}[2]{%
    \newcounter{subctr@#1}
    \expandafter\renewcommand\csname thesubctr@#1\endcsname{\ref*{#2}\alph{subctr@#1}}
    \spnewtheorem{sublemma@#1}[subctr@#1]{Theorem }{\itshape\/}{\rmfamily}
}
\definecolor{darkred}{rgb}{0.8,0.1,0.1}

\Crefname{conjecture}{Conjecture}{Conjectures}

\theoremstyle{remark} 
\newtheorem*{remark}{Remark}

\makeatletter
\newcommand{\labelthis}[2]{%
  \def\@currentlabel{#2}\label{#1}#2%
}
\makeatother

\makeatletter
\newcommand{\labelrule}[2]{%
    \def\@currentlabel{#2}\label{#1}#2%
}
\makeatother

\newif\ifalternative

\newcommand{\HQbricks}{\textsc{HQbricks}}
\newcommand{\lang}{\textsc{HQbricks}\ensuremath{_\infty}}
\newcommand{\Lang}{\textsc{HQbricks}\ensuremath{_\infty}}
\newcommand{\hps}{\texttt{HPS}}
\newcommand{\ihps}{\texttt{IHPS}}
\newcommand{\libname}{\texttt{IHPSlib}}

\newcommand{\N}{\ensuremath{\mathbb{N}}}
\newcommand{\Z}{\ensuremath{\mathbb{Z}}}
\newcommand{\R}{\ensuremath{\mathbb{R}}}
\newcommand{\C}{\ensuremath{\mathbb{C}}}
\newcommand{\Q}{\ensuremath{\mathbb{Q}}}

\newcommand{\B}{\ensuremath{\mathbb{B}}}

\newcommand{\while}[2]{\ensuremath{\mathbf{while}\; #1\; \mathbf{do}\; #2\;
\mathbf{done}}}
\newcommand{\whileAnnot}[3]{\ensuremath{\mathbf{while}\; \{#3\}\; #1\;
\mathbf{do}\; #2\; \mathbf{done}}}
\newcommand{\whileBounded}[3]{\ensuremath{\mathbf{while}_{#1}\; #2\; \mathbf{do}\;
#3\ \mathbf{end}}}
\newcommand{\ifthene}[3]{\ensuremath{\mathbf{if}\; #1\;
            \allowbreak\mathbf{then}\;\allowbreak #2\;\mathbf{ else }
\;#3 \;\mathbf{end}}}
\newcommand{\dowhile}[2]{\ensuremath{\mathbf{do}\; #1\; \mathbf{while}\; #2}}

\newcommand{\prog}{{\ensuremath{\mathtt{p}}}}
\newcommand{\skipProg}{{\ensuremath{\mathbf{skip}}}}
\newcommand{\sequenceProg}[2]{{\ensuremath{#1 \mathbin{\mathtt{;}} #2}}}
\newcommand{\assign}[2]{{\ensuremath{#1 \mathbin{\mathtt{:=}} #2}}}
\newcommand{\bool}{\ensuremath{{\mathtt{\ensuremath{b}}}}}

\renewcommand{\measure}[2]{{\ensuremath{#1 := \mathbf{measure}\; #2}}}
\newcommand{\qubitInitProg}[1]{\ensuremath{\mathbf{qubit}\; #1}}
\newcommand{\cbitInitProg}[1]{\ensuremath{\mathbf{bit}\; #1}}
\newcommand{\intInitProg}[1]{\ensuremath{\mathbf{int}\; #1}}
\newcommand{\qubit}{\ensuremath{\mathtt{\ensuremath{q}}}}
\newcommand{\qbit}{\qubit} 
\newcommand{\cbit}{\ensuremath{\mathtt{\ensuremath{c}}}}
\newcommand{\trueBool}{\ensuremath{\mathtt{tt}}}
\newcommand{\falseBool}{\ensuremath{\mathtt{ff}}}
\newcommand{\unitary}{\ensuremath{\mathtt{U}}}
\newcommand{\integer}{\ensuremath{\mathtt{i}}}

\newcommand{\cInt}{\ensuremath{\mathtt{x}}}
\newcommand{\programSet}{\ensuremath{\mathtt{Prog}}}
\newcommand{\addressSet}{\ensuremath{\mathtt{A}}}
\newcommand{\qbitSet}{\ensuremath{\mathtt{Q}}}
\newcommand{\cbitSet}{\ensuremath{\mathtt{B}}}
\newcommand{\cIntSet}{\ensuremath{\mathtt{I}}}
\newcommand{\classAddressSet}{\ensuremath{\mathtt C}}
\renewcommand{\address}{\ensuremath{\mathtt{a}}}
\newcommand{\programTerm}{\ensuremath{\mathtt{t}}}

\newcommand{\signature}{\ensuremath{\mathfrak s}}
\newcommand{\validProg}[3]{\ensuremath{#1 \vdash #2 : #3}}
\newcommand{\hoare}[3]{\ensuremath{\left\{ #1 \right\}\; #2\; \left\{ #3 \right\}}}

\newcommand{\intType}{\ensuremath{\mathtt{Int}}}
\newcommand{\boolType}{\ensuremath{\mathtt{Bool}}}

\newcommand{\phaseType}{\ensuremath{\mathtt{Phase}}}
\newcommand{\normType}{\ensuremath{\mathtt{Norm}}}
\newcommand{\memoryType}{\ensuremath{\mathtt{Memory}}}
\newcommand{\hpsType}{\ensuremath{\mathtt{IHPS}}}

\newcommand{\boolVar}{c}
\newcommand{\intVar}{x}

\newcommand{\intVarOther}{y}
\newcommand{\var}{a} 
\newcommand{\boolTerm}{b}
\newcommand{\intTerm}{i}
\newcommand{\lift}[1]{\ensuremath{\mathop{\uparrow}} #1}
\newcommand{\freeVars}[1]{\ensuremath{\mathrm{FV}(#1)}}

\newcommand{\varValue}{\ensuremath{v}} 
\newcommand{\intValue}{\varValue} 
\newcommand{\classicalState}{\sigma}

\newcommand{\boolVarSet}{\ensuremath{B}}
\newcommand{\intVarSet}{\ensuremath{I}}
\newcommand{\varSet}{\ensuremath{A}}

\renewcommand{\ket}[2][{}]{\ensuremath{{\left| #2 \right\rangle}_{#1}}}
\renewcommand{\bra}[1]{\ensuremath{\left\langle #1 \right|}}

\newcommand{\ketbra}[2]{\ensuremath{\left| #1 \right\rangle \left\langle #2 \right|}}
\newcommand{\classKet}[2][{}]{\ensuremath{{\left[ {#2} \right]}_{{#1}}}}

\newcommand{\pastKet}[2][{}]{\ensuremath{{\left( {#2} \right)}_{{#1}}}}

\newcommand{\finsubseteq}{\mathrel{\subseteq_{\mathrm{fin}}}}

\newcommand{\equaldef}{\ensuremath{\stackrel{\smash{\raisebox{-0.2ex}{\scalebox{0.5}[0.5]{\textnormal{def}}}}}{=}}}

\newcommand{\ultraperp}{\perp\!\!\!\perp}
\newcommand{\ot}{\leftarrow}

\newcommand{\h}{\ensuremath{h}}
\newcommand{\pasum}[3]{\ensuremath{\left\langle #1,#3\cdot#2\right\rangle}}
\newcommand{\pasuminline}[3]{\ensuremath{\langle #1,#3\cdot#2\rangle}}
\newcommand{\psadd}{\ensuremath{+}}
\newcommand{\bigpsadd}{\sum}
\newcommand{\sqpsadd}{\oplus}

\newcommand{\sumpasumsplit}[3]{%
  &&\ensuremath{\left\langle #1 \vphantom{#1#2#3} \right.} &
\ensuremath{\left. , #3     \vphantom{#1#2#3} \right.} &
\ensuremath{\left. \cdot \, #2 \vphantom{#1#2#3} \right\rangle}&
}

\newcommand{\sumpasummemcols}[1]{r@{}r@{\,}l@{}l@{\,}#1@{}r}

\newcommand{\assignmentContext}{\Gamma}
\NewDocumentCommand{\hilbert}{o}{
  \IfNoValueTF{#1} {\ensuremath{\mathcal{H}}} {\ensuremath{\mathcal{H}(#1)}} }
\NewDocumentCommand{\fockSpace}{d()}{ 
  \IfNoValueTF{#1} {\ensuremath{\mathcal{F}}} {\ensuremath{\mathcal{F}(#1)}} }
\NewDocumentCommand{\basis}{d()}{ 
  \IfNoValueTF{#1} {\ensuremath{\mathrm{Basis}}} {\ensuremath{\mathrm{Basis}(#1)}} }
\NewDocumentCommand{\probability}{d()}{ 
  \IfNoValueTF{#1} {\ensuremath{\mathbb{P}}} {\ensuremath{\mathbb{P}(#1)}} }
\NewDocumentCommand{\expectation}{o}{
  \IfNoValueTF{#1} {\ensuremath{\mathbb{E}}} {\ensuremath{\mathbb{E}[#1]}} }
\newcommand{\pastSpaceDecor}{\ensuremath{\mathcal P}}

\newcommand{\dephase}{\ensuremath{\mathcal D}}
\newcommand{\superop}{\ensuremath{\mathcal E}}

\newcommand{\sem}[1]{\left\llbracket{#1}\right\rrbracket}
\newcommand{\sembreak}[1]{\llbracket{#1}\rrbracket}
\newcommand{\semRho}[2][{}]{\ensuremath{{\mathrm{cq}_{#1}}(#2)}}
\newcommand{\semFock}[2][{}]{{\ensuremath{{\left\{\!\!\left\{ {#2}
\right\}\!\!\right\}}_{{#1}}}}}

\newcommand{\redotimes}[1]{\ensuremath{\mathrm{red}_{\otimes}\left(#1\right)}}
\newcommand{\redsqpsadd}[1]{\ensuremath{\mathrm{red}_{\sqpsadd}\left(#1\right)}}
\newcommand{\applySig}[2]{#1(#2)}

\newcommand{\applyU}[1]{\ensuremath{\mathtt{apply}(#1)}}
\newcommand{\eval}[2]{\ensuremath{\mathtt{ev}_{#2}(#1)}}
\newcommand{\CQ}{\ensuremath{\mathrm{CQ}}}

\ifdraft%
    \newcommand{\romain}[1]{{\color{orange!85!black!80}{RP: #1}}} 
\else
    \newcommand{\romain}[1]{}
\fi

\ifdraft
    
\else
    
\fi
\newcommand{\judgement}[3]{\ensuremath{#1 \vdash #3}}
\newcommand{\logicContext}{\ensuremath{\Delta}}
\newcommand{\equivContext}{\ensuremath{\Xi}}
\newcommand{\semJudgement}[3]{\ensuremath{#1\models #3}}
\newcommand{\equivvdash}{\ensuremath{\vdash_{\equiv}}}
\newcommand{\dom}[1]{\ensuremath{\mathrm{dom}(#1)}}
\newcommand{\equivStrong}{\ensuremath{\equiv_{s}}}
\newcommand{\equivWeak}{\ensuremath{\equiv}}

\ifdraft{}\else

\fi
\newcommand{\appx}{\ensuremath{\ddagger}}
\newcommand{\hoareVar}{\ensuremath{S}}

\begin{document}
\newcommand{\thetitle}{An Effective Quantum Hoare Logic for Hybrid Quantum Programs with Unbounded Loops}
\title{\thetitle}

\author{Jad Issa}
\email{jad.issa@{cea.fr,univ-lorraine.fr}}
\orcid{0009-0000-8595-728X}
\affiliation{
    \institution{CEA LIST/LSL, Qbricks Team}
    \city{Nancy}
    \country{France}
}
\affiliation{
    \institution{Université de Lorraine, CNRS, Inria, LORIA, MOCQUA team}
    \city{Paris}
    \country{France}
}

\author{Christophe Chareton}
\email{christophe.chareton@cea.fr}
\orcid{0000-0001-7113-563X}
\affiliation{
    \institution{CEA LIST/LSL, Qbricks Team}
    \city{Nancy}
    \country{France}
}

\author{Romain Péchoux}
\email{romain.pechoux@loria.fr}
\orcid{0000-0003-0601-5425}
\affiliation{
    \institution{Université de Lorraine, CNRS, Inria, LORIA, MOCQUA team}
    \country{France}
}

\begin{abstract}
While quantum hardware remains limited, hybrid quantum-classical algorithms with
complex control structures, including unbounded loops, are emerging, posing new
challenges for quantum program analysis, including the accurate estimation of
the resource consumption of a given program.  Meanwhile, precise analysis
techniques such as symbolic execution have largely left out hybridization and
unbounded recursion. On the other hand, current quantum Hoare logics that generally support
them are lacking in expressiveness and miss out on efficient computational equational
reasoning that could be implemented in a semi-automated tool. This leaves a gap
awaiting to be filled.  In this work, we answer this challenge with the first
semi-automated static analysis solution combining effective  functional
verification and resource (termination or cost) estimation for hybrid
quantum programs with unbounded loops.  Towards that end, we introduce
\emph{integer hybrid path-sums} ($\ihps{}$), extending path-sums to handle
unbounded \texttt{while} loops, as a representation of possible executions of a
program.  A generic strategy for determining termination and expected resource
consumption via loop invariants is also proposed and illustrated on several
examples. Finally, the solution is implemented as a semi-automatic Haskell
program. This work is the first step toward the design of a complete static
resource analysis tool for hybrid quantum programs, essential for the
development of real-world quantum computing.

\keywords{Quantum Computing \and Hybrid Programs \and Path-Sums \and
Static Analysis \and Resource Analysis  \and Termination Analysis \and Symbolic
Execution}

\end{abstract}

\maketitle

\section{Introduction}%
    \subsection{Context and Motivations}%
\label{sub:Context and motivation}

In the quest for higher computational power, quantum computing has
been under research since its introduction by Feynman and Benioff in the
early 1980s~\cite{B82}.  For much of its history, research in quantum computing
has focused on proving the advantage of quantum computers over classical
computers, be it theoretically (e.g., Shor's algorithm for factoring in
polynomial time~\cite{S94}) or experimentally (e.g., the demonstration
of the so-called quantum advantage by Google in 2024~\cite{Google24}).

However, a major challenge for integrated and scalable quantum computing is
the integration of purely quantum (also called unitary) sequences of computation
in a wider environment that includes classical operations and classical control
structures~\cite{VLRH23,DLPZ25}.  As such, today's quantum applications often
require \textit{hybrid} programming languages offering  hybrid unbounded
recursion as a core component of the computation. This requirement is, for
instance, notably  observed in all variants of \emph{repeat-until-success}
(RUS), where a given block of unitary computations is run until an identifiable
success condition is met.  Such patterns appear, for instance, in the
post-selection~\cite{Y26} model of computation, where a computation is simply
discarded if the measurement does not result in a desired outcome, as well as in
the synthesis of unitaries~\cite{BRS15,PS13,LBK05}. RUS has also been used to
apply \emph{eventually deterministic} 2-qubit gates in linear optical quantum
computing~\cite{LBB+06,LBK05,dGHE+24}, and it often appears in error correction,
for, e.g., the fault-tolerant preparation of logical states by repeatedly
preparing a faulty state and measuring whether it is correct~\cite{WPW26}.
Finally, most quantum algorithms (including Grover, Phase Estimation, Shor,
etc.) do not succeed with probability 1 in a single run and tend to benefit from
being repeated until success. 

The increasing importance of hybrid programs with unbounded recursion in quantum
computing calls for specialized verification techniques to answer both semantic
correctness properties (functional correctness, symbolic simulation, equivalence
checking) and properties that are specific to this class of programs, such as
their physicality or their termination.  This is precisely where formal methods
and verification come into play~\cite{CBDVVX23,FY20,U19}, particularly
considering the inherent difficulty of testing such programs~\cite{CBDVVX23}.

Due to the probabilistic nature of hybrid programs and because of the inclusion
of unbounded recursion, one of the fundamental aspects to be studied concerns
their \emph{resource-awareness}.  This `resource' can be anything from
termination to costs in time, gate count, or more. Estimating precise resource
properties of hybrid programs is critical in an age of limited quantum
resources and requires a very fine understanding of the behavior of
programs.

This need for the functional verification of unbounded hybrid programs has
received  attention from Quantum Hoare Logics~\cite{Y12,liu2019formal, Y24}
(QHL). While this direction opened a  fruitful working
program~\cite{FY21,LZB+25,SFY+26}, it still lacks critical components for
practical scaling. In particular, as  recently observed, the logics still rely
on infeasibly complex computations over matrices, and ``scalable verification
remains elusive''~\cite{YPR26} today.  Furthermore, the expressivity of such
logics, though showing advances very recently~\cite{SFY+26}, remains
insufficient for resource analysis.  As such, there is a need for logics that
have efficient computational aspects: it is no longer enough to produce
\emph{some} specification no matter its size, but rather to produce compact yet
expressive specifications equipped with effective equational theories to rewrite
them according to the different needs of the verification process. While
theoretical models have their place, it is today critical that the verification
process be effective, i.e., amenable to a semi-automatic implementation, and yet
their effectiveness is still lacking in the current literature.

On the other hand, another line of research focuses on such compact, efficient,
and expressive specifications using symbolic representations and equational
theories to reason about them, including
path-sums~\cite{A18,A23,C+21,CIN+26,V24}, diagrammatic calculi~\cite{CD11,CK18},
and more. However, these approaches do not handle unbounded recursion, and only
one very recent approach, \HQbricks~\cite{CIN+26} considers hybrid programs
while recognizing the need for unbounded recursion without providing a solution. 

A gap thus emerges between the expressiveness and computability of symbolic
approaches restricted to unitary programs and the applicability of QHLs to
hybrid unbounded programs without sufficient expressiveness and computability.
Meanwhile, the need, from the applicative side of quantum computing, for filling
this gap and complementing it with resource analysis is becoming more and more
evident.

    \subsection{Contributions}%
\label{sub:Contributions}
In this paper, we tackle this gap in the literature by developing a symbolic
analysis framework for hybrid programs with unbounded loops and classical
integer computations, with applications in the inference of fine guarantees on
the resources consumed by such programs written in the language $\lang$
introduced in this paper.

$\lang$ builds upon the open-source framework $\HQbricks$~\cite{CIN+26}, which
uses symbolic execution to avoid the exponential blow-up in representations of
hybrid quantum states along the execution of \emph{bounded} programs. We extend
\HQbricks{} by allowing for general recursion in the form of unbounded
\texttt{while} loops and computations on classical integers, possibly dependent
on measurement results.  In turn, we solve a fundamentally more difficult
representation problem arising from unbounded recursion: states may have
infinite supports as vectors, and the exponential complexity is raised to a
qualitatively more difficult \emph{undecidability barrier}. 

General recursion allows for non-termination; this implies that the complexity
properties we can guarantee for hybrid quantum programs are fairly subtle
properties that echo, to some extent, the properties of probabilistic
programming, such as almost-sure termination~\cite{BG05} (i.e., termination with
probability 1).  These properties also go well beyond the scope of properties
that can be addressed by (hybrid) quantum \emph{circuits}, which terminate by
construction.  Our work also provides general functional correctness guarantees
as well as guarantees on probabilistic properties of the program, including the
expected number of gates used, expected runtime, almost-sure termination, etc.

The analysis framework presented in this article is based on
\emph{path-sums}~\cite{A18}, a formalism introduced to provide a symbolic and
compact representation of quantum circuits in the context of formal
verification~\cite{CBDVVX23,V24,AL25}. The compactness of this formalism stems
from the fact that quantum operations, generally represented by operators in a
Hilbert space of exponential dimension (in the number of qubits), can be
symbolically represented by a symbolic sum of \emph{paths} (vectors), often
avoiding this exponential blow-up.  To study the properties of hybrid
quantum programs with general recursion, the path-sums formalism has to be
extended to handle both classical and quantum data as well as unbounded
loops. A recent proposition, \HQbricks~\cite{CIN+26}, presents the extension to
\emph{hybrid path-sums} (\hps{}) for the verification of hybrid but bounded
quantum programs. 

We further extend this formalism to \emph{integer HPS} (\ihps{}), which include
variables and constructs for handling integers. This allows us to make a major
leap in the analysis framework by replacing deterministic symbolic computation
with a Hoare logic, an equational theory for \ihps, and invariant-based
analysis of loops, a necessary step in resolving the undecidability issues
inherent in the analysis of unbounded loops. 

Our main contributions to the analysis of unbounded hybrid programs are the
following.

\begin{enumerate}
    \item The introduction of an imperative programming language $\lang$
        supporting hybrid quantum programs with unbounded \texttt{while} loops
        and computation on integers (\cref{fig:language-syntax}). \lang{} is a
        variant of the open-source language $\HQbricks$.
        It is equipped with a denotational semantics
        (\cref{fig:standard-semantics}) defined on Classical-Quantum states
        (\emph{CQ states}, \cref{def:cq-state}), a standard representation of
        hybrid states in terms of density operators and
        super-operators~\cite{W13,FY21}.

    \item The definition of \emph{integer HPS} or \ihps{} (\cref{def:hps}), a
        compact, exact, and symbolic representation of infinite-dimensional CQ
        states extending \hps{}~\cite{CIN+26} used for finite-dimensional CQ
        states. Integer HPS are interpretable as pure vectors in
        higher-dimensional Fock spaces (\cref{def:hps-interpretation}) as well
        as CQ states (\cref{def:density-operator-interpretation}). They are
        equipped with a natural notion of equivalence
        (\cref{def:equivalence-hps}) derivable by sound equational theories
        (\cref{sub:equational-theory}) extending those of
        path-sums~\cite{A18,A23,CIN+26}.

    \item A quantum Hoare logic for $\lang$ in terms of
        transformations of integer HPS (\cref{fig:hoare-logic}) by programs, or
        by substitutions of equivalent \ihps{} according to an equational theory
        of \ihps{}
        (\cref{sub:equational-theory}).
        Compared to previous QHLs, the logic is non-branching, effective, and
        fully expressive (see
        \cref{sec:related-work}). It is proven
        sound with respect to the denotational
        semantics
        (\cref{thm:soundness})
        allowing
        the extraction of properties of the program via proof and symbolic
        execution rather than explicit computation on CQ states
        (\cref{cor:estimation-via-symbolic-execution}).
    \item Invariant-based loop analysis (the logic rule~\ref{rule:while}) and a
        heuristic for forming such invariants (\cref{thm:generic-strategy}),
        which is broadly applicable and illustrated through the running example
        of the repeated-until-success unitary synthesis
        (\cref{lst:rus-unitary}) and through which resource consumption
        properties can be inferred.
    \item An implementation $\libname$ (\cref{sec:implementation}) of an
        \ihps{}-based semi-automated Hoare logic engine for $\lang$. $\libname$ applies the rules of the logic automatically in a forward-directed manner, along the way raising proof obligations for \ihps{} equivalences, which can then be checked semantically. To the best of our knowledge, this is the
        first semi-automated tool for analyzing hybrid quantum programs and
        their resource consumption, with future extensions underway for full
        mechanization of the proofs of \ihps{} equivalences and invariant
        conservation.
    \item A catalog of case studies
        (\cref{sec:applications}) including a generic
        analysis of multiple instances of Repeat-Until-Success, an example that does not
        terminate with probability 1, and an example that is more convoluted than RUS containing a nested
        \texttt{while} loop illustrating the expressive power of the approach
        and its different features, and highlighting the use of the approach in
        resource estimation.
\end{enumerate}

The
details of the more technical definitions and constructions as well as the
proofs of the results are presented in{\ifCR{ an extended version of the
paper}\else{~\cref{app:exhaustive-formalism} and~\cref{sec:proofs},
respectively}\fi}. We will use the symbol $\appx$ to point to the \ifCR{extended
version of the paper}\else{appendix}\fi{} when relevant.

    \subsection{Bird's Eye View of the Approach}\label{sub:Overview}

Among the most ubiquitous patterns in hybrid programs is the
\emph{repeat-until-success} (RUS) pattern, where a program that
probabilistically produces a (checkable) desired outcome is repeated until such an
outcome is obtained.  Our approach handles such RUS patterns as well as more
convoluted loops where success probabilities may depend on
the quantum state in question (\cref{sec:weak-measurements}) and
nested loops where the probability of success for the outer loop depends on a
quantum state, which itself is dependent on the number of iterations before
halting in the inner loop
(\cref{sec:nested-while}). 

\begin{wrapfigure}[13]{r}{0.48\textwidth}
    \centering
    \vskip -0.5em
\begin{minipage}{0.4\textwidth}
    \centering

    \begin{lstlisting}[caption={The repeat-until-success program for unitary
    synthesis \labelthis{prog:rus-unitary}{\textsc{Synth}}},
    label={lst:rus-unitary}]
qubit q$_1$; bit c; int $\intVar$;
X(q$_1$); c := 1; 
do
    X(q$_1$); H(q$_1$); T(q$_1$);
    CNOT(q$_1$,q); H(q$_1$); CNOT(q$_1$,q);
    T(q$_1$); H(q$_1$);
    c := measure(q$_1$);
    x := x + 1;
while c
    \end{lstlisting}
\end{minipage}
\end{wrapfigure}
To illustrate our approach, consider a typical RUS program: the synthesis \ref{prog:rus-unitary}
of the unitary $U = 1/\sqrt 3 (I + i \sqrt 2 X)$ on a qubit $\qubit$ in
\cref{lst:rus-unitary}.  This implementation is one of many similar instances of
RUS synthesis in the literature~\cite{PS13} and was used as a prototypical
example in a related work on quantum expectation transformers~\cite{AMPPZ22}. In
\ref{prog:rus-unitary}, an ancilla qubit $\qubit_1$, a (classical) bit $\cbit$,
and a (classical) integer counter $\cInt$ are all initialized to $0$ (line 1),
then the qubit and bit are immediately set to 1 (line 2). Then,
 lines 4-7 apply the circuit in \cref{fig:rus-unitary-circuit} to $\qubit$
and $\qubit_1$, leaving the result of the measurement of $\qubit_1$ in $\cbit$.
Finally, the counter of iterations $\cInt$ is incremented (line 8). At this point, if
$\cbit$ is $0$, the unitary $U$ has successfully been applied to $\qubit$;
otherwise, $\qubit$ is left unchanged, and a repetition is needed (line 9).

\begin{figure}[ht]
    \centering
    \begin{tikzpicture}
        \node[inner sep = 0pt, scale=0.8] {
                \begin{quantikz}[wire types={q,q}]
                    \lstick{\ket[\qubit]{\psi}} &&&& \targ{} &&\targ{} & &&&
                    \rstick{$U^\boolVar \ket[\qubit]{\psi}$}\\
                    \lstick{\ket[\qubit_1]{1}} &  \gate{X} & \gate{H} & \gate{T} &
                    \ctrl{-1} & \gate{H} & \ctrl{-1} & \gate{T} & \gate{H} &
                    \meter{} & \cw \rstick{$\boolVar$}
            \end{quantikz}};
    \end{tikzpicture}
    \caption{The circuit for the RUS unitary synthesis of $U$ in
    \cref{lst:rus-unitary}.}\label{fig:rus-unitary-circuit}
\end{figure}
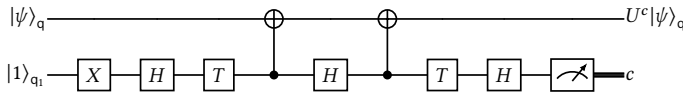

Our analysis of the program relies on \emph{integer hybrid path-sums} (\ihps)
(\cref{sec:hps}), extending the hybrid path-sums of~\cite{CIN+26} to support
unbounded \texttt{while} loops in symbolic analysis. In essence, an integer HPS is a symbolic tuple
{\small $\sum_{\vec \var}\pasum{p}{\ket[\qbit]{\boolTerm_1}\classKet[\cbit]{\boolTerm_2}}{n}$}
of expressions $p$, $n$, $\boolTerm_1$, and $\boolTerm_2$ over tuples $\vec \var = (\var_1,
\ldots, \var_k) $ of boolean or integer \emph{path variables}. This expression
describes a sum of \emph{paths} (complex-weighted basis states) of the form
$n(\vec v) \cdot e^{2\pi i \cdot p(\vec v)}\ket[\qbit]{\boolTerm_1(\vec v)}$ where 
$\vec v = (v_1, \ldots, v_n)$ ranges over the instantiations $v_i \in \N$ or $v_i \in
\B$ of integer or boolean path variables $\var_i$ respectively, and $t(\vec v)$
is the evaluation of a term $t$ with each $\var_i$ assigned to $v_i$.

\[
    \ket{\psi} = \sum_{\substack{1 \leq i \leq k \\ v_i \in
    \N \text{ or } v_i \in \B}} n(\vec{v}) e^{2 \pi i \cdot p(\vec{v})}
    \ket[\qbit]{\boolTerm_1(\vec{v})}
\]

In addition to the quantum states given by $\ket[\qbit]{\boolTerm_1(\vec v)}$, integer HPS,
as their name suggests, also handle classical data $\classKet[\cbit]{\boolTerm_2}$.  This allows
them to symbolically and compactly represent the branching structure of the
execution, both in terms of quantum superpositions and classical probabilistic
branching.  For instance, a qubit $\qbit$ in state $\ket +$ is described using
the \ihps{} {\small $ \sum_\boolVar \pasuminline{0}{{\ket \boolVar}_\qbit}{(1 /{\sqrt
2})}$} with the boolean variable $\boolVar$ encoding the quantum branching in the
superposition $\ket +$ of the basis states $\ket 0$ and $\ket 1$.  Measuring
qubit $\qbit$ (non-destructively) is then encoded as copying the symbolic
expression $\boolVar$ in ${\ket \boolVar}_\qbit$ into a classical bit $\cbit$ as
$\classKet[\cbit]{\boolVar}$ in the path-sum {\small $\sum_\boolVar
\pasuminline{0}{\ket[\qbit]{\boolVar} \classKet[\cbit]{\boolVar}}{(1/ {\sqrt
2})}$}.  The interpretation is then that each measurement outcome $\nu \in \{0,
1\}$ corresponds to a filtration of the sum $\sum_{\boolVar \in \{0, 1\},
\boolVar = \nu} \frac 1 {\sqrt 2} \ket[\qbit]{\boolVar}$ into a vector whose
squared norm is the probability of obtaining the outcome $\nu$ and which, when
re-normalized, is the resulting quantum state if $\nu$ was obtained.

We can perform symbolic execution of programs (\cref{sec:operational-semantics})
as transformations of integer HPS.{}  For instance, a quantum bit-flip
$X$ transforms $\boolTerm_1$ into $\boolTerm_1 \oplus 1$ by performing an XOR, while the
Hadamard $H$ introduces a new path variable $\boolVar$ to sum over corresponding
to the fork of the basis states into the superpositions $\ket +$ and $\ket -$:
\begin{align*}
    X &: \pasum{p}{\ket[\qbit]{\boolTerm_1}\classKet[\cbit]{\boolTerm_2}}{n} \mapsto
\pasum{p}{\ket[\qbit]{\boolTerm_1 \oplus 1}\classKet[\cbit]{\boolTerm_2}}{n}\\
    H &: \pasum{p}{\ket[\qbit]{\boolTerm_1}{\classKet[\cbit]{\boolTerm_2}}}{n} \mapsto
            \sum_\boolVar
            \pasum{p +
            \frac{{\boolTerm_1}\boolVar}{2}}{\ket[\qbit]{\boolVar}{\classKet[\cbit]{\boolTerm_2}}}{\frac n {\sqrt 2}}
\end{align*}

With such a symbolic description of all possible states of the computer at a
given point, many formal verification tasks (including functional correctness,
termination, and resource consumption) can be reformulated as equivalence checks
between integer HPS.{}

For the\ \ref{prog:rus-unitary} program, we can write a functional specification
as a loop invariant $\h_{\textsc{{RUS-inv}}}[\intVar]$ with a free
variable $\intVar$ describing the state of the system after $\intVar$ iterations.  This
symbolic representation encodes all possible branches of execution:
halting by iteration $\intVar$ or still needing to continue.
\[
    \h_{\textsc{RUS-inv}}[\intVar] \equaldef \bigpsadd_{\intVarOther=1}^{\intVar}
\sum_{\boolVar_1}
\pasum{\frac{\boolVar_1}{4}}{\ket[\qubit]{\boolVar\oplus
    \boolVar_1}\ket[\qubit_1]{0}\classKet[\cbit]{0}\classKet[\cInt]{\intVarOther}}{\frac{1}{2^{\intVarOther}}
\sqrt{2}^{\boolVar_1}} \psadd
\pasum{0}{|\boolVar\rangle_\qbit \ket[\qubit_1]{1}\classKet[\cbit]{1}
\classKet[\cInt]{\intVar}}{\frac 1
    {2^{\intVar}}}
\]

We can prove that it is correct by a heuristic based on the symbolic execution
of a generic iteration from ($\intVar$ to $\intVar+1$) (see
\cref{sec:generic-strategy}).  This functional specification can then be used to
study the limiting behavior of the program,
again symbolically, by taking the limit as $\intVar\to \infty$:
\[ \h_{\textsc{RUS-}\infty} \equaldef \bigpsadd_{\intVar > 0} \sum_{\boolVar_1}
    \pasum{\frac{\boolVar_1}{4}}{\ket[\qubit]{\boolVar\oplus
    \boolVar_1}\ket[\qubit_1]{0}\classKet[\cbit]{0}\classKet[\cInt]{\intVar}}{\frac{1}{2^{\intVar}}
\sqrt{2}^{\boolVar_1}} \]
In this form, we can extract the entire probability distribution of the
number of iterations before halting: $\probability(\cInt = 0) = 0$ and
$\probability(\cInt = \intVar) = \frac{3}{4^\intVar}$, from which we
can deduce both the almost-sure termination (AST) of the program, and the
expected number of repeats before success, which is $\mathbb{E}[\cInt] =
\sum_{\intVar > 0} \intVar \cdot \frac{3}{4^\intVar} = \frac 4 3$.  This
corresponds to average runtime. The type of analysis performed here can be
generalized to other types of resources by choosing different ways to increment
the variable $\cInt$.  For example, calculating the average number of gates of a
certain type (e.g., expensive $T$~\cite{V25,MSCRdM19} or multi-qubit
gates~\cite{L18}) by choosing to increment the counter according to
the type of resource considered.

In our implementation $\libname$, the \ref{prog:rus-unitary} program is written
with the \texttt{while} loop annotated with the loop invariant hints
$\h_{\textsc{{RUS-inv}}}[\intVar]$. Once the invariant is given, $\libname$ is then
capable of automatically computing the final \ihps{} $\h_{\textsc{RUS-}\infty}$,
producing proof obligations for \ihps{} equivalences, and checking said
equivalences semantically, granting high confidence in the validity of the
equivalences to be proven.

    \subsection{Related Work}%
\label{sec:related-work}

When it comes to general-purpose verification, our work fits into a broader line
of research on quantum Hoare logics (QHL), including the QHL of Mingsheng Ying~\cite{Y12},
later extended to hybrid programs storing classical variables~\cite{FY21}. There
has also been work on relational QHL~\cite{U19,BGWZ25}. Our work also inherits
from a long line of research on path-sums~\cite{A18,A23,CIN+26}, notably
including \HQbricks{}~\cite{CIN+26} on which it is based.
A few differences appear in our work, however. 

First, we choose a more
expressive \emph{programming language} than the state-of-the-art. In particular,
Ying's \texttt{while} language~\cite{Y12}, on which most QHLs are based, does
not support separate measurement and control, while its extension~\cite{FY21} to
support that still lacks integer computations, largely recognized as an
essential feature in classical computing. This, among others, makes it difficult
to express generic \texttt{for} loops, ubiquitous in quantum algorithms.
On the other hand, \HQbricks{}~\cite{CIN+26} misses both integers and unbounded
loops. Our language \lang{} subsumes these previous languages in expressivity.

Second, we have a more expressive \emph{specification language} representing the
entire ensemble of quantum states, rather than predicates giving only partial
information about the state. This is necessary for resource analysis. Our logic
is also non-branching, as all branches are encoded in a
single symbolic specification at a time. 
Most importantly, as recently noted~\cite{YPR26}, the current QHLs are not
scalable given their requirement for matrix computations. In contrast, choosing
\ihps{} as a specification language gives us effective equational reasoning
tools inherited from path-sums which allowed us to implement the logic into a
semi-automated tool.

Meanwhile, when it comes to resource analysis, and in terms of applications,
some work has provided concrete tools for estimating the resources of quantum
programs, or more precisely, circuit description languages.  In this regard, we
note the development of Microsoft's Azure Quantum Resource
Estimator~\cite{B+22,Azure24} developed especially in the lens of benchmarking:
to understand precisely when the quantum advantage has been achieved and to
verify claims about this advantage.  In the same vein, Colledan et
al.~\cite{CDL24,CDL25} have developed an approach based on dependent types that
allows inferring bounds on the size and depth of a quantum circuit generated by
a program from the proto-Quipper family~\cite{FKRS20}.  However, the control
flow does not depend on measurement outcomes, and they do not support general
recursion. 

On more theoretical aspects, some work has sought, in the spirit of implicit
complexity~\cite{P20}, to characterize complexity classes such as quantum
polynomial time on high-level languages~\cite{Y20,HPS23} and to certify
compilation to circuits of polynomial size~\cite{HPS25}.  However, these
languages are inherently concerned with asymptotic bounds instead of actual
costs. Furthermore, they have a limited hybridity, with no support for  general
recursion.

A closely related approach to our work uses quantum expectation
transformers~\cite{AMPPZ22,LZBY22,KKM+16}.  This technique also makes it
possible to study the properties of resources such as almost-sure
termination or the average cost of quantum programs with general recursion by
computing symbolic weakest pre-expectation.  However, the current
literature still does not address the problem of efficiently representing
the programs or the expectation transformers; an issue which can be resolved
with integer HPS.\ Indeed, integer HPS generally avoid
exponential explosions of representations by moving the exponential branching
into the assignment of formal variables within a symbolic sum. 
Note also that most problems concerning the study of quantum program resources
are highly undecidable~\cite{AMPP24}.  There is therefore a need to sacrifice
completeness in order to enable automation through the development of
heuristics. From this perspective, an additional novelty of our approach is to
offer this kind of reasoning and analysis tool, in particular, in the concern of
loop termination and invariant analysis.

Finally, another line of research addresses the verification of parametrized quantum programs using proof assistants and automated solvers on
parameterized programs~\cite{ZBSLY23,xu2024automating,hietala2020proving,li2024qafny,cheng2025embedding}. 

In the \emph{classical} probabilistic
setting, similar results and tools for runtime estimation are under active
development~\cite{BCJ+23,NYC+25}. However, as these techniques do not take into
account quantum aspects, they are inspiring but cannot be directly applied.

\section{Hybrid Quantum Programs with Iteration}%
\label{sec:language}
    \subsection{Syntax}%
\label{sec:language-syntax}

We present $\lang$, an imperative hybrid (quantum and classical) programming
language whose syntax is given in \cref{fig:language-syntax}. $\Lang$ includes
primitive expressions, denoted $\programTerm$, which can be booleans $\bool$ or
integers $\integer$. The variables include qubits $\qubit$, bits $\cbit$, and
integers $\cInt$ whose
identifiers (hereafter \emph{addresses}) are drawn from the disjoint sets
$\qbitSet$, $\cbitSet$, and $\cIntSet$ respectively.  We write
$\classAddressSet \equaldef  \cbitSet \cup \cIntSet $ for the classical
addresses and $\addressSet \equaldef \qbitSet \cup \classAddressSet$ for the set
of all the addresses, which we assume to be ordered.  In terms of operations,
$\lang$ supports \emph{initializations}
$\qubitInitProg{\qubit}, \cbitInitProg{\cbit}, \intInitProg{\cInt}$, multi-qubit
\emph{unitary applications} $\unitary(\qubit, \ldots, \qubit)$, \emph{classical
assignments} $\assign{\cInt}{\integer}$ and
$\assign{\cbit}{\bool}$, \emph{measurements} $\measure{\cbit}{\qubit}$ in the
computational basis, as well as classically-controlled conditioning
\emph{if-then-else} and \emph{while} loops. Note that we also admit the
syntactic sugar $\dowhile{\prog}{\bool}$ for the sequence $\prog; \while{\bool}{\prog}$.

\begin{figure}[htpb] 
    \centering
    $\begin{array}{rcl@{\hspace{1em}}l}
                \integer &::=& k \mid \cInt \mid \integer + \integer \mid \integer *
                \integer \mid \integer ^ \integer & \emph{Integer expressions} \\
                \bool &::=& \trueBool \mid \falseBool \mid \cbit \mid \integer
                \leq \integer \mid \integer = \integer \mid \bool \land \bool
                \mid \lnot \bool & \emph{Boolean expressions}
                \\
                \prog &::=& \skipProg \mid \qubitInitProg{\qubit} \mid
                \cbitInitProg{\cbit} \mid \intInitProg{\cInt}  & \emph{Programs}
                \\
                && \mid \unitary(\qubit,\ldots,\qubit) \mid
                \assign{\cInt}{\integer} \mid \assign{\cbit}{\bool} \mid
                \measure{\cbit}{\qubit} \\
                && \mid \sequenceProg{\prog}{\prog} \mid
                \ifthene{\bool}{\prog}{\prog} \mid \while{\bool}{\prog}
            \end{array}$
            \caption{Syntax of $\lang$}%
            \label{fig:language-syntax}
\end{figure}

$\lang$ is flexible on the choice of the supported unitaries, but in this
article, we fix\footnote{Other unitaries can still be treated as black boxes;
see \cref{sec:rus-bernoulli}.} $\unitary \in \{\texttt{CNOT}, \texttt{X},
\texttt{Z}, \texttt{H}\} \cup \{\texttt{R}_k =
\texttt{R}_{\texttt{Z}}(\frac{2\pi}{2^k}) \mid k \in \N\}$ which lends itself
well to symbolic representations in terms of path-sums (\cref{sec:hps}) while
remaining pseudo-universal; i.e., capable of approximating any unitary with
arbitrary precision~\cite{DBE95}.

\paragraph{Well-formedness.}%
\label{sec:language-well-formedness}

The programs of $\lang$ are subject to constraints that ensure their physicality, 
including \emph{memory constraints}: no access to unallocated
memory or double allocation, and \emph{unitarity constraints}: qubits
may not be used more than once in the same unitary application.

\ifmargins%
\marginpar{$\signature$, $\validProg{\signature}{\prog}{\signature'}$,
    $\programSet_\signature$, $\applySig{\prog}{\signature}$, }
\fi

A \emph{signature} $\signature$ is a finite subset of addresses, i.e.,
$\signature \finsubseteq \addressSet$.  The validity of a program, according to
the \emph{signature} $\signature$ describing addresses allocated, can be checked
statically by the judgment  $\validProg{\signature_1}{\prog}{\signature_2}$ read
as ``$\prog$ is valid on states with signature $\signature_1$ and transforms
them into states with signature $\signature_2$'' \appx. It
checks that $\prog$ allocates memory homogeneously in each branch of an
\texttt{if-then-else}, does not allocate any memory in a \texttt{while}, and
that $\qubit_1, \ldots, \qubit_n$ are distinct in any application of
$\unitary(\qubit_1, \ldots, \qubit_n)$.  When
$\validProg{\signature_1}{\prog}{\signature_2}$ holds, the signature
$\signature_2$ is unique. Hence, we write $\applySig{\prog}{\signature_1}
\equaldef \signature_2$. We also define sets of programs valid on $\signature$
by $\programSet_\signature \equaldef \{\prog \mid \exists \signature'
\finsubseteq \addressSet, \validProg{\signature}{\prog}{\signature'}\}$, and the
set of all valid programs by $\programSet \equaldef \bigcup_{\signature
\finsubseteq \addressSet} \programSet_\signature$. We use the notation
$\signature(\programTerm)$ for the set of addresses occurring in $\programTerm$.

\vskip 1em
\noindent \begin{minipage}{0.55\textwidth}
    \begin{example}\label{running:validity}
        Along this article, we will use a simple program,
        \ref{cointoss} in \cref{lst:coin-toss}. It tosses a quantum
        coin until it lands on heads, counting the number of tosses in $\cInt$ to
        illustrate atomic concepts and definitions. For this program, we have
        $\validProg{\emptyset}{\ref{cointoss}}{\signature_{\textsc{CT}}}$ with
    $\signature_{\textsc{CT}} \equaldef \{\qbit, \cbit, \cInt\}$.  \end{example}
\end{minipage}
\hfill
\begin{minipage}{0.43\textwidth}
	\centering
    \vspace{-1em}
	\begin{minipage}{0.7\textwidth}
		\begin{lstlisting}[caption={The \labelthis{cointoss}{\textsc{CoinToss}} program}, label={lst:coin-toss}]
qubit $\qbit$; bit $\cbit$; int x;
do  H($\qbit$); $\cbit$ := measure $\qbit$;
    x := x + 1
while ($\lnot \cbit$)
        \end{lstlisting}
	\end{minipage}
	\vspace{-1.5em}
\end{minipage}

    \subsection{Denotational Semantics}%
\label{sec:language-semantics}

In the quantum computing literature, the semantics of quantum programs is often
given in terms of transformations of density operators over some underlying
Hilbert space, e.g.,~\cite{Y24}. In this section, we present such semantics as
the denotational semantics for \lang{} and discuss the limitations of using
density operator semantics in the context of formal analysis. 

For our work, the Hilbert spaces depend on signatures $\signature$.

\begin{definition}[State space ${\hilbert[\signature]}$]%
    \label{def:basis-signature}
    \ifmargins%
        \marginpar{$\basis(\signature)$, $\hilbert[\signature]$}
    \fi
    Given a signature $\signature$, let its \emph{basis} and \emph{state space} be defined by:  
\[
    \basis(\signature) \equaldef \B^{\signature \cap
    (\qbitSet \cup \cbitSet)} \times \Z^{\signature \cap \cIntSet} \quad
    \text{and} \quad
    \hilbert[\signature] \equaldef {\hilbert_2}^{\otimes \left(\signature \cap
    (\qbitSet \cup \cbitSet)\right)} \otimes {\left(\ell^2(\Z)\right)}^{\otimes \signature \cap
    \cIntSet},
\]
where $\B \equaldef \{0,1\}$, $\hilbert_2 \equaldef \C^2=\mathrm{span}(\ket 0,
\ket 1)$ is the 1-qubit Hilbert space, $X^{\otimes Y} \equaldef
\bigotimes_{\address \in Y}{X_\address}$ is the tensor product of $X$ labeled
by elements of $Y$, and $\ell^2(\Z)$ is the Hilbert space over $\Z$ given
by $\textstyle \ell^2(\Z) \equaldef \left\{\alpha \in \C^\Z \mid \sum_{i \in
\Z} |\alpha(i)|^2 < \infty\right\}$.

For disjoint $Y_1, Y_2 \subseteq \addressSet$, we interpret $X_1^{\otimes Y_1}
\otimes X_2^{\otimes Y_2}$ the same as $X_2^{\otimes Y_2} \otimes X_1^{\otimes
Y_1}$ as $\bigotimes_{\address \in Y_1 \cup Y_2} Z_\address$ where $Z_\address =
X_1$ if $\address \in Y_1$ and $Z_\address = X_2$ if $\address \in Y_2$.
\end{definition}

Since the states include classical parts, they are described by density
operators over $\hilbert[\signature]$ called CQ states~\cite{W13} where the
classical data are encoded as quantum data that are dephased (measured).

\begin{definition}[CQ state]%
    \label{def:cq-state}
    \ifmargins%
        \marginpar{$\mathrm{CQ}(\signature)$, $\signature(\rho)$}
    \fi
    A \emph{CQ state} $\rho$ of signature $\signature$ is a linear operator over
    the space $\hilbert[\signature]$ which is self-adjoint $(\rho =
    \rho^\dagger)$, positive semi-definite $(\forall \ket{\psi} \in
    \hilbert[\signature], \bra{\psi} \rho\, \ket{\psi} \geq 0)$, bounded
    $(\mathrm{tr}(\rho) \leq 1)$, and classically dephased:
            \[
                \rho = \left(\dephase_{\signature \cap \classAddressSet} \otimes
                I_{\signature \cap \qbitSet} \right)(\rho) \quad \text{where}
                \quad
                \dephase_{\signature \cap \classAddressSet}(\rho) \equaldef
            \sum_{\classicalState \in \basis(\signature \cap \classAddressSet)} |\classicalState\rangle \langle \classicalState| \rho
                |\classicalState\rangle \langle \classicalState|
            \]
    with $I_{\signature \cap \qbitSet} \text{ the identity on }
    \hilbert(\signature \cap \qbitSet)$ and $D_{\signature \cap
    \classAddressSet}$ the \emph{dephasing channel} on
    $\hilbert(\signature\cap\classAddressSet)$.  The space of CQ states
    over $\hilbert[\signature]$ is denoted by $\mathrm{CQ}(\signature)$. If a
    CQ state $\rho$ is given with an implicit signature, $\signature(\rho)$
    denotes this signature. Similarly, $\signature(\superop)$ denotes
    the signature of a superoperator $\superop$.
\end{definition}

\subsubsection*{Preliminary constructions.}
Given an assignment $\sigma \in
\basis(\signature)$ with $\signature \supseteq \signature(\programTerm)$, the
evaluation of $\programTerm$ in $\sigma$ is denoted by
$\sem{\programTerm}_\sigma$ while $\sigma[\address \mapsto \programTerm]$ is the
assignment where $\address$ is reassigned to $\sem{\programTerm}_\sigma$.
With that, we define \emph{projections} $P_{\assign{\address}{\programTerm}}$, \emph{filters} $F_{\bool}(\rho)$, and
\emph{extensions} $\widehat {\superop}^\signature$ as follows:

\[
    \displaystyle P_{\assign{\address}{\programTerm}} (\rho) \equaldef
        \sum_{\classicalState \in \basis(\signature(\programTerm) \cup \{
        \address\})}
        |\classicalState[\address \mapsto \programTerm]\rangle \langle \classicalState|\,  \rho \,
        |\classicalState\rangle \langle \classicalState[\address \mapsto
    \programTerm]| \]\[
    \displaystyle F_{\bool}(\rho) \equaldef \sum_{\substack{\classicalState \in
                \basis(\signature(\bool))\\
        \sem{\bool}_\classicalState = 1}}
        | \classicalState \rangle \langle \classicalState |\, \rho\, | \classicalState \rangle \langle
        \classicalState | \quad \text{and} \quad
            \displaystyle \widehat {\superop}^\signature \equaldef I_{\hilbert(\signature \setminus
        \signature(\superop))} \otimes \superop
\]

$P_{\assign{\address}{\programTerm}}$ projects basis vectors onto other basis
vectors,
thus performing an assignment operation $\assign{\address}{\programTerm}$;
$F_\bool$ filters a CQ state $\rho$ leaving only those states that satisfy $\bool$; and
$\widehat {\superop}^\signature$ is the extension of the superoperator
$\superop$ defined over the space $\hilbert[\signature (\superop)]$ to the
space $\hilbert[\signature]$. We will usually leave $\signature$ implicit when
unambiguous. Finally, $\mathcal C_U \equaldef \rho \mapsto U \rho U^\dagger$ is the
superoperator corresponding to the unitary $U$.

\vskip 0.5em
With these notations, the denotational semantics of $\lang$ is defined
in~\cref{fig:standard-semantics} as a map 
\[
    \sem{\cdot} :
    \bigcup_{\signature \finsubseteq \addressSet}\ \bigcup_{\prog \in
    \programSet_{\signature}} \mathrm{CQ}(\signature) \to
    \mathrm{CQ}(\applySig{\prog}{\signature}),
\]

\begin{figure}[h]
    \[
        \begin{array}{rl@{\hspace{1em}}rl}
            \sem{\skipProg} \equaldef& id &
              \sem{\qubitInitProg{\qbit}} \equaldef&
             \rho \mapsto  \rho \otimes \ket[\qbit]{0}\! \bra{0}_\qubit \\[0.5em]
            \sem{\cbitInitProg{\cbit}} \equaldef&
            \rho \mapsto \rho \otimes \ket[\cbit]{0}\! \bra{0}_\cbit &
            \sem{\intInitProg{\cInt}} \equaldef&
            \rho \mapsto \rho \otimes \ket[\cInt]{0}\! \bra{{0}}_\cInt \\[0.5em]
            \sem{\unitary(\bar \qbit)} \equaldef& 
            \widehat{\mathcal C_{U}} &
             \sem{\assign{\cInt}{\integer}} \equaldef&
            \widehat{P_{\assign{\cInt}{\integer}}} \\[0.5em]
            \sem{\assign{\cbit}{\bool}} \equaldef&
            \widehat{P_{\assign{\cbit}{\bool}}} &
           \sem{\measure{\cbit}{\qubit}} \equaldef& 
            \widehat{P_{\assign{\cbit}{\qubit}}} 
        \end{array}
    \]
   \begin{align*}
            \sem{\prog_1;\prog_2} &\equaldef \sem{\prog_2} \circ \sem{\prog_1}
            \\
             \sem{\ifthene{\bool}{\prog_1}{\prog_2}} & \equaldef 
                        \rho \mapsto \sem{\prog_1}\left(\widehat{F_{\bool}}(\rho)\right) + 
                \sem{\prog_2}\left(\widehat{F_{\lnot
                \bool}}(\rho)\right)\\
            \sem{\while{\bool}{\prog}} &\equaldef
                \rho \mapsto \lim_{n \to \infty} \widehat{F_{\lnot \bool}}\left(
                \sem{\ifthene{\bool}{\prog}{\skipProg}} ^n(\rho)\right)
     \end{align*}
    \caption{Denotational Semantics of Hybrid Quantum Programs}%
    \label{fig:standard-semantics}
\end{figure}

\begin{example}%
    \label{running:semantics}
    For the \ref{cointoss}, the state space is
    $\hilbert[\signature_{\textsc{CT}}] = \hilbert(\{\qbit, \cbit, \cInt\}) =
    (\hilbert_{2})_{\qbit} \otimes (\hilbert_{2})_{\cbit} \otimes
    {\ell^2(\Z)}_\cInt$. An example CQ state over
    $\hilbert[\signature_{\textsc{CT}}]$ is the result of 
    applying\ \ref{cointoss} to an empty input state
    $I_\emptyset = (1)\in \hilbert[\emptyset]$, namely, the mixed state
    consisting of all possible integer outcomes $i \in \N$ with probability
    $\frac 1 {2^{i+1}}$; that is, $\sem{\ref{cointoss}}(I_{\emptyset}) = \sum_{i
    \in \N} \frac 1 {2^{i+1}} |i\rangle_\cInt \langle i| \otimes |1\rangle_\cbit
    \langle 1| \otimes |1\rangle_\qbit \langle 1| \in
    \mathrm{CQ}(\signature_{\textsc{CT}})$.
\end{example}

\begin{restatable}[Well-definedness of the denotational semantics]{proposition}{restateWelldefinednessDenotational}%
    \label{lem:semantics-well-defined}
    For all $\signature \finsubseteq \addressSet$ and program $\prog \in
    \programSet_{\signature}$, the following holds:
    \[
        \forall\rho
        \in \mathrm{CQ}(\signature),\; \sem{\prog}(\rho) \in
        \mathrm{CQ}(\applySig{\prog}{\signature}).
    \]
\end{restatable}

For the \texttt{while} loop, the existence of the limit without filtering
($F_{\lnot \bool}$) is not guaranteed (e.g., for
$\while{\trueBool}{\assign{\cbit}{\lnot \cbit}}$). The (relatively standard)
approach is to consider the sequence of only the terminating branches ($F_{\lnot
\bool}$). The resulting sequence $\widehat{F_{\lnot \bool}}(
\sembreak{\ifthene{\bool}{\prog}{\skipProg}} ^n(\rho))$ is increasing and
bounded in the trace norm and therefore does \emph{always} converge.  Moreover, the trace of
$\sem{\prog}(\rho)$ is the termination probability of $\prog$ with
$\sem{\prog}(\rho)$ itself being the terminating state. This allows 
reasoning about programs that are not almost surely terminating.

\paragraph{Discussion and limitations.} This denotational semantics gives
meaning to programs, serving as a foundation to verify the soundness of the
logic, but its immediate use for static analysis is not suitable for two main
reasons.

First, in the bounded case, the size of the density operators grows
exponentially with the number of qubits, making the analysis
unfeasible/intractable; in other words, at least as difficult as strong
simulation~\cite{N08}. Secondly, once the unbounded case is considered, it is
generally undecidable to compute the limits, especially over non-trivial spaces
like $\mathrm{CQ}(\signature)$.

As such, there is a need for compact, tractable representations of quantum
states and their transformations that are amenable to static analysis.  In fact,
it is fair to say that the field of static analysis of quantum programs is
almost entirely about the search for the correct, compact abstractions away from
density operators (symbolic execution~\cite{CIN+26,C+21,BLS23},
automata~\cite{ACY+26,ACC+25}, Hoare logics~\cite{FY21}, etc.). So far, however,
it remains impossible to obtain both the level of detail (full state
description) expressible by the support for symbolic approaches and the
unbounded recursion and hybrid features achievable by Hoare logics.  Towards
that end, we introduce our approach based on \emph{hybrid path-sums} (\ihps{}),
following prior hybrid work in \HQbricks~\cite{CIN+26}, whereby states of a
hybrid computer are represented symbolically and evolved by symbolic execution,
allowing us to extract useful information about programs without resorting to
the full and generally infeasible simulation.

\section{Integer HPS}%
\label{sec:hps}

    Our key component to (classically) represent the state
    of a hybrid computer for the purpose of static analysis is our novel 
    \emph{integer HPS} (\ihps{}) representation. Integer HPS are compact
    symbolic representations of a CQ state with a structure appropriate for the
    analysis of hybrid systems. This representation extends the hybrid
    path-sums (\hps{}) introduced in \HQbricks~\cite{CIN+26} to support unbounded \texttt{while}
    loops.\ \hps{} are also themselves extensions of \emph{path-sums}~\cite{A18},
    a discrete version of Feynman's path-integrals~\cite{F48}, to allow for
    hybrid quantum-classical states.  

    \subsection{Basic Definitions and Notation}%
\label{sec:basic-definitions}

In this section, we describe formally the construction of \ihps{}.  First,
let $\boolVarSet$ and $\intVarSet$ be sets of boolean and integer variables,
respectively (denoted $\boolVar$ and $\intVar$), and let $\varSet \equaldef
\boolVarSet \cup \intVarSet$ be the set of \emph{path-variables}.  We emphasize
that $\varSet$ (for \ihps{}) should not be confused with the set of addresses
$\addressSet$, typeset in typewriter font (for programs). Note that, unlike for
addresses, we do not distinguish quantum or classical boolean \emph{variables}.
We note that, besides our proper extensions, some of the more basic \HQbricks{}
constructions are also treated differently; the differences are described in the
comparison paragraph at the end of this section.

To symbolically represent states in the spaces $\hilbert[\signature]$ as the sum
$\sum_{\vec a} \pasum{p}{\ket[\qubit]{\boolTerm_1}\classKet[\cbit]{\boolTerm_2}}{n}$ from
\cref{sub:Overview}, we need
to introduce symbolic terms representing both the basis vectors including
the quantum $\ket[\qubit]{\boolTerm_1}\in \memoryType$ and classical
$\classKet[\cbit]{\boolTerm_2} \in \memoryType$ parts, as well as the terms $p
\in \phaseType$ and $n \in \normType$ used to represent the phase and
normalization factors, respectively, in $n e^{2\pi i p}
\ket[\qubit]{\boolTerm_1}\classKet[\cbit]{\boolTerm_2}$. 

\begin{definition}[\ihps{} components]%
    \label{def:components}
    The \emph{integer}, \emph{boolean}, \emph{phase}, \emph{norm}, and \emph{memory} expressions are defined as follows:
    \[
        \begin{array}{r l l l}
          \intTerm & ::= k \mid \intVar \mid \lift{\boolTerm} \mid \intTerm + \intTerm
    \mid \intTerm \cdot \intTerm \mid \intTerm^\intTerm &   \hspace{1em} &
     \intType \\
    \boolTerm & ::= 0 \mid 1 \mid \boolVar \mid (\intTerm = \intTerm) \mid (\intTerm \leq
    \intTerm) \mid \boolTerm \cdot \boolTerm \mid \boolTerm \oplus
    \boolTerm &  \hspace{1em} & \boolType \\
            p & ::= \boolTerm/2^\intTerm \mid p + p \mid \intTerm \cdot p &
            \hspace{1em} & \phaseType \\
            n & ::= \frac \intTerm {\sqrt{2^\intTerm}} \mid \cos(2\pi p) \mid \sin(2\pi
            p) & \hspace{1em} & \normType \\
               & \phantom{::=} \mid n \cdot n \mid n + n \mid n / n \mid
               \sqrt n \\
        m & ::= \emptyset \mid \ket[\qubit]{\boolTerm} \mid
        {\classKet[\cbit]{\boolTerm}} \mid {\pastKet[\B]{\boolTerm}} \mid
        {\classKet[\cInt]{\intTerm}} \mid
        {\pastKet[\Z]{\intTerm}}
        \mid m \otimes m
          & \hspace{1em} & \memoryType
        \end{array}
    \]
    
    \noindent where $k \in \Z$, $\lift{\boolTerm}$ is the casting of a boolean
    expression in $\boolType$ into an integer  in $\intType$, and $\cbit \in
    \cbitSet$, $\qbit \in \qbitSet$, as well as $\cInt \in \cIntSet$ are
    addresses. We also admit typical definable integer and boolean expressions
    such as $\lnot b \equaldef 1 \oplus b$ or  $1 \leq x < i \equaldef (1 \leq
    x) \cdot (x \leq i) \cdot (1 \oplus (x = i))$ as syntactic sugar. 
\end{definition}

The constructors of terms in $\phaseType$ and $\normType$ allow both exact representations (compared
to floating-point numbers) and expressivity against the pseudo-universal
gate-set (Clifford+$R_k$).  In fact, the $R_k$ gate produces \emph{dyadic}
phases of the form $e^{2\pi i \frac{1}{2^k}}$, and summing such numbers produces
\emph{constructible} norm terms of the form $\sin( 2\pi p)$ and $\cos(2\pi p)$
(which can in turn be rewritten in terms of $+$, $-$, $\cdot$, $/$, and
$\sqrt{-}$ using half-angle formulae).

The \memoryType{} is used to represent a basis vector of the Hilbert space where
$\ket[\qbit]{\boolTerm}$, $\classKet[\cbit]{\boolTerm}$, and
$\classKet[\cInt]{\intTerm}$ represent the states of the qubit $\qubit$,  bit
$\cbit$, and integer $\cInt$ respectively to be understood as symbolic
representations of basis vectors in $\hilbert[\signature]$. These do not
suffice, though, as if we measure qubit $\measure{\cbit}{\qbit}$, then reset
$\assign{\cbit}{0}$, we lose the information ``$\qbit$ was measured/projected''.
To remedy that, we keep a log of past expressions in the form of `past'
classical bits ${\pastKet[\B] \boolTerm}$ and integers $\pastKet[\Z]{\intTerm}$,
essential for the correct symbolic representation, but which do not correspond
to any variable that is accessible to the program.  Memory terms $m_1$ and $m_2$
can be combined into $m_1 m_2 \equaldef m_1 \otimes m_2$, as long as they do not
share addresses. The signature $\signature(m)$ of a memory $m$ is then the set
of (present) addresses of $\addressSet$ appearing in $m$, with past subterms
$\pastKet[\B]{b}$ and $\pastKet[\Z]{i}$ not being excluded (i.e.,
$\signature(\pastKet[\B]{b}) = \signature(\pastKet[\Z]{i}) =\emptyset$).
With these definitions in place, we can now define integer hybrid path-susm (\ihps{}).

\begin{definition}[Integer hybrid path-sums]%
    \label{def:hps}
    \emph{Integer hybrid path-sums} (\ihps{}) are terms defined by:
    \[ \begin{array}{r l l l}
        \h & ::= \pasum{p}{m}{n} \mid \h \psadd \h \mid \h
        \otimes \h \mid \h \sqpsadd \h \mid \bigpsadd_{\var} \h \mid \bigotimes_{\var} \h \mid
        \lim_{\intVar} \h & \hspace{1em}  & \hpsType
    \end{array}\]
    with $p \in \phaseType$, $n \in \normType$, $m\in \memoryType$, $\var \in \varSet$, and $\intVar \in \intVarSet$.  
\end{definition}

The constructors of an \ihps{} correspond to the linear algebra operations needed in
the symbolic execution and analysis of a program. First, with $j$ being the
imaginary unit, the triplet $\pasum{p}{m}{n}$ describes the vector $n e^{2\pi j
p} \cdot m$ which we call a \emph{path}.  Next, $+$ and $\otimes$ correspond to
the usual addition and tensor product of vectors, while $\sqpsadd$ is a form of
`direct sum' which maps two vectors to orthogonal subspaces before adding
them. Finally, the variable binders $\bigpsadd_{\var}$, $\bigotimes_{\var}$, and
$\lim_{\intVar}$ are used to bind variables to perform sums, tensor products,
and limits over the formal variables $\var$ and $\intVar$. Here, $\lim_\intVar$
is a purely syntactic construct; see
\cref{sec:interpretation} for its
interpretation and for questions of convergence. 
Variables not \emph{bound} by $\bigpsadd$, $\bigotimes$, or $\lim$ are
\emph{free} in a term.  An \ihps{} $\h$ is \emph{closed} if it does not have
free variables.  We can define the partial map $\signature(-)$ for the signature
of an \ihps{} below.\footnote{$\signature(\h) = \emptyset$ in $\bigotimes$
    implies purely past memories $\pastKet[\B]{-}$ and $\pastKet[\Z]{-}$, not
empty memories.} When the conditions on the right are not met, $\signature$ is
undefined. The image of $\signature$ consists of all the well-formed \ihps{}.
\[
\begin{array}{rl@{\quad}l}
    \signature(\pasum{p}{m}{n}) &\equaldef \signature(m) & \\
    \signature(\h_1 + \h_2) & \equaldef
    \signature(\h_1) & \text{if } \signature(\h_1) = \signature(\h_2) \\
    \signature(\h_1 \otimes \h_2) &\equaldef \signature(\h_1) \cup
    \signature(\h_2) & \text{if } \signature(\h_1) \cap \signature(\h_2) =
    \emptyset \\
    \signature(\h_1 \oplus \h_2) &\equaldef \signature(\h_1) & \text{if } \signature(\h_1) = \signature(\h_2) \\
    \signature(\sum_{\var} \h) &\equaldef
        \signature(\h) \\
    \signature(\bigotimes_{\var} \h) &\equaldef \emptyset & \text{if }
    \signature(\h) = \emptyset\\
    \signature(\lim_\intVar \h) &\equaldef
        \signature(\h) 
\end{array}
\]

\begin{example}\label{ex:ihps}
    In \ref{cointoss}, a single round of a Hadamard followed by a measurement
    produces the \ihps{} state $\h_1 \equaldef \sum_{c}
    \pasuminline{0}{\ket[\qubit]{\boolVar} \classKet[\cbit]{\boolVar}}{1/
    {\sqrt 2}}$. Meanwhile, the limiting state of\ \ref{cointoss} is given by an
    \ihps{} including an infinite sum over the integer variable $\intVar$:
    \ihps{} $\h_{\textsc{CT-}\infty} \equaldef \bigpsadd_{\intVar}
    \pasuminline{0}{\ket[\qbit]{0}\classKet[\cbit]{0} \classKet[\cInt]{\intVar}}
    {1 / {\sqrt {2^{{\intVar}+1}}}}$. The derivation of $\h_1$ and
    $\h_{\textsc{CT-}\infty}$ will be given in more detail in
    \cref{tab:hoare-logic-derivation}.
\end{example}

\paragraph{Comparison with hybrid path-sums from \HQbricks{}~\cite{CIN+26}.}%
\ihps{} are an extension of the \hps{} from \HQbricks{} for use in the analysis of
unbounded loops. 
\ihps{} support integer variables and terms, infinite domain binders (i.e.,
$\lim_{\intVar}$, $\bigpsadd_{\intVar}$, and $\bigotimes_{\intVar}$), and the
direct sum $\oplus$, while \hps{} do not. An important implication of lacking
infinite domains is that \HQbricks{} can
afford to use only one constructor $\pasum{p}{m}{n}_{\{\boolVar_1, \ldots,
\boolVar_n\}}$ for path-sums representing the explicit sum $\sum_{\boolVar_1,
\ldots, \boolVar_n} \pasum{p}{m}{n}$, in \ihps{} and define $+$ and
$\otimes$ as syntactic sugar. Meanwhile, in \ihps{}, infinite domains
blocks the reduction of both the binary and the variable binder constructors to
syntactic sugar, hence the need for their explicit introduction in \ihps{}.
Finally, we note a couple of notational differences: first, \hps{} uses qubit
arrays, while \ihps{} names every qubit separately for simplicity; second, in
\hps{} a (single) classical bit stores a list of boolean expressions
$\classKet[\cbit]{\boolTerm_1 : \cdots : \boolTerm_n}$ which includes the present
value $\boolTerm_1$ and the log of the history $\boolTerm_2 : \cdots :
\boolTerm_n$. In \ihps{}, this connection between a bit and its history is
relaxed allowing more flexibility in the use of the past values; instead, the
\hps{} expression is interpreted as: $\classKet[\cbit]{\boolTerm_1} \otimes
\pastKet[\B]{\boolTerm_2} \otimes \cdots \otimes \pastKet[\B]{\boolTerm_n}$. For
instance, the result of $H(\qbit); \measure{\cbit}{\qbit}; \assign{\cbit}{0}$
which would be represented in \hps{} as $\pasum{0}{\ket[\qbit]{y}
\classKet[\cbit]{0 : y}}{1}$ is represented in \ihps{} as $\sum_{\boolVar}
\pasum{0}{\ket[\qbit]{\boolVar} \classKet[\cbit]{0} \pastKet[\B]{\boolVar}}{1}$.

    \subsection{Interpretations of \ihps{}}%
\label{sec:interpretation}

Under (free) variable assignment environment, the expressions defined in
\cref{def:components,def:hps} admit interpretations into concrete domains,
denoted $\semFock{T}$ for each type $T$, as follows.
\[
    \begin{array}{r l@{\hspace{2em}}r l@{\hspace{2em}}rl}
        \semFock{\intType} &\equaldef \Z & 
        \semFock{\phaseType} &\equaldef \R &
        \semFock{\memoryType} &\equaldef \fockSpace \\
        \semFock{\boolType} &\equaldef \B & 
        \semFock{\normType} &\equaldef \R & 
        \semFock{\hpsType} &\equaldef \fockSpace \sqcup \{\bot\} 
    \end{array}
\]

Here, the phases and norms evaluate to real numbers,\footnote{Specifically, the
dyadics $\mathbb D = \{\frac i {2^j} : i,j \in \Z\}$ and the constructibles
(generated by $\Q$ and $\sqrt{-}$), respectively.} while the memories evaluate
to basis vectors in the Fock space $\fockSpace$ defined in \cref{def:fock-space}. Fock
spaces are the standard approach to modeling systems of arbitrary
size~\cite{F932}, which we need here to describe past logs of varying lengths.
\begin{definition}[Extended state space]%
    \label{def:fock-space}
    \ifmargins%
        \marginpar{$\fockSpace$,
        $\basis_{\pastSpaceDecor}$, $\fockSpace_{\pastSpaceDecor}$}
    \fi
    The \emph{Fock space}\footnote{A variation of the bosonic/fermionic
        Fock space is used here that does not involve symmetric/anti-symmetric
    operators.}
    $\fockSpace(\signature)$ associated to the signature $\signature$ is given by:
    \[
            \fockSpace(\signature) \equaldef \hilbert[\signature] \otimes
            \fockSpace_{\pastSpaceDecor}
            \quad \text{where} \quad
        \fockSpace_{\pastSpaceDecor} \equaldef \overline{\bigoplus_{i\in \N}
        {\hilbert_2}^{\otimes i}} \otimes \overline {\bigoplus_{i\in \N}
        {\ell^2(\Z)}^{\otimes i}}
    \]
    We also define the total space $\fockSpace \equaldef \bigcup_{\signature
    \finsubseteq \addressSet} \fockSpace(\signature)$, and
    $\basis_{\pastSpaceDecor} \equaldef \mathtt{List}(\B) \times
    \mathtt{List}(\Z)$ the set of basis vectors of $\fockSpace_{\pastSpaceDecor}$.
\end{definition}

The Fock space $\fockSpace(\signature)$ for $\signature$ is precisely the
Hilbert state $\hilbert[\signature]$ of the present state tensored with the Fock
space $\fockSpace_{\pastSpaceDecor}$ containing a log of all past classical
values. 
Fock spaces are \emph{completed} (the overline) as metric
spaces, meaning that they may contain vectors supported on infinitely many basis
vectors as long as their Hilbert norm remains finite. In our case, this means
that we can encode infinite mixes of states with unbounded past logs, which is
essential in the context of \texttt{while} loops.  Finally, note that the total
space $\fockSpace$ is intentionally not made into a vector space itself, as
otherwise we could encode undesirable superpositions of different signatures,
hence no longer being able to decide if a program is well-formed or not.

Variable assignment environments are given as partial maps $\assignmentContext:
(\intVarSet \rightharpoonup \Z) \times (\boolVarSet \rightharpoonup \B)$ for
which we commit an abuse of notation: for a variable $\var$, we write
$\assignmentContext(\var)$ in $\B$ or in $\Z$, when defined, and according to
the type of $\var$. We also write
$\assignmentContext' =
\assignmentContext[\var \mapsto v]$ to denote the environment such that
$\assignmentContext'(\var) = v$ and $\assignmentContext'(\var') =
\assignmentContext(\var')$ for all $\var' \neq \var$\footnote{Using a Kleene
equality so that $\assignmentContext(\var')$ for $\var'\neq \var$ is undefined
if $\assignmentContext(\var')$ is also undefined}. We write
$\semFock[\assignmentContext]{t}$ for the interpretation of a term $t$ under an
environment $\assignmentContext$ containing
all its free variables
(i.e., $\mathrm{FV}(t) \subseteq
\mathrm{dom}(\assignmentContext)\subseteq \varSet$), with $\semFock[\assignmentContext]{t} \in
\semFock{T}$, when $t$ is of type $T$. If $t$ is a closed term (no free
variables), we write $\semFock{t} \equaldef
\semFock{t}_\emptyset$. These interpretations are defined by structural
induction on $t$. For $\intType$, $\boolType$, $\phaseType$, and
$\normType$, they are as expected ($\semFock[\assignmentContext]{\intTerm_1 + \intTerm_2} =
\semFock[\assignmentContext]{\intTerm_1} +
\semFock[\assignmentContext]{\intTerm_2}$) \appx. More notable is the
interpretation\footnote{In \HQbricks~\cite{CIN+26}, the analog of $\semFock{\h}$ is
$\mathcal V(\h)$.} $\semFock[\assignmentContext]{\h}$
of an \ihps{} $\h$ as a vector of $\fockSpace(\signature(\h))$ where, $\mathcal
F(\signature(\h))$ is the Fock space from \cref{def:fock-space} and $\bot$ is
the image of an $\h$ containing $\sum$, $\lim$, or $\bigotimes$ that do not
converge. In practice, the soundness of the logic
(\cref{thm:soundness}) and the convergence of the denotational semantics
(\cref{lem:semantics-well-defined}) ensure that only convergent
sequences appear from program analysis.

\begin{definition}[\ihps{} interpretation]%
    \label{def:hps-interpretation}
    \ifmargins%
        \marginpar{$\semFock{\h}$}
    \fi
    An \ihps{} $\h$ is interpreted as an element of
    $\fockSpace(\signature(\h)) \cup \{\bot\}$
\[
    \begin{array}{rl@{\quad}rl}
        \semFock[\assignmentContext]{\pasum{p}{m}{n}} &\equaldef \semFock[\assignmentContext]{n} e^{2\pi i
        \semFock[\assignmentContext]{p}} \semFock[\assignmentContext]{m} &\textstyle \semFock[\assignmentContext]{\bigpsadd_{\var} \h} &\equaldef \textstyle\sum_{\varValue \in \mathrm{dom}(\var)} \semFock{\h}_{\assignmentContext[\var
        \mapsto \varValue]} \\
        \semFock[\assignmentContext]{\h_1 \psadd \h_2} & \equaldef \semFock[\assignmentContext]{\h_1} +
        \semFock[\assignmentContext]{\h_2} &\textstyle \semFock[\assignmentContext]{\bigotimes_{\var} \h} &\textstyle \equaldef \bigotimes_{\varValue \in \mathrm{dom}(\var)} \semFock{\h}_{\assignmentContext[\var
            \mapsto \varValue]} \\
        \semFock[\assignmentContext]{\h_1 \sqpsadd \h_2} & \equaldef
        \semFock[\assignmentContext]{\h_1} \oplus
        \semFock[\assignmentContext]{\h_2} & \textstyle \semFock[\assignmentContext]{\lim_\intVar \h} &\textstyle
                \equaldef \lim_{\intValue \to \infty}
                \semFock{\h}_{\assignmentContext[\intVar \mapsto \intValue]}
 \\
        \semFock[\assignmentContext]{\h_1 \otimes \h_2} & \equaldef \semFock[\assignmentContext]{\h_1} \otimes
        \semFock[\assignmentContext]{\h_2}\\
                            \end{array}
\]
where $v_1 \oplus v_2 \equaldef \pastKet[\B]{0} \otimes v_1 + \pastKet[\B]{1}
\otimes v_2$ and $\mathrm{dom}(\intVar) = \N$ for integers\footnote{We do not sum over
$\Z$; only over $\N$.} and $\mathrm{dom}(\boolVar)=\B$ for booleans. When
$\sum$, $\bigotimes$, or $\lim$ do not converge in the Hilbert norm on
$\fockSpace(\signature)$, we write $\semFock[\assignmentContext]{\h} = \bot$, an
absorbing element for all constructors. When $\assignmentContext = \emptyset$,
we write $\semFock{\h} \equaldef \semFock[\emptyset]{\h}$. 
\end{definition}

\begin{example}%
    \label{running:bigpsadd}
    The states $\h_1$ and $\h_{\textsc{CT-}\infty}$ from \cref{ex:ihps} are
    interpreted in $\fockSpace$ as $\semFock{\h_1} = \frac 1 {\sqrt 2}
    \ket[\qbit]{0} \classKet[\cbit]{0} + \frac 1 {\sqrt 2} \ket[\qbit]{1}
    \classKet[\cbit]{1}$ and $\semFock{\h_{\textsc{CT-}\infty}} =
    \sum_{{\intVar}=0}^\infty \frac 1 {\sqrt {2^{{\intVar}+1}}} e^{2\pi i \cdot
    0} \ket[\qbit]{0} \classKet[\cbit]{0} \classKet[\cInt]{\intVar} \in \mathcal
    F(\signature_{\textsc{CT}}) \subseteq \fockSpace$.
\end{example}

    \subsection{CQ Interpretations and Observable Information}%
\label{sub:Observable information}

The interpretation $\semFock[]{\h}$ of a closed \ihps{} $\h$ contains more
information than is accessible via actual physical observations without
knowledge of the system's past or of the phases that are fully determined by
the classical parts. Alternatively, we can extract exactly and only the
observable information of $\h$ as $\semRho{\h}$, the corresponding CQ state
(\cref{def:cq-state}). To obtain $\semRho{\h}$, start by forming
$\ketbra{\semFock{\h}}{\semFock{\h}}$, then drop the past (using a partial
trace) and erase unobservable phases (using a dephasing channel) as described in
\cref{def:density-operator-interpretation}.

\begin{definition}[CQ state interpretation]%
    \label{def:density-operator-interpretation}
    \ifmargins%
        \marginpar{$\semRho{\h}$, $\|\h\|$}
    \fi
    Let $\h \in \hpsType$, and $\assignmentContext$ be an environment with domain
    $\dom{\assignmentContext} \supseteq \freeVars{\h}$. If
    $\semFock[\assignmentContext]{\h} \neq \bot$, the \emph{CQ state
    interpretation} of $\h$ under $\assignmentContext$ is given by
    \[
        \semRho[\assignmentContext]{\h} \equaldef \left(\widehat{\dephase_{\signature(\h) \cap
            \classAddressSet}} \right) (\mathrm{tr}_{\fockSpace_{\pastSpaceDecor}} (
            |\semFock[\assignmentContext]{\h}\rangle
            \langle\semFock[\assignmentContext]{\h}|)) \in
            \mathrm{CQ}(\signature(\h)),
    \]
    
    \noindent where $\mathrm{tr}_{\fockSpace_{\pastSpaceDecor}}$ is the partial
    trace over $\fockSpace_{\pastSpaceDecor}$, and
    $\widehat{\dephase_{\signature(\h) \cap \classAddressSet}}$ is the dephasing
    channel of \cref{def:cq-state}. As with $\semFock{\h}$, we write
    $\semRho{\h}$ when $\h$ is closed and $\Gamma = \emptyset$. Note,
    $\semRho{\h}$ is the \ihps{} analog of \HQbricks' $\mathcal{DO}(\h)$.
\end{definition}
\begin{example}[CQ interpretation]%
    \label{ex:density-operator-interpretation}
    Consider $\h = \sum_{\boolVar_2} \sum_{\boolVar_1} \pasuminline{\frac{\boolVar_1
    \boolVar_2}{2}}{\ket[\qbit]{\boolVar_2} \pastKet[\B]{\boolVar_1}}
    {\frac 1 2}$, which is a mixture of $\ket +$ and $\ket -$. While the Fock
    interpretation retains the decomposition $\semFock{\h} = \frac 1 2 \ket +
    \otimes{\pastKet[\B]{0}} + \frac 1 2 \ket - \otimes{\pastKet[\B]{1}}$, the
    CQ interpretation reflects only the observable mixture: $\semRho{\h} =
    \frac 1 {\sqrt 2} \ketbra{+}{+}_\qbit + \frac 1 {\sqrt 2}
    \ketbra{-}{-}_\qbit = \frac 1 2 I_\qbit$.
\end{example}

The CQ state interpretation naturally introduces a notion of equivalence between
\ihps{} representing the same CQ state.

\begin{definition}[Equivalence]%
    \label{def:equivalence-hps}
    An \emph{equivalence statement} is a pair
    $\h_1 \equiv \h_2$ of \ihps{} $\h_1, \h_2 \in \hpsType$. Let $\assignmentContext$ be an assignment
    environment over $\freeVars{\h_1} \cup \freeVars{\h_2}$. We say that
    $\assignmentContext$ \emph{models} $\h_1 \equiv \h_2$ iff
    $\semRho[\assignmentContext]{\h_1} = \semRho[\assignmentContext]{\h_2}$,
    denoted $\assignmentContext \models \h_1 \equiv \h_2$. The equivalence
    $\h_1 \equiv \h_2$ is \emph{valid} if $\assignmentContext \models \h_1
    \equiv \h_2$ for all assignment environments $\assignmentContext$ over
    $\freeVars{\h_1} \cup \freeVars{\h_2}$, denoted $\models \h_1 \equiv \h_2$.
    In this case, we also say that $\h_1$ and $\h_2$ are \emph{equivalent}.
\end{definition}

\begin{example}%
    \label{ex:equivalence}
We have $\models \sum_\boolVar \pasuminline{\frac \boolVar 2}{\ket[\qbit]{\boolVar}
\pastKet[\B]{\boolVar}}{\frac 1 {\sqrt 2}} \equiv
    \sum_{\boolVar_1, \boolVar_2} \pasuminline{\frac
    {\boolVar_1\boolVar_2}{2}}{\ket[\qbit]{\boolVar_2} \pastKet[\B]{\boolVar_1}}{\frac 1
{\sqrt 2}}$, as both represent the maximally mixed state $\frac 1 2 I_\qbit$.
\end{example}

The last definition we add is that of the norm of an \ihps{}, which serves to
compute the probabilities of certain events such as termination or an address
containing a certain value.
\begin{definition}[\ihps{} norm]%
    \label{def:hps-norm} 
    \ifmargins%
        \marginpar{$\|\h\|$}
    \fi
    Given a closed $\h \in \hpsType$ such that $\semFock{\h} \neq \bot$, the
    \emph{norm} of $\h$ is $\|\h\| \equaldef \sqrt{\mathrm{tr}(\semRho{\h})}$;
    equivalently, $\|\h\| = \|\semFock{\h}\|$.
\end{definition}
\begin{example}%
    \label{running:hps-norm}
    $\|\h_{\textsc{CT-}\infty}\| =
    \sqrt{\mathrm{tr}(\semRho{\h_{\textsc{CT-}\infty}})} =
    \sum_{{\intVar}=0}^\infty \frac 1 {2^{{\intVar}+1}}
    = 1$, the termination probability of \ref{cointoss}.
\end{example}

\section{\ihps{}-based Symbolic Analysis}\label{sec:operational-semantics}
The \ihps{}-based framework we propose relies on two main pillars: (i) a \emph{quantum
Hoare logic} with \ihps{} used as predicates for pre- and
post-conditions, (ii) an \emph{equational theory} allowing rewriting \ihps{} from
one form to another. The strategy is therefore to use the logic to perform a
forward-directed symbolic execution of the program while substituting one \ihps{}
with another when needed for the analysis of loop invariants and the simplification of specifications. Overall, this
gives the logic a more operational, automatable flavor, highlighted notably by
the implementation in \cref{sec:implementation}.

In this section, we first lay out the foundations of the logic
(\cref{sub:operational-semantics}), then those of the
equational theory (\cref{sub:equational-theory}), and
prove their soundness and (partial) adequacy (\cref{sub:soundness-and-adequacy})

    \subsection{Hoare Logic}%
\label{sub:operational-semantics}

\ihps{} are used to provide a Hoare logic for $\lang$ symbolically 
representing program execution.  This logic is
non-branching under program evolution. This allows for a symbolic and compact
representation of the program structure, which avoids growing 
exponentially in the number of branching and relies on \ihps{} invariants to
capture the infinite branchings of loops.

For $\h, \h' \in \hpsType$ and $\prog \in \programSet$, when
$\validProg{\signature(\h)}{\prog}{\signature(\h')}$, we can write the Hoare
triple $\hoare{\h}{\prog}{\h'}$ to denote that $\prog$ transforms $\h$ into
$\h'$. A general triplet will be denoted by $\hoareVar$ for `specification'.
We also define the logical context $\logicContext$ as a set of Hoare triples and
write judgments of the form
\[
    \judgement{\logicContext}{\equivContext}{\hoare{\h}{\prog}{\h'}}
\]
to denote that $\hoare{\h}{\prog}{\h'}$ can be derived from $\logicContext$
using the logic rules. When $\logicContext = \emptyset$,
we simply omit writing it as in $\judgement{}{}{\hoare{\h}{\prog}{\h'}}$. The
\ihps{} appearing in the Hoare triples are subject to substitution by equivalent
ones (\cref{def:equivalence-hps}). As such, we assume a sound equational theory $\equivvdash$ for \ihps{}
(\cref{sub:equational-theory}).)

The rules for deriving judgments of our Hoare logic are given in \cref{fig:hoare-logic}.
The logic rules use auxiliary operations on \ihps{} as follows. The assignment operators
$\h[\address \ot \programTerm]$ assigns to $\address$ the expression to which
$\programTerm$ evaluates in $\h$ while prepending the previously held expression
to the past $\pastKet -$; filters $\bool * \h$ select only those paths that
satisfy $\bool$ from $\h$; $\applyU{\unitary(\bar \qbit)}(\h)$ applies the
unitary $\unitary$ to the qubits $\bar \qbit$ in $\h$.  These operators have
very technical definitions \appx; we also refer the reader to
\cref{running:hoare-logic} for illustration and intuition.
Finally, $\h[\intVar+1/\intVar]$ is
the substitution of the free occurrences of $\intVar$ in $\h$ by $\intVar+1$, and
$\h \otimes m \equaldef \h \otimes \pasuminline{0}{m}{1}$.

\begin{figure}[t]
    \begin{prooftree}
        \AxiomC{$\hoareVar \in \logicContext$}
        \RightLabel{\labelrule{rule:ax}{$\textsc{Ax}$}}
        \UnaryInfC{$\judgement{\logicContext}{\equivContext}{\hoareVar}$}
        \DisplayProof\hskip 1em
        \AxiomC{$\phantom{\Delta}$}
        \RightLabel{\labelrule{rule:skip}{$\textsc{Skip}$}}
        \UnaryInfC{$\judgement{\logicContext}{\equivContext}{\hoare{\h}{\skipProg}{\h}}$}
        \DisplayProof\hskip 1em
        \AxiomC{$\phantom{\Delta}$}
        \RightLabel{\labelrule{rule:qinit}{$\textsc{QInit}$}}
        \UnaryInfC{$\judgement{\logicContext}{\equivContext}{\hoare{\h}{\qubitInitProg{\qbit}}{\h \otimes
        \ket[\qbit]{0}}}$}
        \DisplayProof\vskip 1em
        \AxiomC{}
        \RightLabel{\labelrule{rule:cinit}{$\textsc{CInit}$}}
        \UnaryInfC{$\judgement{\logicContext}{\equivContext}{\hoare{\h}{\cbitInitProg{\cbit}}{\h \otimes
        \classKet[\cbit]{0}}}$}
        \DisplayProof\hskip 1em
        \AxiomC{}
        \RightLabel{\labelrule{rule:intinit}{$\textsc{IntInit}$}}
        \UnaryInfC{$\judgement{\logicContext}{\equivContext}{\hoare{\h}{\intInitProg{\cInt}}{\h \otimes
        \classKet[\cInt]{0}}}$}
        \DisplayProof\vskip 1em
        \AxiomC{}
        \RightLabel{\labelrule{rule:unitary}{$\textsc{Unitary}$}}
        \UnaryInfC{$\judgement{\logicContext}{\equivContext}{\hoare{\h}{\unitary(\bar \qbit)}{
            \applyU{\unitary(\bar \qbit)}(\h)}}$}
        \DisplayProof\vskip 1em
        \AxiomC{}
        \RightLabel{\labelrule{rule:assign-int}{$\textsc{Assign}_\integer$}}
        \UnaryInfC{$\judgement{\logicContext}{\equivContext}{\hoare{\h}{\assign{\cInt}{\integer}}{\h[\cInt \ot \integer]}}$}
        \DisplayProof\hskip 1em
        \AxiomC{}
        \RightLabel{\labelrule{rule:assign-bool}{$\textsc{Assign}_\bool$}}
        \UnaryInfC{$\judgement{\logicContext}{\equivContext}{\hoare{\h}{\assign{\cbit}{\bool}}{\h[\cbit \ot \bool]}}$}
        \DisplayProof\vskip 1em
        \AxiomC{}
        \RightLabel{\labelrule{rule:measure}{$\textsc{Measure}$}}
        \UnaryInfC{$\judgement{\logicContext}{\equivContext}{\hoare{\h}{\measure{\cbit}{\qbit}}{\h[\cbit \ot \qbit]}} $}
        \DisplayProof\vskip 1em
        \AxiomC{$\judgement{\logicContext}{\equivContext}{\hoare{\h_1}{\prog_1}{\h_2}}$}
        \AxiomC{$\judgement{\logicContext}{\equivContext}{\hoare{\h_2}{\prog_2}{\h_3}}$}
        \RightLabel{\labelrule{rule:seq}{$\textsc{Seq}$}}
        \BinaryInfC{$\judgement{\logicContext}{\equivContext}{\hoare{\h_1}{\prog_1 ; \prog_2}{\h_3}}$}
        \DisplayProof\vskip 1em
        \AxiomC{$\judgement{\logicContext}{\equivContext}{\hoare{\bool * \h}{\prog_1}{\h_1}}$}
        \AxiomC{$\judgement{\logicContext}{\equivContext}{\hoare{(1 \oplus \bool) * \h}{\prog_2}{\h_2}}$}
        \RightLabel{\labelrule{rule:if}{$\textsc{If}$}}
        \BinaryInfC{$\judgement{\logicContext}{\equivContext}{\hoare{\h}{\ifthene{\bool}{\prog_1}{\prog_2}}{\h_1 \sqpsadd \h_2}}$}
        \DisplayProof\vskip 1em
        \AxiomC{$\judgement{\logicContext}{\equivContext}{\hoare{\h[{\intVar}]}
            {\ifthene{\bool}{\prog}{\skipProg}}{\h[{\intVar}+1/{\intVar}]}}$}
        \RightLabel{\labelrule{rule:while}{$\textsc{While}$}}
        \UnaryInfC{$\judgement{\logicContext}{\equivContext}{\hoare{\h[0/{\intVar}]}{\while{\bool}{\prog}}{\displaystyle
        \lim_{{\intVar}} \left((1 \oplus \bool) * \h[\intVar]\right)}}$}
        \DisplayProof\vskip 1em
        \AxiomC{$\equivvdash \h_1' \equiv \h_1$}
        \AxiomC{$\judgement{\logicContext}{\equivContext}{\hoare{\h_1}{\prog}{\h_2}}$}
        \AxiomC{$\equivvdash \h_2 \equiv \h_2'$}
        \RightLabel{\labelrule{rule:equiv}{$\textsc{Equiv}$}}
        \TrinaryInfC{$\judgement{\logicContext}{\equivContext}{\hoare{\h_1'}{\prog}{\h_2'}}$}
        \end{prooftree}
        \caption{Hoare logic rules for \ihps{}.}\label{fig:hoare-logic}
\end{figure}

\subsubsection*{Intuition and discussion.}
The logic is non-branching: \emph{quantum superposition} with the Hadamard
gate $H$ are introduced by the addition of a path variable, the \emph{classical
nondeterminism} of measurement merely relabels qubits (see
\cref{running:hoare-logic}), and \emph{conditionals} (see
\cref{fig:hoare-logic}) are encoded in the same \ihps{}.
The logic rules include monoidal rules (\ref{rule:skip},\ \ref{rule:seq}),  and
a substitution rule (\ref{rule:equiv}) that allows us to
substitute equivalent \ihps{} with each other, according to an equational theory
$\equivvdash$ (see \cref{sub:equational-theory}). They also include the
application of unitaries (\ref{rule:unitary}) such as the $X$ and
$H$ gates described in \cref{sub:Overview}. Next, initializations
(\ref{rule:qinit},\ \ref{rule:cinit},\ \ref{rule:intinit}) are handled by
tensoring extra subsystems in the zero state. On the more nuanced side, the
assignments\ \ref{rule:assign-bool} and\ \ref{rule:assign-int} and the
measurement\ \ref{rule:measure} additionally tensor the previously held boolean
or integer expressions to the history of the \ihps{}. Finally, classical control
(\ref{rule:if},\ \ref{rule:while}) uses the filtering operation $\bool * \h$ to
select the relevant paths of the \ihps{} according to the condition $\bool$ and
applies the corresponding branch or loop body accordingly. In\ \ref{rule:if},
using $\sqpsadd$ instead of $\psadd$ ensures the orthogonality of the two
branches is preserved despite substitutions with\ \ref{rule:equiv}.
Finally, the\ \ref{rule:while} rule uses a loop invariant $\h$ in the form of an
\ihps{} with a free integer variable ${\intVar}$ representing the iteration
number of the loop. The limiting behavior of the loop is then captured by
filtering out the non-exiting states and pushing ${\intVar}$ to
infinity.

\begin{example}%
    \label{running:hoare-logic}
    We illustrate our logic in \cref{tab:hoare-logic-derivation} by applying the
    first iteration of~\ref{cointoss}, excluding the counter $\cInt$ for
    simplicity.  We can see in \cref{tab:hoare-logic-derivation} the use of
    $\applyU{-}$ on row \texttt{2} as well as projections $\h_2[\cbit \ot
    \qbit]$ on row \texttt{3} to model measurement. To illustrate the
    functioning of the filtering $\bool * \h$ in the logic rule \ref{rule:if}, consider
    computing $\cbit * \h_3$ to select the branch where $\cbit$ is $1$ by
    multiplying the norm by a factor of $\boolVar$:
    \[
        \cbit * \h_3 = \bigpsadd_{\boolVar} \pasum{0}{\ket[\qbit]{\boolVar}
        \classKet[\cbit]{1}}{\boolVar \frac 1 {\sqrt 2}} \equiv
        \pasum{0}{\ket[\qbit]{1} \classKet[\cbit]{1}}{\frac 1
        {\sqrt 2}}.
    \]

    \begin{table}[htpb]
        \centering
        \caption{The derivation of the first iteration of~\ref{cointoss} in our
        logic.}\label{tab:hoare-logic-derivation}
        \[
            \begin{array}{c@{\;}|c@{\;}|c|@{\;}l|}
                \cline{2-4}
                \text{} & \text{Step}  & \text{State ($\qbit,
                \cbit$)} & \multicolumn{1}{c|}{\text{Hybrid path-sum}} \\
                \hline
                \color{red}\texttt{1} & \text{Initialization}& \ket{0}, 0&
                \rule{0pt}{3ex}\rule[-2.2ex]{0pt}{0pt}\;\h_0 \equaldef \pasum{0}{\ket[\qbit]{0}\classKet[\cbit]{0}}{1} \\
                \hline
                                      &&& \rule{0pt}{2.5ex}\;\h_2 \equaldef \applyU{H(\qbit)}(\h_0) \\[3pt]
                \raisebox{1em}{\color{red}\texttt{2}} & \mathtt{H}(\qbit) & \frac{1}{\sqrt 2}
                \left(\ket{0} + \ket{1} \right), 0&
                \rule[-2.2ex]{0pt}{0pt}\phantom{\h_2} \;= \bigpsadd_{\boolVar} \pasum{0}{\ket[\qbit]{\boolVar}\classKet[\cbit]{0}}{\frac 1 {\sqrt 2}} \\
                \hline
                                                  &&& \rule{0pt}{2.5ex}\;\h_3 \equaldef \h_2[\cbit \ot \qbit] \\
                \raisebox{1em}{\color{red}\texttt{3}} & \measure{\cbit}{\qbit} &
                \begin{array}{l@{\quad}l}
                    \rule{0pt}{1.2em} 
                    |0\rangle,0 & \text{with probability } \frac 1 2 \\[3pt]
                    |1\rangle,1 & \text{with probability } \frac 1 2 \\
                    \vspace{-0.8em}
                \end{array} & 
                \phantom{\h_3} \;= \bigpsadd_{\boolVar} \pasum{0}{\ket[\qbit]{\boolVar} {\left[\boolVar\right]}_{\cbit}}{\frac 1 {\sqrt 2}} \\
                \hline
            \end{array}
        \]
    \end{table}

\end{example}

\paragraph{Reasoning modulo theory.}
The problems of checking the equivalence
of two \ihps{} and of checking the validity of a Hoare triple are undecidable (see \cref{thm:undecidability}).
To overcome this, the framework allows the user
to assume certain Hoare triples
$\logicContext$ as axioms and continue using the logic to derive the desired
properties.  This
technique also allows the framework to be used as-is with extensions
of $\lang$ to new unitaries or other black-box operations, as long as their
semantics are provided as Hoare triple axioms. This is illustrated by the
Quantum Bernoulli Factory (QBF) in \cref{sec:rus-bernoulli} where the \ihps{}
semantics of the unitary $U_p = \sqrt{p} I + \sqrt{1-p} X$, which is not part of
Clifford+$R_k$, is given as an axiom of the form $\hoare{h_0}{U_p(\qbit)}{h_1}
\in \logicContext$.
This allows us to analyze the QBF despite the language not technically including
$U_p$.

\paragraph{Semantics.}
The semantics of the judgments is as follows. Let $\assignmentContext$ be a
variable assignment environment over the set $\freeVars{\h} \cup \freeVars{\h'}$ of
free variables occurring in $\h$ and $\h'$. Then, $\assignmentContext$
\emph{models} a triple $\hoare{\h}{\prog}{\h'}$ denoted $\assignmentContext
\models \hoare{\h}{\prog}{\h'}$ iff $\sem{\prog}(\semRho[\assignmentContext]{ \h
}) = \semRho[\assignmentContext]{\h' }$.  Similarly, $\assignmentContext$
\emph{models} $\judgement{\logicContext}{\equivContext}{\hoareVar}$ denoted
$\assignmentContext \models \judgement{\logicContext}{\equivContext}{\hoareVar}$ iff
$\assignmentContext \models \hoareVar$.  The judgment
$\judgement{\logicContext}{\equivContext}{\hoare{\h}{\prog}{\h'}}$ is
\emph{valid} iff $\assignmentContext \models
\judgement{\logicContext}{\equivContext}{\hoare{\h}{\prog}{\h'}}$ for all
$\assignmentContext$ over $\freeVars{\h} \cup \freeVars{\h'}$. In this case, we
write $\semJudgement{\logicContext}{\equivContext}{\hoare{\h}{\prog}{\h'}}$.

\paragraph{The question of decidability.} The problem of checking the validity
of a triple $\hoare{\h}{\prog}{\h'}$ is undecidable, as described in
\cref{thm:undecidability} below. Specifically, it is $\Pi_1^0$-hard; that is, the
problem of deciding the validity of formulae of the form $\forall \intVar_1,
\ldots, \forall \intVar_n \psi(\intVar_1, \ldots, \intVar_n)$, where $\intVar_1,
\ldots, \intVar_n$ are integer variables and $\psi$ is a quantifier-free
formula, can be reduced to the problem of checking the validity of some Hoare
triple $\hoare{\h_1}{\prog}{\h_2}$ \appx.  This is expected given the expressivity of the language, specifically,
the existence of unbounded loops.

\begin{restatable}[Undecidability]{theorem}{restateUndecidability}%
    \label{thm:undecidability}
    The following problems are $\Pi_1^0$-hard, for \ihps{} terms $\h_1$ and
    $\h_2$:
    \begin{enumerate}
        \item Checking whether $\models \h_1 \equiv \h_2$.
        \item Checking whether $\|\h_1\| = \|\h_2\|$.
        \item Checking whether $\models \hoare{\h_1}{\prog}{\h_2}$ for a given
            program $\prog$.
    \end{enumerate}
    The problems remain $\Pi_1^0$-hard, even when $\h_1$ and $\h_2$ are closed terms.
\end{restatable}

The practical implication of this is that the
logic is not complete and can never be fully automated. As such, we rely on some
user input in terms of loop invariants and certain equivalence checks to
complete the proofs.


\subsection{Equational Theory}\label{sub:equational-theory}
In contrast to the state-of-the-art quantum Hoare logics, our logic relies on the \ihps{} representation, which is amenable to effective and tractable equational theories. In this section, we elaborate on the choice of said theory.

Substituting an \ihps{} by an equivalent one is necessary for the analysis of
loops. In fact, in general, the rules of the logic produce \ihps{}
post-conditions that are structurally larger than the pre-conditions and loop
invariants cannot be shown to be conserved exactly, but only up to equivalence.
This is where the strength of the path-sum approach shines as it has, since its
inception, been designed to be amenable to rewriting~\cite{A18} with complete
rewriting theories~\cite{V24,A23} having been developed for the purely quantum
case, as well as richer extensions to equational theories of \hps{} in the
hybrid case, as developed in \HQbricks~\cite{CIN+26}.

We enrich and adapt the equational
theories from the \hps{} formalism introduced in \HQbricks~\cite{CIN+26} to the
\ihps{} formalism introduced in this article.

The rules of the equational theory that are inherited from
\HQbricks~\cite{CIN+26} can roughly be divided into four categories: the
\emph{interfere rules} dating back to the original path-sum formalism~\cite{A18}
and allowing for the simplification of interference patterns appearing from
specific circuit equivalence instances such as\
\ref{rule:phase-bisector}\footnote{Technically, \ref{rule:phase-bisector} proves
a stronger $\equiv^{\mathtt{P}}$ equivalence stating $\semFock{\h_1} = \semFock{\h_2}$. \appx}, a
generalization of the HH rule based on $HH= I$~\cite{A18}. The
\emph{algebraic} rules corresponding to standard axioms of vector spaces (e.g., 
\textcolor{darkred}{Add-comm} in \HQbricks), and dating back to the introduction
of unbalanced path-sums~\cite{A23}; and the \emph{world combination} rules,
specific to \hps{}, allowing simplifications specific to the hybrid aspect of
\hps{} in terms of `gauge' symmetries such as global phase elimination
(\textcolor{darkred}{PE} in \HQbricks), elimination of constant past values not
contributing to any separation of worlds (\ref{rule:forget}), merging worlds
differing only in their past values, etc.

\begin{prooftree}
    \AxiomC{$y \not \in \text{Var}(p,q,n,f)$}
    \RightLabel{\labelrule{rule:phase-bisector}{\textsc{PhaseBisector}}}
    \UnaryInfC{$\equivContext \equivvdash \sum_y \pasum{p + y q}{f}{n}
            \equiv \pasum{p + \frac q 2}{f}{2\cos(2\pi \frac q 2) n}$}
    \DisplayProof\vskip 1em
    \AxiomC{$\text{Var}(f) = \emptyset$}
    \RightLabel{\labelrule{rule:forget}{\textsc{Forget}}}
    \UnaryInfC{$\equivContext \equivvdash \pasum{p}{m \pastKet[\mathbb T]{f}}{n}
            \equivWeak \pasum{p}{m}{n}$}
\end{prooftree}

In addition to the rules inherited from \HQbricks~\cite{CIN+26}, we naturally
introduce new \emph{limit rules} specific to the integer \ihps{} formalism,
allowing for the computation and simplification of quantifiers over infinite
domains, including limits $\lim_{\intVar} \h$, sums $\sum_{\intVar} \h$ and
products $\bigotimes_{\intVar} \h$. Non-exhaustively, these include the
computation of limits of path-sums when the underlying functions $p, n, f$
converge\ \ref{rule:comp-cont}, the vanishing of path-sums when the
amplitude function $n$ converges to zero\ \ref{rule:vanish}, and various
operator commutation rules corresponding to various continuity theorems, such
as\ \ref{rule:plus-cont}. Note: in\ \ref{rule:vanish} and\ \ref{rule:comp-cont},
$\to_k$ denotes the convergence, in the standard topology, of a real term to a
real constant when the variable $k$ tends to infinity, established separately,
either manually or by the aid of computer algebra systems. 

\begin{prooftree}
    \AxiomC{\vphantom{$a^b$}}
    \RightLabel{\labelrule{rule:plus-cont}{\textsc{PlusCont}}}
    \UnaryInfC{$\equivContext \equivvdash \lim_k (\h_1 + \h_2) \equiv \lim_k \h_1 + \lim_k \h_2$}
    \DisplayProof\hskip 1em
    \AxiomC{$n \to_k 0$}
    \RightLabel{\labelrule{rule:vanish}{\textsc{Vanish}}}
    \UnaryInfC{$\equivContext \equivvdash \lim_k \pasum{p}{f}{n} \equivWeak 0$}
    \DisplayProof\vskip 1em
    \AxiomC{$p \to_k p_\infty$}
    \AxiomC{$n \to_k n_\infty$}
    \AxiomC{$k \not \in \freeVars{f}$}
    \RightLabel{\labelrule{rule:comp-cont}{\textsc{CompCont}}}
    \TrinaryInfC{$\equivContext \equivvdash \lim_k \pasum{p}{f}{n} \equivWeak
            \pasum{p_\infty}{f}{n_\infty}$}
\end{prooftree}
\vskip 1em

    \subsection{Soundness and Adequacy}%
\label{sub:soundness-and-adequacy}

We note that, as is typical in analysis, nearly all limit rules require some
form of convergence condition. These conditions impose a natural restriction on
the domain of soundness of the rules as in \cref{thm:rewrite-soundness}.

\begin{restatable}[Soundness of the equational theory]{theorem}{restate Rewrite Soundness}%
	\label{thm:rewrite-soundness}
    \[
        \forall \h_1, \h_2, \equivvdash \h_1 \equiv \h_2 \implies \models_\equiv
        \h_1 \equiv \h_2
    \]
 
\end{restatable}

Once the convergence condition is satisfied, the soundness of the limit rules
is an immediate consequence of basic analysis results, namely, the continuity of
vector addition for \ref{rule:plus-cont} and of scalar multiplication for
\ref{rule:comp-cont} as well as absolute convergence implying convergence for
\ref{rule:vanish}. As for the rules inherited from the \hps{} formalism in
\HQbricks~\cite{CIN+26}, we do not reprove them and refer the reader to that
article for details.

We similarly show the soundness of the logic with respect to the denotational
semantics 
(\cref{fig:standard-semantics}).

\begin{restatable}[Soundness of the logic]{theorem}{restatePhysicality}%
    \label{thm:soundness}
    \[
        \forall \prog, \h_1, \h_2,
        \validProg{\signature(\h_1)}{\prog}{\signature(\h_2)} \land
        \hoare{\h_1}{\prog}{\h_2} \implies \sem{\prog}(\semRho{\h_1})
        = \semRho{\h_2}
    \]
\end{restatable}

The soundness of the semantics is a core result of our work, as it shows that
the properties of the program obtained by symbolic analysis are indeed correct with
respect to the denotational semantics. In particular, the expected values
$\expectation[A \mid \rho] \equaldef \mathrm{tr}(\rho A)$ of an observable $A \in
\mathcal L(\hilbert[\signature(\rho)])$ over the CQ state $\rho$,
including termination probability (for $A = I$), expected runtime (for $A =
\cInt$ with $\cInt$ a loop counter), and others, can be extracted without the
need to compute the difficult CQ state semantics directly. Instead, we could
pass to \ihps{}-based symbolic execution and rewriting.

\begin{restatable}{corollary}{restateEstimationViaSymbolicExecution}%
    \label{cor:estimation-via-symbolic-execution}
    For all $\prog, \h_1, \h_2$ such that $\validProg{\signature(\h_1)}{\prog}{\signature(\h_2)}$, and observable $A$ over $\hilbert[\signature(\h_2)]$,   
    \[
        \hoare{\h_1}{\prog}{\h_2} \implies
        \expectation[A \mid \sem{\prog}(\semRho{\h_1})] = \expectation[A \mid
        \semRho{\h_2}]
    \]
\end{restatable}

The Hoare logic is also adequate with respect to the
denotational semantics on the fragment of the language with no \texttt{while}
loops; equivalently, with only bounded loops that can be fully unfolded.

\begin{restatable}[Adequacy on bounded programs]{theorem}{restateAdequacyBounded}%
    \label{thm:adequacy-bounded}
    For any program $\prog$ and state $\h_1$ such that $\prog$ contains no loops
    and $\exists \signature' \subseteq \addressSet,
    \validProg{\signature(\h)}{\prog}{\signature'}$, and for any CQ state
    $\rho$,
    \[
        \sem{\prog}(\semRho{\h_1}) = \rho \implies \exists \h_2, \hoare{\h_1}{\prog}{\h_2} \land \semRho{\h_2} = \rho
    \]
\end{restatable}

The essential point is that a symbolic representation $\h_2$ always exists for
finitely terminating programs. On the other hand, adequacy is conjectured not to
hold over $\lang$ unrestricted, as the existence of a closed form for the loop
invariant $\h$ for the\ \ref{rule:while} rule is not guaranteed without further
large-scale extensions of \ihps{} we find to be counterproductive for the
intent of this article.  We recall that our goal is not to fully automate the
analysis, which includes undecidable properties such as almost-sure termination
and the computation of expected values~\cite{AMPP24}, but to provide a sound and practical
framework for a semi-automated analysis. In any case, while adequacy is a nice
property to have, soundness is the key property for our purposes, allowing us to
symbolically estimate properties of interest
(\cref{cor:estimation-via-symbolic-execution}).

\section{Heuristics and Application to Resource Analysis}%
\label{sec:generic-strategy}

The analysis of \texttt{while} loops, even symbolically in terms of \ihps{} using
the rule\ \ref{rule:while}, remains challenging as it requires finding a loop
invariant expressed in closed form as an \ihps{} term $\h$ such that
$\hoare{\h[{\intVar}]}{\ifthene{\bool}{\prog}{\skipProg}}{\h[{\intVar}+1/{\intVar}]}$.
By the undecidability of the logic (\cref{thm:undecidability}), this closed-form
invariant may not always exist, and when it does, there is no reasonable
automatic way to compute it in general. To alleviate these issues, we propose a
heuristic, in the form of a new admissible rule of the logic, for finding an
invariant $\h[x]$ of the particular form $\h[{\intVar}] \equaldef
\bigpsadd_{l=0}^{{\intVar}} \left({\h}^{{\lnot \bool}}l/{\intVar}] \otimes
\pastKet[\Z]{l}\right) \psadd \left({\h}^{\bool} \otimes
\pastKet[\Z]{\intVar}\right)$ which is guided by the one-sided branching
structure of the execution of a \texttt{while} loop. The tensoring with
$\pastKet[\Z]{l}$ ensures that branches halting at $l$ do not interfere and are
not in a superposition with states halting at $l' \neq l$.

\begin{restatable}[Heuristic for loops]{theorem}{restateGenericStrategy}%
    \label{thm:generic-strategy}

The following rule is admissible for the logic:
\begin{prooftree}
    \AxiomC{$\equivvdash \h^\bool \equiv \bool * \h^\bool$}
    \AxiomC{$\equivvdash \h^{\lnot \bool} \equiv (\lnot \bool) * \h^\bool$}
    \AxiomC{$\judgement{\logicContext}{}{\hoare{\h^\bool}{\prog}{\h^\bool[\intVar+1/\intVar]
    \oplus \h^{\lnot \bool}[x+1/x]}}$}
    \RightLabel{\labelthis{rule:collect}{\textsc{Collect}}}
    \TrinaryInfC{$\judgement{\logicContext}{}{\hoare{\h^\bool \oplus \h^{\lnot \bool}[0/\intVar]}{\while{\bool}{\prog}}{\sum_x \h^{\lnot b}\otimes \pastKet[\Z]{\intVar}}}$}
\end{prooftree}
\end{restatable}

A limit \ihps{} of this form $\h_{\infty} \equaldef \sum_\intVar \h^{\lnot
\bool} \otimes \pastKet[\Z]{\intVar}$ makes resource analysis significantly
simpler. Indeed, the (sub-)\textbf{probability distribution} of
the termination time is $\probability(\intTerm) = \|\h^{\lnot \bool}[\intTerm/x]\|^2$.
We can then express the estimations we are interested in
(\cref{sub:Context and motivation}) as expectations of observables
(\cref{cor:estimation-via-symbolic-execution}) $\expectation[A \mid \h]$:
\begin{enumerate}
    \item Termination within $\intTerm$ iterations: $\probability(\cInt \leq 
        \intTerm) \equaldef \expectation[\widehat{F_{\cInt \leq {\intTerm}}} \mid \h_{\infty}] =
        \sum_{l=0}^{\intVar} \probability(l)$.
    \item Termination at all: $\probability(\cInt < \infty) \equaldef
        \expectation[I \mid \h_{\infty}] =
    \|\h_\infty\|^2= \sum_{{\intTerm}=0}^\infty \probability({\intTerm})$.
    \item Expected iteration count:
        \[
            \expectation[\cInt \mid \h_{\infty}]  \equaldef
            \textstyle \frac 1 {\probability(\cInt < \infty)} \expectation\left[\textstyle \left(\sum_{\sigma \in \mathrm{Basis}(\signature)}
            \sem{\cInt}_\sigma \ketbra{\sigma}{\sigma}\right) \mid \h_{\infty}\right]
            = \textstyle \frac 1 {\probability(\cInt <
            \infty)} \sum_{{\intTerm}=0}^{\infty} {\intTerm} \probability({\intTerm})
        \]
\end{enumerate}
Beyond these properties, depending on the form of $\h^{\lnot \bool}$, it may be
more or less easy, but nonetheless possible to also extract the probability of
an arbitrary predicate $\bool'$ as $\probability(\bool' \mid \h_{\infty}) =
\|(\bool' * \h^{\lnot \bool})\|^2 $ or the expectation of an arbitrary integer
variable $\texttt{y} \in \signature(\h^{\lnot \bool})$ as
$\expectation[\texttt{y} \mid \h_{\infty}]= \frac 1 {\probability(\texttt{y} <
\infty)} \sum_{{\intTerm}=0}^{\infty} \intTerm \cdot \|(\texttt{y} = \intTerm) *
\h^{\lnot \bool}\|^2$.  In particular, programs can always be modified to
include a counter variable for any resource of interest (e.g., number of
applications of a $T$-gate, weighted gate count, etc.), meaning that the
extraction of the expected value of a program variable extends to an estimation
of a broad class of resources.

\begin{example}%
    \label{running:strategy}

In\ \ref{prog:rus-unitary}, we are able to estimate the resources in terms of
time (the average of $\cInt$) by applying \ref{rule:collect} with
    $
    \h_{\textsc{RUS}}^b \equaldef
    \pasuminline{0}{|\boolVar \rangle_\qbit{[1]}_\cbit\classKet[\cInt]{\intVar}}{\frac 1
    {2^{\intVar}}}
    $ and $
    \h_{\textsc{RUS}}^{\lnot b} \equaldef
    \sum_{\boolVar_1} \pasuminline{\frac{c_1}{4}}{\ket[\qubit]{\boolVar \oplus
    \boolVar_1}{[1]}_\cbit\classKet[\cInt]{\intVar}}{\frac 1
{2^{\intVar}} \cdot \sqrt{2}^{\boolVar_1} \cdot {\lift(\intVar >0)}}
$ to obtain the limiting state
\[
    \h_{\textsc{RUS-}\infty} =
    \bigpsadd_{x}\sum_{\boolVar_1} \pasum{\frac{c_1}{4}}{\ket[\qubit]{\boolVar
        \oplus \boolVar_1}{[1]}_\cbit
        {[{\intVar}]}_\cInt}{\frac 1 {2^{\intVar}} \cdot \sqrt{2}^{\boolVar_1}
    \cdot {\lift(\intVar >0)}} 
    \]
which then allows us to deduce that the probability of the program halting
within ${\intTerm}$ iterations is $\probability({\leq \intTerm}) =
\sum_{j=1}^{\intTerm} {\frac 3 {4^j}} = 1 - \frac 1 {4^{\intTerm}}$, that
the program terminates almost surely ($\lim \probability({\leq \intTerm}) = 1$),
and that the expected termination time is $\expectation[\cInt] =
\sum_{{\intVar}=1}^{\infty} {\intVar} \cdot \probability({\intVar}) = \frac 4 3$
iterations.

\ref{rule:collect} can also be applied to \ref{cointoss} with  $
    \h_{\textsc{CT}}^b \equaldef
    \pasuminline{0}{|0\rangle_\qbit{[0]}_\cbit\classKet[\cInt]{\intVar}}{1/
    {\sqrt{2^{\intVar}}}}
    $ and $
    \h_{\textsc{CT}}^{\lnot b} \equaldef
    \pasuminline{0}{|1\rangle_\qbit{[1]}_\cbit\classKet[\cInt]{\intVar}}{ 1/
    {\sqrt{2^{\intVar}}}\cdot \lift(\intVar >0)}
$ to obtain an average of $\expectation[\cInt] = 2$ iterations.

\end{example}

\ref{rule:collect} is, of course, not restricted to simple examples such as
\ \ref{cointoss}. In general, it is easy to apply to the resource analysis of
the repeat-until-success schema which are ubiquitous to most quantum programs
due to their probabilistic nature (see \cref{sec:rus}). We can also use
this strategy for analyzing more complex nested \texttt{while} loops where, for
example, the quantum states measured in the outer loop may depend on the number
of iterations of the inner loop (see \cref{sec:nested-while}).

\section{Applications}\label{sec:applications}

In this section, we illustrate the different features of the framework through a
number of examples. These include axiomatic
reasoning by using non-empty $\logicContext$ (\ref{repeatedQBF},
\ref{weakMeasurement}), the observability of almost-sure-termination and the
possibility to analyze non-almost-surely-terminating programs
(\ref{weakMeasurement}), as well as a stress-test (\ref{nestedWhile}) of the
framework highlighting its ability to handle complex nested \texttt{while} loops
with subtle forms of classical control. We also revisit in detail the unitary
synthesis example \ref{prog:rus-unitary} of \cref{sub:Overview}.

In the examples below, specifically when the first iteration of a \texttt{while}
is certain to be performed, it is often convenient to sum over \emph{non-zero}
integers. It is also practical to group summation quantifiers over multiple
variables together.  For the sake of readability, we will therefore write
$\sum_{\intVar > 0} \h$ to mean $\sum_{\intVar} \lift(\intVar \neq 0) \cdot \h$,
and $\sum_{\var_1, \ldots, \var_n} \h$ to mean $\sum_{\var_1} \cdots
\sum_{\var_n} \h$.

\subsection{Repeat-Until-Success}\label{sec:rus}

Our first three examples are instances of the repeat-until-success (RUS) pattern
(\cref{sub:Overview}) where
a certain operation is performed ending in a measurement indicating
whether it succeeded to produce the desired effect or not, in which case, it is
repeated. In this section, as common in the literature, `repeat-until-success'
refers to the case when the initial state is recoverable after failure and when
the success probability is fixed.  Our framework is applicable generically on
all such patterns: the only dependence of the invariant on the number of
iterations is through the norm term as $p^\intVar$; therefore, the core of the
analysis is fundamentally unchanged across the different instances of RUS.

\subsubsection{Unitary synthesis}\label{sec:rus-unitary}

\phantom{a}

\begin{wrapfigure}[13]{r}{0.48\textwidth}
    \centering
    \vskip -1em
    \begin{tikzpicture}
        \node[inner sep = 0pt, scale=0.6] {
                \begin{quantikz}[wire types={q,q}]
                    \lstick{\ket[\qubit]{\psi}} &&& \targ{} &&\targ{} & &&&
                    \rstick{$U^\boolVar \ket[\qubit]{\psi}$}\\
                    \lstick{\ket[\qubit_1]{0}} &  \gate{H} & \gate{T} &
                    \ctrl{-1} & \gate{H} & \ctrl{-1} & \gate{T} & \gate{H} &
                \meter{} & \cw \rstick{$\boolVar$}
            \end{quantikz}};
        \end{tikzpicture}
        \vskip 1em
\begin{minipage}{0.4\textwidth}
    \centering
    \begin{lstlisting}[caption={The \ref{prog:rus-unitary} program.},
    label={lst:rus-unitary-again}]
qubit q$_1$; bit c;
X(q$_1$); c := 1;
do
    X(q$_1$); H(q$_1$); T(q$_1$);
    CNOT(q$_1$,q); H(q$_1$); CNOT(q$_1$,q);
    T(q$_1$); H(q$_1$);
    c := measure(q$_1$);
while c
    \end{lstlisting}
\end{minipage}
\end{wrapfigure}
In this section, we elaborate on the example from \cref{sub:Overview} of
repeat-until-success unitary synthesis of $U= \frac {1}{\sqrt 3} (I + i \sqrt{2}
X)$ over a qubit $\qubit$. Its code and circuit are reproduced in 
\cref{lst:rus-unitary-again}.


In \cref{fig:rus-unitary-derivation}, we derive the
path-sum Hoare triple for the first iteration of the loop, starting from the
generic basis state $\pasum{0}{\ket[\qubit]{\boolVar}}{1}$ on the qubit $\qubit$. For
the sake of readability, we do not write $\hoare{\h}{\prog}{\h'}$,
instead write an alternation of path-sums with either a program or a 
rule of the equational theory. 
This is to be interpreted as such: the first line
is an \hps{} $\h_1$, it defines a currently derived triple
$\hoare{\h_1}{\skipProg}{\h_1}$ which is updated for each subsequent line as
follows: if the line is $\xrightarrow{\prog}{\h_2}$, we apply \ref{rule:seq} of
the current triple $\hoare{\h_1}{\prog'}{\h'}$ with $\hoare{\h'}{\prog}{\h_2}$
to get a new current triple $\hoare{\h_1}{\prog';\prog}{\h_2}$. If instead the
line is $\xRightarrow{\mathrm{Rule}} \h_2$ then we apply \ref{rule:equiv} to the
current triple $\hoare{\h_1}{\prog'}{\h'}$ and $\h' \equiv \h_2$ to update the
current derived triple to $\hoare{\h_1}{\prog'}{\h_2}$, where the equivalence is
derived from the equational theory with the rule named `Rule'. In particular, on
lines 11 and 14, we perform changes of variables, line 12 uses
\ref{rule:phase-bisector}, line 13 expands the sum over $\boolVar_3$, line 15
eliminates the null case of $\boolVar_1 = 0$, and finally, line 16 eliminates
(unobservable) global phases. We also note that between line 12 and line 13, an
implicit rewriting of the cosine term $n$ was performed.

\begin{figure}[htpb]
    \centering
    \footnotesize
\[
\begin{array}{r|r\sumpasummemcols{l@{}l@{}l}}
    1 & \textcolor{gray}{\mathrm{Init}}                                       & \sumpasumsplit{0}{\ket[\qubit]{\boolVar}
		\ket[\qubit_1]{1}\classKet[\cbit]{1}}{1}
	                                                                      &                                                            \\
        2 & \xrightarrow{X(\qubit_1)}                                             & \sumpasumsplit{0}{\ket[\qubit]{\boolVar}
	\ket[\qubit_1]{0}\classKet[\cbit]{1}}{1}                                                                                           \\
            3 & \xrightarrow{H(\qubit_1)}                                             & \sum_{\boolVar_1} \sumpasumsplit{0}{\ket[\qubit]{\boolVar}
	\ket[\qubit_1]{\boolVar_1}\classKet[\cbit]{1}}{\frac{1}{\sqrt{2}}}                                                                 \\
                4 & \xrightarrow{T(\qubit_1)}                                             & \sum_{\boolVar_1}
	\sumpasumsplit{\frac{\boolVar_1}{8}}{\ket[\qubit]{\boolVar}
	\ket[\qubit_1]{\boolVar_1}\classKet[\cbit]{1}}{\frac{1}{\sqrt{2}}}                                                                 \\
                5 & \xrightarrow{CNOT(\qubit_1, \qubit)}                                  & \sum_{\boolVar_1}
	\sumpasumsplit{\frac{\boolVar_1}{8}}{\ket[\qubit]{\boolVar \oplus \boolVar_1}
	\ket[\qubit_1]{\boolVar_1}\classKet[\cbit]{1}}{\frac{1}{\sqrt{2}}}                                                                 \\
                6 & \xrightarrow{H(\qubit_1)}                                             &
	\sum_{\boolVar_2, \boolVar_1}
	\sumpasumsplit{\frac{\boolVar_1}{8} + \frac{c_1 c_2}{2}}{\ket[\qubit]{\boolVar
	\oplus \boolVar_1} \ket[\qubit_1]{\boolVar_2}\classKet[\cbit]{1}}{\frac{1}{2}}                                                     \\
                7 & \xrightarrow{CNOT(\qubit_1, \qubit)}                                  &
	\sum_{\boolVar_2, \boolVar_1}
	\sumpasumsplit{\frac{\boolVar_1}{8} + \frac{c_1 c_2}{2}}{\ket[\qubit]{\boolVar
			\oplus \boolVar_1 \oplus \boolVar_2}
	\ket[\qubit_1]{\boolVar_2}\classKet[\cbit]{1}}{\frac{1}{2}}                                                                        \\
                8 & \xrightarrow{T(\qubit_1)}                                             &
	\sum_{\boolVar_2, \boolVar_1}
	\sumpasumsplit{\frac{\boolVar_1}{8} + \frac{c_1 c_2}{2} +
		\frac{\boolVar_2}{8}}{\ket[\qubit]{\boolVar \oplus \boolVar_1 \oplus
	\boolVar_2} \ket[\qubit_1]{\boolVar_2}\classKet[\cbit]{1}}{\frac{1}{2}}                                                            \\
                9 & \xrightarrow{H(\qubit_1)}                                             &
	\sum_{\boolVar_3, \boolVar_2, \boolVar_1}
	\sumpasumsplit{\frac{\boolVar_1}{8} + \frac{c_1 c_2}{2} + \frac{\boolVar_2}{8} +
		\frac{c_2 c_3}{2}}{\ket[\qubit]{\boolVar \oplus \boolVar_1 \oplus
	\boolVar_2} \ket[\qubit_1]{\boolVar_3}\classKet[\cbit]{1}}{\frac{1}{\sqrt{8}}}                                                     \\
                10 & \xrightarrow{\mathrm{Measure}}                                        &
	\sum_{\boolVar_3, \boolVar_2, \boolVar_1}
	\sumpasumsplit{\frac{\boolVar_1}{8} + \frac{c_1 c_2}{2} + \frac{\boolVar_2}{8} +
		\frac{c_2 c_3}{2}}{\ket[\qubit]{\boolVar \oplus \boolVar_1 \oplus
	\boolVar_2}\ket[\qubit_1]{c_3}\classKet[\cbit]{c_3}}{\frac{1}{\sqrt{8}}}                                                           \\
	\noalign{\vskip 0.6ex}
	\hdashline
	\noalign{\vskip 0.6ex}
                11 & \xRightarrow{\mathrm{CV}(\boolVar_1 := \boolVar_1 \oplus \boolVar_2)} &
	\sum_{\boolVar_3, \boolVar_2, \boolVar_1}
	\sumpasumsplit{\frac{\boolVar_1}{8}
		+ \frac{\boolVar_1 \boolVar_2}{4} +
		\frac{3 \boolVar_2}{4} + \frac{\boolVar_2 \boolVar_3}{2}
	}{\ket[\qubit]{\boolVar \oplus \boolVar_1}
	\ket[\qubit_1]{c_3}\classKet[\cbit]{c_3}}{\frac{1}{\sqrt{8}}}                                                                      \\
    12 & \xRightarrow{\mathrm{Phase-Bisector}(\boolVar_2)}                     &
	\sum_{\boolVar_3,\boolVar_1} \sumpasumsplit{
		\frac{\boolVar_1}{4} + \frac{3}{8} + \frac{\boolVar_3}{4}
	}{\ket[\qubit]{\boolVar \oplus \boolVar_1}\ket[\qubit_1]{c_3}\classKet[\cbit]{c_3}}
	{\frac{1}{\sqrt{8}} \cdot n}                                                                                                       \\
    13 & \xRightarrow{\mathrm{SplitClass}(\boolVar_3)}                         &
	\sum_{\boolVar_1} \sumpasumsplit{\frac{3}{8} +
		\frac{1}{2}+\frac{\boolVar_1}{4}}{\ket[\qubit]{\boolVar\oplus\boolVar_1}\ket[\qubit_1]{0}\classKet[\cbit]{0}}{\frac{1}{2}
		\sqrt{2}^{\boolVar_1}}
	\\
      & \oplus                                                                &
	\sum_{\boolVar_1} \sumpasumsplit{\frac{5}{8} + \frac{1}{2} +
		\frac{\boolVar_1}{4}}{\ket[\qubit]{\boolVar\oplus\boolVar_1}\ket[\qubit_1]{1}\classKet[\cbit]{1}}{\frac{1}{2}
		\cdot (1 \oplus \boolVar_1)}
	\\
    14 & \xRightarrow{\mathrm{CV}(\boolVar_1 := 1 \oplus \boolVar_1)}          &
	\sum_{\boolVar_1} \sumpasumsplit{\frac{3}{8} +
		\frac{1}{2}+\frac{\boolVar_1}{4}}{\ket[\qubit]{\boolVar\oplus\boolVar_1}\ket[\qubit_1]{0}\classKet[\cbit]{0}}{\frac{1}{2}
		\sqrt{2}^{\boolVar_1}}
	\\
      & \oplus                                                                &
	\sum_{\boolVar_1} \sumpasumsplit{\frac{5}{8} + \frac{1}{2} +
		\frac{1 \oplus
			\boolVar_1}{4}}{\ket[\qubit]{\boolVar\oplus\boolVar_1\oplus 1}\ket[\qubit_1]{1}\classKet[\cbit]{1}}{\frac{1}{2}
		\cdot \boolVar_1}
	\\
    15 & \xRightarrow{\mathrm{Filter}(\boolVar_1)}                             &
	\sum_{\boolVar_1} \sumpasumsplit{\frac{3}{8} +
		\frac{1}{2}+\frac{\boolVar_1}{4}}{\ket[\qubit]{\boolVar\oplus
			\boolVar_1}\ket[\qubit_1]{0}\classKet[\cbit]{0}}{\frac{1}{2}
		\sqrt{2}^{\boolVar_1}}
	\\
      & \oplus                                                                &
	\sumpasumsplit{\frac{5}{8} +
		\frac{1}{2}}{\ket[\qubit]{\boolVar}\ket[\qubit_1]{1}\classKet[\cbit]{1}}{\frac{1}{2}}
	\\
    16 & \xRightarrow{\mathrm{Phase-Elim}}                                     &
	\sum_{\boolVar_1}
	\sumpasumsplit{\frac{\boolVar_1}{4}}{\ket[\qubit]{\boolVar\oplus \boolVar_1}\ket[\qubit_1]{0}\classKet[\cbit]{0}}{\frac{1}{2}
		\sqrt{2}^{\boolVar_1}}

	\\
                                                                          & \oplus                                                                &
	\sumpasumsplit{0}{\ket[\qubit]{\boolVar}\ket[\qubit_1]{1}\classKet[\cbit]{1}}{\frac{1}{2}}
	\\
\end{array}\]
where $n = 2 \cos\left(2\pi \frac{3 + \boolVar_1 + 2 \boolVar_3}{8}\right)$
    \caption{Derivation of the first iteration of the loop in
    \cref{lst:rus-unitary}.}%
    \label{fig:rus-unitary-derivation}
\end{figure}
We can then pick, for the sake of applying \cref{thm:generic-strategy},
$\h^\bool_{\mathrm{RUS}} \equaldef
\pasum{0}{\ket[\qubit]{\boolVar}\ket[\qubit_1]{1}\classKet[\cbit]{1}}{\frac{1}{2^{\intVar}}}$
and $\h^{\lnot \bool}_{\mathrm{RUS}} \equaldef \lift(\intVar \neq 0) \cdot \sum_{\boolVar_1}
\pasum{\frac{\boolVar_1}{4}}{\ket[\qubit]{\boolVar\oplus
    \boolVar_1}\ket[\qubit_1]{0}\classKet[\cbit]{0}}{\frac{1}{2^{\intVar}}
\sqrt{2}^\boolVar_1}$. It is relatively direct to observe that the derivation in
\cref{fig:rus-unitary-derivation} can be adapted nearly verbatim to show that
$\hoare{\h^\bool_{\mathrm{RUS}}}{\prog}{\left(\h^\bool_{\mathrm{RUS}} \oplus \h^{\lnot
\bool}_{\mathrm{RUS}}\right)[\intVar + 1 / \intVar]}$, with the only difference
being that the norm terms in all \ihps{} that appear in the derivation should
now be multiplied by $\frac{1}{2^{\intVar}}$, and the terms where $\qubit$
contains $0$ must also be multiplied by $\lift(\intVar \neq 0)$.

Finally, by applying \cref{thm:generic-strategy} (and\ \ref{rule:seq}), a Hoare
triple for the program as a whole can be derived describing clearly that the
program will implement the unitary $U$ on $\qubit$ with probability 1, and that
the distribution of the number of iterations before success is $\probability(\intVar)
= \frac{1}{2^{\intVar}} \cdot \frac34$.
\[
    \hoare{\pasum{0}{\ket[\qubit]{\boolVar}}{1}}{\ref{prog:rus-unitary}}{
    \sum_{\intVar > 0} \sum_{\boolVar_1}
\pasum{\frac{\boolVar_1}{4}}{\ket[\qubit]{\boolVar\oplus
    \boolVar_1}\ket[\qubit_1]{0}\classKet[\cbit]{0}\pastKet[\Z]{\intVar}}{\frac{1}{2^{\intVar}}
\sqrt{2}^{\boolVar_1}}}\]

\ifarXiv{\newpage}\fi

\subsubsection{Quantum Bernoulli Factory (QBF)}\label{sec:rus-bernoulli}

\phantom{as}

\begin{wrapfigure}[8]{r}{0.35\textwidth}
    \centering
	\begin{minipage}{0.27\textwidth}
    \vspace{-1em}
        \begin{lstlisting}[caption={The
            \labelthis{repeatedQBF}{\textsc{QBF}} program.},
                label={lst:rus-bernoulli}]
qubit $\qbit$; bit $\cbit$; int $\cInt$;
while ($\lnot \cbit$) do
    U$_p$($\qbit$);
    $\cbit$ := measure $\qbit$;
    x := x + 1
done
            \end{lstlisting}
	\end{minipage}
\end{wrapfigure}

A Quantum Bernoulli Factory (QBF) is a quantum circuit that can produce a random
bit with some probability $p$, given access to a unitary of the form $U_p =
\sqrt{p} I + \sqrt{1-p} X$. It is a prototypical example of a quantum
\texttt{while} program of interest for resource estimation~\cite{LZB+25}.  QBF
is similar to the \ref{cointoss} example, except that the unitary $U_p$ does not
actually belong to the Clifford+$R_k$ set of primitives of $\lang{}$ (see
\cref{sec:language-syntax}). However, we can define its behavior with a family
of axioms $\logicContext$ as follows:

\begin{center}
    $\logicContext_{\mathrm{QBF}} \equaldef \left\{\hoare{\h}{U_p}{\sqrt{p} \h + \sqrt{1-p}
    \h[\qubit:=\lnot{\qubit}]} \mid \h \in \hpsType\right\}$
\end{center}
where $\h[\qubit:=\lnot{\qubit}]$ is the path-sum $\h$ where the expression in
$\qubit$ is negated. 

The heuristic \ref{rule:collect} from \cref{sec:generic-strategy} can then be
applied with $ \h_{\textsc{QBF}}^b \equaldef
\pasuminline{0}{|0\rangle_\qbit{[0]}_\cbit}{\sqrt{p^{\intVar}}} $ and $
\h_{\textsc{QBF}}^{\lnot b} \equaldef \lift(\intVar \neq 0) \cdot
\pasuminline{0}{|1\rangle_\qbit{[1]}_\cbit}{\sqrt{p^{\intVar}} \sqrt{1-p}}$ to
conclude with a final derivation:
\[
    \logicContext_{\mathrm{QBF}} \vdash \hoare{\pasum{0}{|0\rangle_\qbit}{1}}{\ref{repeatedQBF}}{
    \sum_{\intVar > 0} \pasum{0}{|1\rangle_\qbit{[1]}_\cbit\pastKet[\Z]{\intVar}}{\sqrt{p^{\intVar}} \sqrt{1-p}}}
\]

\subsubsection{Weak measurements}\label{sec:weak-measurements}

Another application is to repeat the so-called \emph{weak
$\kappa$-measurement}~\cite{AMH22} which trades a decrease in the success
probability of a quantum measurement from $p$ to $\kappa p$ for the possibility
not to lose the state entirely when measurement fails. In the case of failure,
instead of losing all the amplitude of the desired state, the weak measurement
only drops it by a factor of $\sqrt{1-\kappa}$. In this example, there is
usually a non-zero probability of divergence; in fact, this is precisely the
probability that the (standard) measurement fails. Using our framework, we show
that $\kappa$ is a genuine measure of strength: conditional upon success, the
expected number of measurements required to reach it is $\frac 1 {\kappa}$.

This is also another occasion to illustrate axiomatic analysis in the framework.
Instead of implementing a specific state preparation algorithm and the weak
measurement algorithm concretely, we can instead define them as the oracles
\texttt{Prepare} which prepares a state $\h_0$ on a signature $\signature_0$ to
be measured and \texttt{WeakMeas-$\kappa$} which performs the weak measurement
storing the result in a bit $\cbit$. We can use these oracles by giving them
the following Hoare-logic specifications as axioms.
\[
    \logicContext_{\textsc{WM}} \equaldef \left\{
        \begin{array}{lll}
            \{\h\} & \texttt{Prepare} & \{\h \otimes \h_0\} \mid \h \in \ihps,\\
            \{\h\}&\texttt{WeakMeas-}\kappa&\{\sqrt{\kappa} \cdot
        Q * \h_{\top} + \sqrt{1-\kappa} \cdot Q * \h_{\bot} + (\lnot Q) *
\h_{\bot}\} \mid \h \in \ihps{}
            \end{array}
    \right\}
\]

Here, the semantics of \texttt{Prepare} is straightforward: it simply ignores
what's in $\h$, and prepares a new state $\h_0$ along-side it. As for
\texttt{WeakMeas-$\kappa$}, it results in two branches: either the measurement
succeeds, in which case we project the state $\h$ onto the subspace satisfying
the predicate $Q$ (denoted $Q * -$)\footnote{Assuming $Q$ is expressed as a
boolean expression}, or it fails, in which case we keep the
failing part of the state ($\lnot Q$) and reduce the amplitude of the success
part by $\sqrt{1-\kappa}$\footnote{For the article, $\kappa$ is assumed to be a
constructible number. In the implementation $\libname$, $\kappa$ can be a
formal variable.}.  The result of the measurement is stored in $\cbit$
(specifically, $\h_{\bot}$ is $\h$ marked with failure and $\h_{\top}$ is $\h$
marked with success; i.e., $\models \hoare{\h}{\assign{\cbit}{0}}{\h_{\bot}}$,
and $\models \hoare{\h}{\assign{\cbit}{1}}{\h_{\top}}$).

\begin{wrapfigure}[9]{r}{0.33\textwidth}
    \vspace{-1em}
    \centering
	\begin{minipage}{0.25\textwidth}
        \begin{lstlisting}[label=lst:weak-measurement, caption={The repeated
        \labelthis{weakMeasurement}{\textsc{WeakMeas}} program.}]
Prepare;
int x; bit c;
while ($\lnot$c) do
    WeakMeas-$\kappa$;
    x := x + 1;
done
        \end{lstlisting}
	\end{minipage}
\end{wrapfigure}
With these oracles constructed, we now write a program
(\cref{lst:weak-measurement}) designed to express
the (conditional) average number of weak measurements necessary for a
positive measurement.

Given an initial state $\h_0$, and a classical predicate $Q$\footnote{assumed to
be expressible as a boolean expression}, we
can split $\h_0$ into orthogonal parts $\h_{\mathrm{good}}$ and
$\h_{\mathrm{bad}}$
satisfying or not the predicate $Q$ respectively. The invariant of weak
measurement is then given by $\h^b = \h_{\mathrm{bad}} \otimes \classKet[\cbit]{0}
\otimes {[\intVar]}_\cInt + \sqrt{{(1-\kappa)}^\intVar} \; \h_{\mathrm{good}} \otimes
\classKet[\cbit]{0} \otimes {[\intVar]}_\cInt$ and $\h^{\lnot b} =
\sqrt{{(1-\kappa)}^{\intVar-1} \kappa} \; \h_{\mathrm{good}} \otimes \classKet[\cbit]{1}
\otimes {[\intVar]}_\cInt$ with the limit state being $\h_\infty =
\sum_{\intVar>0}
\sqrt{{(1-\kappa)}^{\intVar-1} \kappa} \; \h_{\mathrm{good}} \otimes
\classKet[\cbit]{1} {[\intVar]}_\cInt$; that is, we can derive the triple:
\[
    \hoare{\pasum{0}{\emptyset}{1}}{\ref{weakMeasurement}}{
    \h_\infty =
        \sum_{\intVar>0} \sqrt{{(1-\kappa)}^{\intVar-1} \kappa} \;
    \h_{\mathrm{good}} \otimes \classKet[\cbit]{1} {[\intVar]}_\cInt}
\]

In conclusion, \ref{weakMeasurement} succeeds with
probability $\sum_{\intVar=1}^{\infty} {(1 - \kappa)}^{\intVar-1} \kappa
|\h_{\mathrm{good}}|^2 = |\h_{\mathrm{good}}|^2$. Moreover, conditional on
success, the expected number of weak measurements performed before halting is
$\frac 1 {\kappa}$. That is, $\kappa$ does indeed express the strength of the
measurement: $\kappa = 1$ is a strong measurement, and as $\kappa$ decreases,
the measurement is buffered over more and more iterations.


\subsection{Nested While Loops}\label{sec:nested-while}

\begin{wrapfigure}[17]{r}{0.39\textwidth}
    \centering
    \vspace{-1em}
    \begin{minipage}{0.30\textwidth}
        \begin{lstlisting}[caption={\labelthis{nestedWhile}{\textsc{Nested}}
        while loops.},
                label={lst:nested-while}]
bit c1; bit c2;
int x1; int x2;
qubit q1; qubit q2;
while($\lnot$c2) do
    H(q2); c1 := 0;
    while($\lnot$c1) do
        H(q1);
        c1 := measure q1;
        R(q2, r);
        x1 := x1 + 1;
    done;
    H(q2);
    c2 := measure q2;
    x2 := x2 + 1; x1 := 0;
done
        \end{lstlisting}
    \end{minipage}
\end{wrapfigure}
Using \ihps{}, it is possible to analyze the
behavior of complex \texttt{while} loops symbolically without necessarily having
to calculate difficult limits over the reals or complex numbers. In fact, as
long as a sequence of complex numbers ${(n[\intVar] e^{2\pi i
p[\intVar]})}_{\intVar \in \N}$ can be expressed in closed form using the syntax
of \normType{} and \phaseType{}, its limit can be expressed symbolically
by injecting it into an \ihps{} $\lim_\intVar \pasum{p}{\emptyset}{n}$ with
empty memory.

Consider, for example, the nested loop in the~\ref{nestedWhile} program in
\cref{lst:nested-while}. Let \labelthis{innerLoop}{\textsc{InnerLoop}} be the
section of the program between lines 6 and 11, and
\labelthis{outerLoop}{\textsc{OuterLoop}} be the section between lines 4 and 15.
This program is designed specifically to stress test the capabilities of our symbolic
representation in the context of nested, communicating \texttt{while} loops.
The\ \ref{innerLoop} program is very similar to the\ \ref{cointoss} program we
have been using as a running example, which tosses qubit $\qbit_1$, except that
it also applies a rotation (line 9) $\texttt{R(qi, r)}$ to a different qubit
$\qbit_2$ which is in a Hadamard basis state $\ket +$ or $\ket -$ at each
iteration of the loop. As such, the final phase applied to $\qbit_2$ after
exiting the\ \ref{innerLoop} depends on the number of iterations
$k$ of the inner loop before exiting the\ \ref{innerLoop}. This, in turn,
influences the probability distribution in the measurement of the second qubit
$\qbit_2$ in the~\ref{outerLoop} program. In fact, each iteration of
the~\ref{outerLoop} roughly corresponds to applying $H \cdot {R(r)}^k \cdot H$
to $\qbit_2$ before measuring it, thereby making the probability of measuring
$0$ or $1$ in $\cbit_2$ dependent on $k$.

By a analysis similar to that of\ \ref{cointoss} applied this time
to~\ref{innerLoop}, starting at line 6, from a state
$\h_{\mathtt{inner},0} \equaldef \sum_\boolVar \pasuminline{0}{\ket[\qbit_1]{0}
\ket[\qbit_2]{\boolVar} {[0]}_{\cbit_1} {[0]}_{\cbit_2}  {[0]}_{\cInt_1}
{[0]}_{\cInt_2}}{1},$ we reach, by line 11, the state
$\h_{\mathtt{inner},\infty} \equaldef \sum_\intVar \sum_\boolVar
\pasuminline{\frac {\intVar \boolVar} {2^{r}}}{\ket[\qbit_1]{1}
    \ket[\qbit_2]{\boolVar} {[1]}_{\cbit_1} {[0]}_{\cbit_2}
{[\intVar]}_{\cInt_1} {[0]}_{\cInt_2}}{\frac 1 {\sqrt{2^{\intVar+1}}}}.$ Then,
by line 14, we have reached the state $\h_{14} \equaldef \sum_{\boolVar'}
\sum_\intVar \sum_\boolVar \pasuminline{\frac {\intVar\boolVar} {2^{r}} + \frac
    {\boolVar\boolVar'}{2} }{\ket[\qbit_1]{1} \ket[\qbit_2]{\boolVar'}
    {[1]}_{\cbit_1} {[\boolVar']}_{\cbit_2} {[0]}_{\cInt_1}
{[1]}_{\cInt_2}{\pastKet{\intVar}}_{\Z}}{\frac 1 {\sqrt{2^{\intVar+1}}}}.$ with
$\boolVar'$ being introduced by an application of $H$ at line 12
and $\intVar$ being moved out of $\cInt_1$ and into the history
$\pastKet[\Z]{\intVar}$ by the
resetting of $\cInt_1$ to $0$ at line 14, and the fact that the first
iteration of the~\ref{outerLoop} has been performed being marked by 
incrementing $\cInt_2$ by $1$ at line 14 from $0$ to $1$.

At this point, we attempt to apply the \ref{rule:collect} strategy for analyzing
the~\ref{outerLoop} program by separating $\h_{14}$ into two parts
$\h^\bool[1/\intVar']$ and $\h^{\lnot \bool}[1/\intVar']$, according to the
values of $\boolVar'$. However, are struck by an issue: each iteration of
the~\ref{outerLoop} appears to introduce a new integer path variable $\intVar$,
meaning that the size of the \ihps{} itself depends on the number of iterations
of the~\ref{outerLoop}. To resolve that, we separate the probabilistic analysis
of the inner loop from that of the outer loop by factoring $\h_{14}$:
\[
    \h_{14} \equiv \sum_{\boolVar'}\left(\sum_{\boolVar} \sum_{k} \pasum{\frac {\intVar\boolVar}
        {2^{r}} + \frac {\boolVar\boolVar'}{2}}{{\pastKet{k}}_{\Z}}{\frac 1 {\sqrt{2^{\intVar+1}}}}
        \otimes \pasum{0}{\ket[\qbit_1]{1} \ket[\qbit_2]{\boolVar'} {[1]}_{\cbit_1}
    {[\boolVar']}_{\cbit_2} {[0]}_{\cInt_1} {[1]}_{\cInt_2}}{1} \right)
\]
Then, we define two \ihps{} $\alpha$ and $\beta$ with null signatures (i.e.\
scalars):
\[
    \alpha \equaldef \sum_{\boolVar} \sum_{\intVar} \pasum{\frac
    {\intVar\boolVar} {2^{r}} }{\pastKet[\Z]{\intVar}}{\frac 1
    {\sqrt{2^{\intVar+1}}}} \quad \mathrm{and} \quad \beta \equaldef
    \sum_{\boolVar} \sum_{\intVar} \pasum{\frac {\intVar\boolVar} {2^{r}} +
    \frac {\boolVar}{2} }{\pastKet[\Z]{\intVar}}{\frac 1 {\sqrt{2^{\intVar+1}}}}
\]
Indeed, since $\signature(\alpha) = \signature(\beta) = \emptyset$, both $\alpha$ and
$\beta$ are interpreted in CQ states as no more than the scalar probability of
obtaining the measurement outcomes $0$ and $1$ in $\cbit_2$ respectively:
\[
    \semRho{\alpha} = \sum_{\intVar} \frac 1 {2^{\intVar+1}} \left| \sum_\boolVar
    e^{2\pi i \frac {\intVar\boolVar} {2^{r}}} \right|^2 \quad \mathrm{and} \quad
    \semRho{\beta} = \sum_{\intVar} \frac 1 {2^{\intVar+1}} \left| \sum_\boolVar
        e^{2\pi i (\frac {\intVar\boolVar} {2^{r}} + \frac {\boolVar}{2})} \right|^2
\]
Those probabilities are highly non-trivial, and yet, we can represent them
symbolically and work however we wish with them without ever having to
explicitly compute them as concrete real numbers.
In any case, this allows us to
express the loop invariant of the~\ref{outerLoop} program as parametrized by the
number of iterations $\intVar'$ as such:
\begin{align*}
    \h_{\mathtt{outer}}^{\bool}[\intVar'] &\equaldef \alpha \beta^{\intVar'} \otimes
    \pasum{0}{\ket[\qbit_1]{1} \ket[\qbit_2]{1} {[1]}_{\cbit_1}
    {[1]}_{\cbit_2} {[0]}_{\cInt_1} {[\intVar'+1]}_{\cInt_2}}{1} \\
    \h_{\mathtt{outer}}^{\lnot \bool}[\intVar'] &\equaldef \beta^{\intVar'+1} \otimes
    \pasum{0}{\ket[\qbit_1]{1} \ket[\qbit_2]{0} {[1]}_{\cbit_1}
    {[0]}_{\cbit_2} {[0]}_{\cInt_1} {[\intVar'+1]}_{\cInt_2}}{1}
\end{align*}
where \[
    \alpha^{\intVar'} \equaldef \bigotimes_{\intVar''} \left(\pasum{0}{\emptyset}{\lift(l \leq \intVar')}
    \otimes \alpha + \pasum{0}{\emptyset}{\lift(1 \oplus (\intVar'' \leq \intVar'))}\right)
\]
By applying \cref{thm:generic-strategy}, we reach:
\[
    \models \hoare{\pasum{0}{\emptyset}{1}}{\ref{nestedWhile}}{\sum_{\intVar'}
    \left(\alpha^{\intVar'} \otimes \pasum{0}{\ket[\qbit_1]{1} \ket[\qbit_2]{1}
{[1]}_{\cbit_1} {[1]}_{\cbit_2} {[0]}_{\cInt_1}
{[\intVar']}_{\cInt_2}}{1}\right)}
\]
which allows us to deduce the probability distribution of the number of
iterations of the~\ref{outerLoop}, despite the nesting and the communication
of the two loops, as $\probability(x_2 = \intVar') = \|\alpha\|^2 \cdot
\|\beta^{\intVar'}\|^2$.
Once again, we can also extract the probability of termination as
$\probability(\mathrm{termination}) = \|\alpha\|^2 \cdot \sum_{\intVar' \in \N}
\|\beta\|^{2\intVar'}$. It is then clear that this is a geometric series of
ratio $\|\beta\|^2<1$; therefore, it can be rewritten as
\[\probability(\mathrm{termination}) = \frac {\|\alpha\|^2} {1 - \|\beta\|^2} = 1\]

While this example is admittedly ad-hoc, it is designed to stress test the
framework so as to illustrate what can be done by the \ihps{} symbolic
representation. Specifically, it illustrates how symbolic execution can be
performed and composed in the context of nested \texttt{while} loops, all
without requiring the computation of limits of sequences over real or complex
numbers.

\ifarXiv{\newpage}\fi
\section{Implementation}%
\label{sec:implementation}

\begin{wrapfigure}[12]{r}{0.60\textwidth}
    \centering
    \vspace{-1em}%
    \begin{minipage}{0.5\textwidth}
        \begin{lstlisting}[caption={Excerpt of \texttt{coin-toss.hyq} containing
        the invariant
        annotation},label={lst:annotation},basicstyle=\ttfamily\small]
while 
{x : Int,
($\Sigma$_{y $\in \mathbb N$} ($\langle $0, liftC($\uparrow$((y $\leq$ x))) * liftC(1)/sqrt(liftC(2^(y+1))) $\cdot$ |1$\rangle $_q[1]_{c : $\mathbb B$}[y]_{x : $\mathbb Z$}$\rangle $)) 
+ ($\langle $0, liftC(1)/sqrt(liftC(2^(x+1))) $\cdot$ |0$\rangle $_q[0]_{c : $\mathbb B$}[x]_{x : $\mathbb Z$}$\rangle $) } 
!c do
        \end{lstlisting}
    \end{minipage}
\end{wrapfigure}

The technical material of this article is implemented in a Haskell library
$\libname{}$. The library includes an implementation of \lang{} from
\cref{sec:language} with
invariant annotations
(parser, AST, signatures, well-\hspace{0pt}formedness), the semantic spaces ($\hilbert[ \signature ], \fockSpace[
\signature ]$ in \cref{def:basis-signature,def:fock-space}), the \ihps{}
representation of \cref{sec:hps} (inductive types in
\cref{def:components,def:hps}, interpretation in
\cref{def:hps-interpretation,def:density-operator-interpretation}, equivalence
\cref{def:equivalence-hps}, filtering and projection, \ldots{} \appx), and a
forward-driven symbolic execution engine using the rules of the logic of
\cref{fig:hoare-logic} which raises proof obligations for the
\emph{initialization} and \emph{conservation} of loop invariants. The library is
accompanied by an executable for demonstration.
We explain the functioning of $\libname$ more concretely with the \ref{cointoss}
example.  We write the source code of \ref{cointoss} \appx{} in a file
\texttt{coin-toss.hyq} with the \texttt{while} loop being annotated with
invariants as the excerpt in \cref{lst:annotation}, and the final \TeX output is rendered in \cref{fig:cointoss-output}.

\begin{figure}[htpb]
    \centering
        \begin{minipage}{0.9\textwidth}
            \small
            \[\text{Result:}\quad \vDash \left\{\pasum{0}{\emptyset}{1}\right\}\texttt{prog}\left\{\lim_{x\in \mathbb{N}} \left(\sum_{y \in \mathbb{N}} \left(\pasum{0}{\ket[\mathtt{q}]{1}\classKet[\mathtt{c}]{1}\classKet[\mathtt{x}]{y}}{\frac{\uparrow\left(y \leq x\right) \cdot 1}{\sqrt{2^{y + 1}}}}\right)\right)\right\}\]

            \noindent \textbf{\large Proof obligations produced:}

            \emph{Sanity checks done with floats of absolute tolerance \textcolor{red}{1.0e-15} and for free integer variables ranging from \textcolor{red}{0 to 100}}:

            \noindent Equivalence 1 (Initialization) --- \textcolor{blue}{PASSED} semantic sanity check: \ensuremath{\textcolor{gray}{\text{FOUND} \equiv \text{EXPECTED}}}
            \begin{dmath*} \sum_{c_0 \in \mathbb{B}} \left(\pasum{0}{\ket[\mathtt{q}]{c_0}\classKet[\mathtt{c}]{c_0}\classKet[\mathtt{x}]{0}\pastKet[\mathbb{B}]{0}}{\frac{1}{\sqrt{2}}}\right)\equiv \left(\sum_{y \in \mathbb{N}} \left(\pasum{0}{\ket[\mathtt{q}]{1}\classKet[\mathtt{c}]{1}\classKet[\mathtt{x}]{y}}{\frac{\uparrow\left(y \leq 0\right)}{\sqrt{2^{y + 1}}}}\right) \psadd \pasum{0}{\ket[\mathtt{q}]{0}\classKet[\mathtt{c}]{0}\classKet[\mathtt{x}]{0}}{\frac{1}{\sqrt{2^{1}}}}\right)\end{dmath*}

            \noindent Equivalence 2 (Conservation) --- \textcolor{blue}{PASSED} semantic sanity check: \ensuremath{\textcolor{gray}{\text{FOUND} \equiv \text{EXPECTED}}}
            \begin{dmath*} \left(\sum_{c_0 \in \mathbb{B}} \left(\pasum{0}{\ket[\mathtt{q}]{c_0}\classKet[\mathtt{c}]{c_0}\classKet[\mathtt{x}]{\left( x + 1 \right)}\pastKet[\mathbb{B}]{0}\pastKet[\mathbb{Z}]{x}}{\frac{1}{\sqrt{2^{x + 2}}}}\right) \sqpsadd \sum_{y \in \mathbb{N}} \left(\pasum{0}{\ket[\mathtt{q}]{1}\classKet[\mathtt{c}]{1}\classKet[\mathtt{x}]{y}}{\frac{\uparrow\left(y \leq x\right) \cdot 1}{\sqrt{2^{y + 1}}}}\right)\right)\equiv \left(\sum_{y \in \mathbb{N}} \left(\pasum{0}{\ket[\mathtt{q}]{1}\classKet[\mathtt{c}]{1}\classKet[\mathtt{x}]{y}}{\frac{\uparrow\left(y \leq \left( x + 1 \right)\right)}{\sqrt{2^{y + 1}}}}\right) \psadd \pasum{0}{\ket[\mathtt{q}]{0}\classKet[\mathtt{c}]{0}\classKet[\mathtt{x}]{\left( x + 1 \right)}}{\frac{1}{\sqrt{2^{x + 2}}}}\right)\end{dmath*}

        \end{minipage}
    \caption{Output of \ihps{} on the annotated source code \texttt{coin-toss.hyq} of
    \ref{cointoss}.}\label{fig:cointoss-output}
\end{figure}

In \cref{lst:annotation}, line 2 defines
the iteration index $\intVar$ used in the invariant. Lines 3 and 4 then show the
invariant \ihps{} itself given by the exiting cases (line 3) and the non-exiting
case (line 4). 
When we call the tool on this file using 
\verb|$ cabal run ihps -- < coin-toss.hyq|, the minimal signature $\signature$
for which the program is well-formed; that is, such that $\prog \in
\programSet_\signature$, is computed, if it exists.  Next, an initial \ihps{} of
the form $\h_0 \equaldef \pasum{0}{m}{1}$ where $m$ is the memory of signature
$\signature$ and where all the addresses are assigned to 0 is generated. For
\ref{cointoss}, $\signature=\emptyset$ and $\h_0 = \pasum{0}{\emptyset}{1}$. The
tool is then able to compute the \ihps{} $\h$ and a set of equivalences
$\equivContext$ such that $\models
\judgement{\;}{\equivContext}{\hoare{\h_0}{\prog}{\h}}$ by forward application
of the logic rules of \cref{fig:hoare-logic}, along the way raising proof
obligations in the form of initialization and conservation equivalences in
$\equivContext$, the resolution of which is left for future work via proof
assistants and/or SMT solvers. The equivalences do pass through a semantic check
which does detect most errors. The tool is efficient: the generation of the
\ihps{} and equivalences takes milliseconds, with the remaining time being
dominated by the semantic check, but remaining on the order of seconds for the
considered examples for all practical purposes.  Finally,
\libname{} then pretty-prints a `report' of the results in \TeX{} format. For
\ref{cointoss}, this produced \TeX{} is, verbatim, the code of
\cref{fig:cointoss-output}.

In \cref{fig:cointoss-output}, we see the resulting derived Hoare triple as
well as the equivalences in $\equivContext$. The tool also states that the
equivalences have passed a semantic sanity check and describes its parameters:
the floating-point CQ interpretations are compared with absolute tolerance
$10^{-15}$ and the invariant conservation equivalences are tested for free
integer variables ranging from 0 to 100, parameters choosable by the user with
command-line arguments.  Note also that the \ihps{} expressions are largely
simplified to a more readable form.

Finally, $\libname$ can handle analyses with real-valued parameters, such as the
analysis of \ref{repeatedQBF} where the value $p$ is a formal symbolic variable.
This means that the analysis is valid universally over the values of $p$. For
the semantic check of this universal validity, the user can specify a range of
values to be tested for $p$.

\section{Conclusion and Future Work}%
\label{sec:conclusion}

We have introduced \lang{} and \ihps{}, forming a framework for symbolic
execution and reasoning about hybrid quantum programs that allows the analysis
of quantum programs with unbounded loops. This static analysis framework is
essential for understanding the behavior of quantum programs in the upcoming era
of practical quantum computing, where testing and benchmarking remain nearly
impossible.  In the future, we seek to extend the expressiveness and the level
of automation of the implementation and integrate the rewrite system into it.
Some open questions in this regard are: what syntactic restrictions can be
imposed on programs to ensure the decidability of the analysis, and when it is
decidable, is the complexity of the analysis reasonable? Another future
direction of work is certainly to loosen the analysis: instead of producing
exact results about fixpoints, expectations, and probabilities, we could aim for
approximations and bounds. This could include abstract interpretation techniques
such as interval analysis, where bounds on the results can be obtained. In
short, we believe this work to be a foundation for a rich framework for the
analysis of the rich hybrid quantum programs expected to run on near-term
quantum computers.

\bibliographystyle{splncs04}
\bibliography{bibliography.bib}

\vfill
\pagebreak
\appendix
\crefname{section}{Appendix}{Appendices}
\Crefname{section}{Appendix}{Appendices}


\section{Exhaustive formalism}%
\label{app:exhaustive-formalism}

In this section of the appendix, we provide the full details of the formalism
which are excluded from the main text for readability and comprehensibility.

\subsection{Memory behavior of programs}%
\label{app:memory-behavior}

As mentioned in Section~\ref{sec:language}, not all programs are valid starting from
any state. We define the validity of a program $\prog$ over a state according to
the memory allocation profile or $\signature$ of the state by a judgement
$\validProg{\signature_1}{\prog}{\signature_2}$ read as ``the program $\prog$ is
valid to be executed starting from a state with memory allocation profile
$\signature_1$ and will result in a state with memory allocation profile
$\signature_2$''. The rules defining this judgement are given in
\cref{fig:language-well-formedness}, and notably include a check on the
unitarity of the program in the~\ref{rule:valid:unitary} rule.

\begin{figure}[H]
        \begin{prooftree}
            \AxiomC{}
            \RightLabel{\labelrule{rule:valid:skip}{\textsc{Skip}}}
            \UnaryInfC{$\validProg{\signature }{ \skipProg }{ \signature}$}
            \DisplayProof\hskip 1em
            \AxiomC{$\validProg{\signature_1 }{ \prog_1 }{ \signature_2}$}
            \AxiomC{$\validProg{\signature_2 }{ \prog_2 }{ \signature_3}$}
            \RightLabel{\labelrule{rule:valid:seq}{\textsc{Seq}}}
            \BinaryInfC{$\validProg{\signature_1}{\sequenceProg{\prog_1}{\prog_2}}{
            \signature_3}$}
            \DisplayProof\vskip 1em
            \AxiomC{$\qubit_1, \ldots, \qubit_n$ distinct}
            \AxiomC{$\{\qubit_1, \ldots, \qubit_n\} \subseteq \signature$}
            \RightLabel{\labelrule{rule:valid:unitary}{\textsc{Unitary}}}
            \BinaryInfC{$\validProg{\signature}{\unitary(\qubit_1, \ldots,
                \qubit_n)}{ \signature}$}
            \DisplayProof\vskip 1em
            \AxiomC{$\cbit \in \signature$}
            \AxiomC{$\qubit \in \signature$}
            \RightLabel{\labelrule{rule:valid:measure}{\textsc{Measure}}}
            \BinaryInfC{$\validProg{\signature }{ \measure{\cbit}{\qubit} }{ \signature}$}
            \DisplayProof\vskip 1em
            \AxiomC{$\cInt \in \signature$}
            \AxiomC{$\mathrm{Var}(\integer) \subseteq \signature$}
            \RightLabel{\labelrule{rule:valid:intassign}{\textsc{IntAssign}}}
            \BinaryInfC{$\validProg{\signature }{ \assign{\cInt}{\integer} }{ \signature}$}
            \DisplayProof\hskip 1em
            \AxiomC{$\cbit \in \signature$}
            \AxiomC{$\mathrm{Var}(\bool) \subseteq \signature$}
            \RightLabel{\labelrule{rule:valid:boolassign}{\textsc{BoolAssign}}}
            \BinaryInfC{$\validProg{\signature }{ \assign{\cbit}{\bool} }{ \signature}$}
            \DisplayProof\vskip 1em

            \AxiomC{$\qbit \notin \signature$}
            \RightLabel{\labelrule{rule:valid:qinit}{\textsc{QInit}}}
            \UnaryInfC{$\validProg{\signature }{ \qubitInitProg{\qbit} }{ \signature \cup \{\qbit\}}$}
            \DisplayProof\hskip 1em
            \AxiomC{$\cbit \notin \signature$}
            \RightLabel{\labelrule{rule:valid:cinit}{\textsc{CInit}}}
            \UnaryInfC{$\validProg{\signature }{ \cbitInitProg{\cbit} }{ \signature \cup \{\cbit\}}$}
            \DisplayProof\vskip 1em

            \AxiomC{$\cInt \notin \signature$}
            \RightLabel{\labelrule{rule:valid:intinit}{\textsc{IntInit}}}
            \UnaryInfC{$\validProg{\signature }{ \intInitProg{\cInt} }{ \signature \cup \{\cInt\}}$}
            \DisplayProof\vskip 1em

            \AxiomC{$\mathrm{Var}(\bool) \subseteq \signature$}
            \AxiomC{$\validProg{\signature }{ \prog_1 }{ \signature'}$}
            \AxiomC{$\validProg{\signature }{ \prog_2 }{ \signature'}$}
            \RightLabel{\labelrule{rule:valid:if}{\textsc{If}}}
            \TrinaryInfC{$\validProg{\signature }{ \ifthene{\bool}{\prog_1}{\prog_2} }{ \signature'}$}
                \DisplayProof\vskip 1em
                \AxiomC{$\mathrm{Var}(\bool) \subseteq \signature$}
                \AxiomC{$\validProg{\signature }{ \prog }{ \signature}$}
                \RightLabel{\labelrule{rule:valid:while}{\textsc{While}}}
                \BinaryInfC{$\validProg{\signature }{ \while{\bool}{\prog} }{ \signature}$}
            \end{prooftree}
        \caption{Memory behavior of programs and the judgement
        $\validProg{\signature}{\prog}{\signature'}$}%
        \label{fig:language-well-formedness}
\end{figure}
\hfill\pagebreak

\subsection{Interpretations of terms under variable assignment environments}%
\label{app:interpretation-terms}

\begin{figure}[H]
    \centering
        \begin{subfigure}{.45\textwidth}
            \begin{align*}
                \semFock[\assignmentContext]{k} & \equaldef k \in \Z \\
                \semFock[\assignmentContext]{\intVar} & \equaldef \assignmentContext(\intVar) \in \Z \\
                \semFock[\assignmentContext]{\lift{\boolTerm}} & \equaldef \begin{cases}
                    0 & \text{if } \semFock[\assignmentContext]{\boolTerm} = 0 \\
                    1 & \text{if } \semFock[\assignmentContext]{\boolTerm} = 1
                \end{cases}  \\
                    \semFock[\assignmentContext]{\intTerm_1 + \intTerm_2} & \equaldef \semFock[\assignmentContext]{\intTerm_1} +
                    \semFock[\assignmentContext]{\intTerm_2} \\
                    \semFock[\assignmentContext]{\intTerm_1 \cdot \intTerm_2} & \equaldef \semFock[\assignmentContext]{\intTerm_1}
                    \cdot
                    \semFock[\assignmentContext]{\intTerm_2}
                \end{align*}
                \caption{Interpretation of $\intType$ in $\Z$}%
                \label{fig:integer-interpretation}
            \end{subfigure}
            \hfill
            \begin{subfigure}{.45\textwidth}
                \begin{align*}
                    \semFock[\assignmentContext]{y} & \equaldef \assignmentContext(y) \in \B \\
                    \semFock[\assignmentContext]{\intTerm_1 = \intTerm_2} & \equaldef \begin{cases}
                        1 & \text{if } \semFock[\assignmentContext]{\intTerm_1} = \semFock[\assignmentContext]{\intTerm_2} \\
                        0 & \text{otherwise}
                    \end{cases} \\
                        \semFock[\assignmentContext]{\intTerm_1 \leq \intTerm_2} & \equaldef \begin{cases}
                            1 & \text{if } \semFock[\assignmentContext]{\intTerm_1} \leq \semFock[\assignmentContext]{\intTerm_2} \\
                            0 & \text{otherwise}
                        \end{cases} \\
                            \semFock[\assignmentContext]{\boolTerm_1 \cdot \boolTerm_2} & \equaldef \semFock[\assignmentContext]{\boolTerm_1}
                            \cdot
                            \semFock[\assignmentContext]{\boolTerm_2}  \\
                            \semFock[\assignmentContext]{\boolTerm_1 \oplus \boolTerm_2} & \equaldef \semFock[\assignmentContext]{\boolTerm_1}
                            + 
                            \semFock[\assignmentContext]{\boolTerm_2}
                        \end{align*}
                        \caption{Interpretation of $\boolType$ in $\B$}%
                        \label{fig:boolTermean-interpretation}
                    \end{subfigure}
                    \begin{subfigure}{.45\textwidth}
                        \begin{align*}
                            \semFock[\assignmentContext]{\boolTerm/2^\intTerm} & \equaldef
                            \frac{\semFock[\assignmentContext]{\boolTerm}}{2^{\semFock[\assignmentContext]{\intTerm}}}
                            \\ \semFock[\assignmentContext]{p_1 + p_2} & \equaldef
                            \semFock[\assignmentContext]{p_1} + \semFock[\assignmentContext]{p_2} \\
                            \semFock[\assignmentContext]{\intTerm \cdot p} & \equaldef
                            \semFock[\assignmentContext]{\intTerm} \cdot
                            \semFock[\assignmentContext]{p}  
                        \end{align*}
                        \caption{Interpretation of $\phaseType$ in $\R$ }%
                        \label{fig:phase-interpretation}
                        \begin{align*}
                            \semFock[\assignmentContext]{\ket[\qbit]{\boolTerm}} & \equaldef
                            \ket[\qbit]{\semFock[\assignmentContext]{\boolTerm}}\\
                            \semFock[\assignmentContext]{{[\intTerm]}_\cInt} & \equaldef
                            {[\semFock[\assignmentContext]{\intTerm}]}^\intTerm_\cInt \\
                            \semFock[\assignmentContext]{m_1 \otimes m_2} & \equaldef 
                            \semFock[\assignmentContext]{m_1} \otimes \semFock[\assignmentContext]{m_2}
                        \end{align*}
                        \caption{Interpretation of $\memoryType$ in $\fockSpace$}%
                        \label{fig:memory-interpretation}
                    \end{subfigure}
                    \hfill
                    \begin{subfigure}{.45\textwidth}
                        \begin{align*}
                            \semFock[\assignmentContext]{\frac \intTerm {\sqrt{2^\intTerm}}} & \equaldef \frac{\semFock[\assignmentContext]{\intTerm}}{\sqrt{2^{\semFock[\assignmentContext]{\intTerm}}}}
                           \\
                            \semFock[\assignmentContext]{\cos(2\pi p)} & \equaldef \cos(2\pi \semFock[\assignmentContext]{p})\\
                            \semFock[\assignmentContext]{\sin(2\pi p)} & \equaldef \sin(2\pi \semFock[\assignmentContext]{p})\\
                            \semFock[\assignmentContext]{n_1 + n_2} & \equaldef \semFock[\assignmentContext]{n_1} + \semFock[\assignmentContext]{n_2}\\
                            \semFock[\assignmentContext]{n_1 \cdot n_2} & \equaldef \semFock[\assignmentContext]{n_1}
                            \cdot
                            \semFock[\assignmentContext]{n_2} \\
                            \semFock[\assignmentContext]{n_1 / n_2} & \equaldef \semFock[\assignmentContext]{n_1} / \semFock[\assignmentContext]{n_2} \\
                            \semFock[\assignmentContext]{\sqrt n} & \equaldef \sqrt{\semFock[\assignmentContext]{n}}
                        \end{align*}
                        \caption{Interpretation of $\normType$ in $\R$}%
                        \label{fig:norm-interpretation}
                    \end{subfigure}
                \caption{Interpretation of terms under variable assignment
                environments}%
                \label{fig:term-interpretations}
            \end{figure}

\subsection{Memory access}%
\label{app:get-function}

The functions $\bool * \h$ and $\h[a \ot t]$ require accessing the memory of
$\h$. Given that $\h$ is a complex expression which is not necessarily of the
form $\pasum{p}{m}{n}$ where the memory $m$ is accessible. Therefore, we need an
intermediate function $\redotimes{\cdot}$ which reduces $\h$ to an equivalent form
$\redotimes{\h} \equiv \h$ where, whenever $\h$ is of the form $\h_1 \otimes
\h_2$, either $\signature(\h_1) = \emptyset$ or $\signature(\h_2) =
\emptyset$; that is, one side of the tensor product contains all the (present)
memory needed for access. The function $\redotimes{\h}$ is defined inductively
on the structure of $\h$ as given in \cref{fig:redotimes}. Assuming this form,
we give the following definitions for $\bool * \h$ and $\h[a \ot t]$:

\begin{figure}[H]
    \centering
        \begin{alignat*}{2}
            \bool * \pasum{p}{m}{n} &\equaldef \pasum{p}{m}{\eval{\bool}{m} n} \\
            \bool * \sum_x \h &\equaldef \sum_x \bool * \h \\
            \bool * \lim_k \h &\equaldef \lim_k \bool * \h \\
            \bool * (\h_1 \psadd \h_2) &\equaldef \bool * \h_1 \psadd \bool * \h_2 \\
            \bool * (\h_1 \sqpsadd \h_2) &\equaldef \bool * \h_1 \sqpsadd \bool * \h_2
            \\
            \bool_1 * (\h_1 \otimes \h_2) &\equaldef (\bool_1 * \h_1) \otimes \h_2 &&
            \text{if } \mathrm{Var}(\bool_1) \subseteq \signature(\h_1) \\
            \bool_2 * (\h_1 \otimes \h_2) &\equaldef \h_1 \otimes (\bool_2 * \h_2) &&
            \text{if } \mathrm{Var}(\bool_2) \subseteq \signature(\h_2) \\
        \end{alignat*}
    \begin{alignat*}{2}
        \pasum{p}{m}{n}[a \ot t] &\equaldef \pasum{p}{\{\eval{t}{m}\} \otimes m[a \ot \eval{t}{m}]}{n} \\
        \left(\sum_x \h\right)[a \ot t] &\equaldef \sum_x \h[a \ot t] \\
        \left(\lim_k \h\right)[a \ot t] &\equaldef \lim_k \h[a \ot t] \\
        (\h_1 \psadd \h_2)[a \ot t] &\equaldef \h_1[a \ot t] \psadd \h_2[a \ot t] \\
        (\h_1 \sqpsadd \h_2)[a \ot t] &\equaldef \h_1[a \ot t] \sqpsadd \h_2[a \ot t] \\
        (\h_1 \otimes \h_2)[a \ot t] &\equaldef (\h_1[a \ot t]) \otimes \h_2 &&
        \text{if } \mathrm{Var}(t) \subseteq \signature(\h_1) \\
        (\h_1 \otimes \h_2)[a \ot t] &\equaldef \h_1 \otimes (\h_2[a \ot t]) &&
        \text{if } \mathrm{Var}(t) \subseteq \signature(\h_2)
    \end{alignat*}
    \caption{The functions $\bool * \h$ and $\h[a \ot t]$ applied to $\h$
    rewritten as $\redotimes{\h}$}%
    \label{fig:projections-and-filtering}%
\end{figure}

Certain desired algebraic rules, such as the one saying that if $\h_1 \equiv
\h_1'$ and $\h_2 \equiv \h_2'$, then $\h_1 + \h_2 \equiv \h_1' + \h_2'$, are not
technically sound for the equivalence introduced since erasable phases global to
$\h_1$ and $\h_2$ become relative in $\h_1 + \h_2$. Therefore, there is a need
to define a stronger equivalence relation $\equivStrong$ for which
$\assignmentContext \models \h_1 \equiv^{\mathtt P} \h_2$ is equivalent to
$\semFock[\assignmentContext]{\h_1} = \semFock[\assignmentContext]{\h_2}$. This
stronger equivalence is not directly used in this article, but remains a notable
technical detail for the soundness of certain rules of the equational theory. We
refer the reader to \HQbricks~\cite{CIN+26}, which explains $\equiv^{\mathtt P}$
in further detail via its \HQbricks{} analogue written $\equiv^{\mathtt P}$.

\begin{definition}[Strong equivalence]%
    \label{def:strong-equivalence}
    $\h_1, \h_2 \in \hpsType$ are \emph{strongly equivalent}, denoted $\h_1
    \equiv^{\mathtt{P}} \h_2$, if and only if
    \[
        \signature(\h_1) = \signature(\h_2)\ \mathrm{and}\ \forall \assignmentContext, \mathrm{s.t.} \mathrm{dom}(\assignmentContext) \supseteq
        \mathrm{Var}(\h_1) \cup \mathrm{Var}(\h_2), \; \semFock{\h_1} =
        \semFock{\h_2}
    \]
\end{definition}

\begin{restatable}[Correctness of
    $\redotimes{\cdot}$]{lemma}{restateCorrectnessRedotimes}%
    \label{lem:redotimes-correctness}
    For all $\h \in \hpsType$ that converge, that is, $\semFock{\h} \neq \bot$, we have
    \begin{enumerate}[(i)]
        \item $\h \equiv_s \redotimes{\h}$, and
        \item If $\h_1 \otimes \h_2$ is a subterm of $\redotimes{\h}$, then either
            $\signature(\h_1) = \emptyset$ or $\signature(\h_2) = \emptyset$.
    \end{enumerate}
\end{restatable}

\begin{figure}[H]
    \centering
        \begin{align*}
            \redotimes{\pasum{p}{m}{n}} &\equaldef \pasum{p}{m}{n} \\
            \redotimes{\h_1 \psadd \h_2} &\equaldef \redotimes{\h_1} \psadd
            \redotimes{\h_2} \\
            \redotimes{\h_1 \sqpsadd \h_2} &\equaldef \redotimes{\h_1} \sqpsadd
            \redotimes{\h_2} \\
            \redotimes{\sum_x \h} &\equaldef \sum_x \redotimes{\h} \\
            \redotimes{\lim_k \h} &\equaldef \lim_k \redotimes{\h}\\
            \redotimes{\bigotimes_x \h} &\equaldef \bigotimes_x \redotimes{\h}
            \\
            \redotimes{(\h_1 \psadd \h_2) \otimes \h_3} &\equaldef
            \redotimes{\h_1 \otimes \h_3} \psadd \redotimes{\h_2 \otimes \h_3}
            \\
            \redotimes{(\h_1 \sqpsadd \h_2) \otimes \h_3} &\equaldef
            \redotimes{\h_1 \otimes \h_3} \sqpsadd \redotimes{\h_2 \otimes \h_3}
            \\
            \redotimes{\lim_\intVar \h_1 \otimes \h_2} &\equaldef
            \lim_{\intVar'} \redotimes{\h_1[\intVar'/\intVar] \otimes \h_2} \\
            \redotimes{\sum_{\var} \h_1 \otimes \h_2} &\equaldef
            \sum_{\var'} \redotimes{\h_1[\var'/\var] \otimes \h_2} \\
            \redotimes{\left(\bigotimes_{\var} \h_1\right) \otimes \h_2}
                                                    &\equaldef
            \left(\bigotimes_{\var}
            \redotimes{\h_1}\right)
            \otimes \redotimes{\h_2} \\
            \redotimes{(\h_1 \otimes \h_2) \otimes \h_3} &\equaldef
            \redotimes{\h_1 \otimes \redotimes{\h_2 \otimes \h_3}} \\
            \redotimes{\pasum{p}{m}{n} \otimes \pasum{p'}{m'}{n'}} &\equaldef
            \pasum{p + p'}{m \otimes m'}{n \cdot n'} \\
            \redotimes{\pasum{p}{m}{n} \otimes (\h_2 \psadd \h_3)} &\equaldef
            \redotimes{\pasum{p}{m}{n} \otimes \h_2} \psadd \redotimes{\pasum{p}{m}{n} \otimes \h_3}
            \\
            \redotimes{\pasum{p}{m}{n} \otimes (\h_2 \sqpsadd \h_3)} &\equaldef
            \redotimes{\pasum{p}{m}{n} \otimes \h_2} \sqpsadd \redotimes{\pasum{p}{m}{n} \otimes \h_3}
            \\
            \redotimes{\pasum{p}{m}{n} \otimes \lim_\intVar \h_2} &\equaldef
            \lim_{\intVar'}
            \redotimes{\pasum{p}{m}{n} \otimes \h_2[\intVar'/\intVar]} \\
            \redotimes{\pasum{p}{m}{n} \otimes \sum_{\var} \h_2} &\equaldef
            \sum_{\var'} \redotimes{\pasum{p}{m}{n} \otimes \h_2[\var'/\var]} \\
            \redotimes{\pasum{p}{m}{n} \otimes \left(\bigotimes_{\var} \h_2\right)}
                                                    &\equaldef
            \redotimes{\pasum{p}{m}{n}} \otimes \left(\bigotimes_{\var} \redotimes{\h_2}\right) \\
            \redotimes{\pasum{p}{m}{n} \otimes (\h_2 \otimes \h_3)} &\equaldef
            (\redotimes{\pasum{p}{m}{n} \otimes \h_2}) \otimes \redotimes{\h_3}
        \end{align*}
        with $\intVar'$ and $\var'$ fresh variables.
    \caption{$\otimes$-reduced form $\redotimes{\h}$ of an \ihps{} $\h$}%
    \label{fig:redotimes}
\end{figure}
\begin{figure}[H]
    \begin{subfigure}[htpb]{0.47\textwidth}
                \begin{align*}
                    \eval{\trueBool}{m} &\equaldef 1 \\
                    \eval{\falseBool}{m} &\equaldef 0 \\
                    \eval{\integer_1 \leq \integer_2}{m} &\equaldef \eval{\integer_1}{m} \leq
                    \eval{\integer_2}{m} \\
                    \eval{\integer_1 = \integer_2}{m} &\equaldef \eval{\integer_1}{m} =
                    \eval{\integer_2}{m} \\
                    \eval{\bool_1 \land \bool_2}{m} &\equaldef \eval{\bool_1}{m} \cdot
                    \eval{\bool_2}{m} \\
                    \eval{\lnot{\bool}}{m} &\equaldef 1 \oplus {\eval{\bool}{m}} \\
                    \eval{\cbit}{m_1  \ket[\cbit]{\boolTerm}  m_2} &\equaldef \boolTerm \\
                    \eval{\qbit}{m_1  \ket[\qbit]{\boolTerm}  m_2} &\equaldef \boolTerm \\
                \end{align*}
            \caption{$\eval{\bool}{m}$ for boolean $\bool$ values}%
            \label{fig:bool-eval}
    \end{subfigure}
    \hfill
    \begin{subfigure}[htpb]{0.47\textwidth}
                \vspace{0.5em}
            \begin{align*}
                \eval{k}{m} &\equaldef k \\
                \eval{\integer_1 + \integer_2}{m} &\equaldef \eval{\integer_1}{m} +
                \eval{\integer_2}{m} \\
                \eval{\integer_1 * \integer_2}{m} &\equaldef \eval{\integer_1}{m} \times
                \eval{\integer_2}{m} \\
                \eval{\integer_1 ^ {\integer_2}}{m} &\equaldef
                \eval{\integer_1}{m}^{\eval{\integer_2}{m}} \\
                \eval{\cInt}{m_1  {\left[\intTerm\right]}_\cInt 
            m_2} & \equaldef \intTerm \\
            \end{align*}
            \caption{$\eval{\integer}{m}$ for integer $\integer$ values}%
            \label{fig:integer-eval}
    \end{subfigure}
    \caption{The inductive definition of $\eval{\cdot}{\h}$ evaluating boolean
    $\bool$ and integer $\integer$ values within the environment $\h$.}%
    \label{fig:full-eval}
\end{figure}

\subsection{Hoare-style rules for unitaries}%
\label{sub:hoare-style-rules-for-unitaries}

\begin{figure}[H]
    \centering
        \begin{align*}
            \applyU{\unitary}(\h_1 \psadd \h_2) &\equaldef \applyU{\unitary}(\h_1) \psadd
            \applyU{\unitary}(\h_2) \\
            \applyU{\unitary}(\h_1 \sqpsadd \h_2) &\equaldef \applyU{\unitary}(\h_1) \sqpsadd
            \applyU{\unitary}(\h_2) \\
            \applyU{\unitary}(\h_1 \otimes \h_2) &\equaldef
            \applyU{\unitary}(\h_1) \otimes \h_2 \text{ if } \signature(\h_2) =
            \emptyset \\
            \applyU{\unitary}(\h_1 \otimes \h_2) &\equaldef
            \h_1 \otimes \applyU{\unitary}(\h_2) \text{ if } \signature(\h_1) =
            \emptyset \\
            \applyU{\unitary}\left(\sum_x \h\right) &\equaldef \sum_x \applyU{\unitary}(\h) \\
            \applyU{\unitary}\left(\lim_k \h\right) &\equaldef \lim_k \applyU{\unitary}(\h)\\
            \applyU{\unitary}\left(\bigotimes_x \h\right) &\equaldef
            \text{never occurs since } \signature(\bigotimes_x \h) = \emptyset \\
            \applyU{\texttt{CNOT}(\qbit_1, \qbit_2)}(
            \pasum{p}{\ket[\qbit_1]{f} \otimes \ket[\qbit_2]{g}}{n})
                                                          &\equaldef
            \pasum{p}{\ket[\qbit_1]{f} \otimes \ket[\qbit_2]{f \oplus g}}{n} \\
            \applyU{\texttt{H}(\qbit)}^*(
            \pasum{p}{\ket[\qbit]{f}}{n}) &\equaldef
            \bigpsadd_y \pasum{p + \frac{f
            y}{2}}{\ket[\qbit]{y}}{\frac{1}{\sqrt{2}}n}
            \\
            \applyU{\texttt{R}_k(\qbit)}(
            \pasum{p}{\ket[\qbit]{f}}{n}) &\equaldef
            \pasum{p + \frac{f}{2^k}}{\ket[\qbit]{f}}{n} \\
            \applyU{\texttt{X}(\qbit)}(
            \pasum{p}{\ket[\qbit]{f}}{n}) &\equaldef
            \pasum{p}{\ket[\qbit]{f \oplus 1}}{n} \\
            \applyU{\texttt{Z}(\qbit)}(
            \pasum{p}{\ket[\qbit]{f}}{n}) &\equaldef
            \pasum{p + \frac{f}{2}}{\ket[\qbit]{f}}{n} 
        \end{align*}
    \caption{Definition of $\applyU{\unitary}(\h)$ given for $\h$ rewritten
    in $\otimes$-reduced form $\h \mapsto \redotimes{\h}$}\label{fig:applyU}
\end{figure}

Unitary gates are applied again as symbolic transformations of \ihps{} terms. We
define them inductively on \ihps{} in general, and for each unitary gate, we
define its effect on primitive \ihps{} terms. They are given in
\cref{fig:applyU}. In said figure, it is assumed that $\applyU{\unitary}(\h)$
will only be used when the program is valid on $\h$; that is, $\unitary \in
\programSet_{\signature(\h)}$. In particular, it is assumed that $\unitary$
always has distinct arguments, and that said arguments are qubits already
initialized in $\h$. Also note that the inductive case of $\bigotimes_x \h$
never occurs under these conditions since $\signature(\bigotimes_x \h) =
\emptyset$ and a unitary must act on at least one qubit. Furthermore, for the
case of $\h_1 \otimes \h_2$, we face the same issues as in
\cref{app:get-function} of memory addressing for a memory that is split over two
\ihps{} terms, and solve it in the same way by assuming $\applyU{\unitary}$
accepts as input the equivalent $\otimes$-reduced form $\redotimes{\h}$ of $\h$
(\cref{fig:redotimes}).

\section{Proofs}%
\label{sec:proofs}

\subsection{Undecidability results}%

\restateUndecidability*
\begin{proof}
    Let $\forall \intVar_1 \cdots \forall \intVar_n \psi(\intVar_1, \ldots,
    \intVar_n)$ be any $\Pi_1^0$ arithmetic formula with $\psi$ being
    quantifier-free. Let $\boolTerm \in \boolType$ be
    the boolean expression corresponding to $\psi$.  We construct the closed
    \ihps{} expressions $\h_\psi \equaldef \sum_{\intVar_1} \cdots
    \sum_{\intVar_n} \pasuminline{0}{\classKet[\cInt_1]{\intVar_1} \cdots
    \classKet[\cInt_n]{\intVar_n} }{\frac 1 {\sqrt{2^{\intVar_1 + \cdots +
    \intVar_n}}} \cdot (1-\lift(\boolTerm(\intVar_1, \ldots, \intVar_n)))}$ and $0 =
    \pasuminline{0}{0}{\emptyset}$, and check if $\h_\psi$ and $0$ are equivalent.
    The equivalence holds iff $\|\h_\psi\| = 0$, which is the case iff
    $\boolTerm(\intVar_1, \ldots, \intVar_n)$ is true for all $\intVar_1, \ldots,
    \intVar_n$, i.e., iff $\forall \intVar_1 \cdots \forall \intVar_n
    \psi(\intVar_1, \ldots, \intVar_n)$ is valid. Furthermore, checking the
    equivalence $\h_1 \equiv \h_2$ reduces to checking $\models
    \hoare{\h_1}{\skipProg}{\h_2}$.
\end{proof}

\subsection{Coherence theorems}%
\label{sub:coherence-theorems}

\restateWelldefinednessDenotational*
\begin{proof}
    We proceed by induction on programs $\prog$.
    \begin{enumerate}
        \item $\prog = \skipProg$: Trivial, as
            $\sem{\skipProg}(\rho) = \rho \in
            \CQ(\applySig{\skipProg}{\signature(\rho)}) = \CQ(\signature(\rho))$.
        \item $\prog = \unitary(\bar \qbit)$: Similar to $\skipProg$
            in that $\applySig{\unitary(\bar \qbit)}{\signature(\rho)} =
            \signature(\rho)$, but also satisfies $\bar \qbit \in
            \signature(\rho)$ since $\prog \in \programSet_{\signature(\rho)}$.
        \item $\prog = \measure{\cbit}{\qbit}$, $\prog =
            \assign{\cbit_1}{\boolTerm}$, and $\prog =
            \assign{\integer_1}{\intTerm}$: almost exactly identical to
            $\unitary(\bar \qbit)$.
        \item $\prog = \qubitInitProg{\qbit}$: We have that
            $\signature(\rho) \not\ni \qbit$ since $\prog \in
            \programSet_{\signature(\rho)}$, and thus $\rho \otimes
            \ket[\qbit]{0}\bra{0}_{\qbit}$ is valid and in 
            $\applySig{\qubitInitProg{\qbit}}{\signature(\rho)} =
            \signature(\rho) \cup \{\qbit\}$. 
        \item $\prog = \cbitInitProg{\cbit}$ and $\prog = \intInitProg{\cInt}$:
            nearly identical to $\qubitInitProg{\qbit}$.
        \item $\prog = \sequenceProg{\prog_1}{\prog_2}$: By induction
            hypothesis, $\sem{\prog_1}(\rho) \in
            \CQ(\applySig{\prog_1}{\signature(\rho)})$, and by another
            application of the induction hypothesis, we have that
            $\sem{\prog_2}(\sem{\prog_1}(\rho)) \in
            \CQ(\applySig{\prog_2}{\applySig{\prog_1}{\signature(\rho)}})$,
            which is exactly
            $\CQ(\applySig{\sequenceProg{\prog_1}{\prog_2}}{\signature(\rho)})$.
        \item $\prog = \ifthene{\bool}{\prog_1}{\prog_2}$: We have that
                $\mathrm {Var}(\bool) \subseteq \signature(\rho)$ and 
                that $\signature(\rho) = \signature(\widehat F_\boolTerm(\rho))
                = \signature(\widehat F_{\lnot \bool}(\rho))$. Therefore, by induction,
                we have that $\sem{\prog_1}(\widehat F_\boolTerm(\rho)) \in
                \CQ(\applySig{\prog_1}{\signature(\rho)})$ and that
                $\sem{\prog_2}(\widehat F_{\lnot \bool}(\rho)) \in
                \CQ(\applySig{\prog_2}{\signature(\rho)})$. However, we assumed
                that $\applySig{\prog_1}{\signature(\rho)} =
                \applySig{\prog_2}{\signature(\rho)}$ for \textbf{if}
                statements, as such, the addition $\sem{\prog_1}(\widehat
                F_\boolTerm(\rho)) + \sem{\prog_2}(\widehat F_{\lnot \bool}(\rho))$ is
                valid as both terms belong to the same same, and the addition
                remains in that space, namely
                $\CQ(\applySig{\prog_1}{\signature(\rho)}) =
                \CQ(\applySig{\prog_2}{\signature(\rho)}) =
                \CQ(\applySig{\prog}{\signature(\rho)})$.
        \item $\prog = \while{\bool}{\prog'}$: In this case, we have
            $\applySig{\prog'}{\signature(\rho)} = \signature(\rho)$, and we can
            conclude that $\applySig{\prog'}{\signature(\rho)} =
            \applySig{\skipProg}{\signature(\rho)}$, which makes
            $\ifthene{\bool}{\prog'}{\skipProg}$ a valid program. By the
                induction hypothesis, and the fact that the signature is
                unchanged, we can then conclude that
                $\sem{\ifthene{\bool}{\prog'}{\skipProg}}^n(\rho) \in
                \CQ(\signature(\rho))$ for all $n \in \N$. Finally, we need to
                show that the limit converges. Indeed, compare the state $\widehat
                F_{\lnot \bool}(\sem {\ifthene{\bool}{\prog'} {\skipProg}}^n
                (\rho))$ with the next state $\widehat F_{\lnot \bool}
                (\sem{\ifthene{\bool} {\prog'}{\skipProg}}^{n+1} (\rho))$:
                \begin{align*}
                    & \widehat F_{\lnot \bool}
                    (\sem{\ifthene{\bool}{\prog'}{\skipProg}}^{n+1} (\rho)) \\
                    =& \widehat F_{\lnot \bool}
                    (\sem{\ifthene{\bool}{\prog'}{\skipProg}}
                    (\sem{\ifthene{\bool}{\prog'}{\skipProg}}^{n} (\rho))) \\
                    =& \widehat F_{\lnot \bool}
                    (\sem{\prog'}(\widehat F_{\boolTerm}
                    (\sem{\ifthene{\bool}{\prog'}{\skipProg}}^{n} (\rho))) + \\
                     & \hspace{4em} \widehat F_{\lnot \bool}
                    (\sem{\ifthene{\bool}{\prog'}{\skipProg}}^{n} (\rho))) \\
                    =& \widehat F_{\lnot \bool}
                    (\sem{\prog'}(\widehat F_{\boolTerm}
                    (\sem{\ifthene{\bool}{\prog'}{\skipProg}}^{n} (\rho)))) + \\
                     & \hspace{4em}\widehat F_{\lnot \bool}
                    (\sem{\ifthene{\bool}{\prog'}{\skipProg}}^{n} (\rho)) \\
                \end{align*}
                In other words, the sequence ${\left(\widehat F_{\lnot
                \bool} (\sem{\ifthene{\bool}{\prog'}{\skipProg}}^{n}
            (\rho))\right)}_{n \in \N}$ has a difference between successive
                terms which is always positive semidefinite. Moreover, since
                programs are trace-non-increasing, the sequence is bounded by
                $\mathrm{tr}(\rho)$. As such, it converges. That is, the limit
                \[
                    \lim_{n
                    \to \infty} \sem{\ifthene{\bool}{\prog'}{\skipProg}}^{n}
                        (\rho)
                    \]
                    exists within the subspace $\CQ(\signature(\rho))$.
    \end{enumerate}
    \qed%
\end{proof}


    

\restateCorrectnessRedotimes*
\begin{proof}
    This follows by induction on the calculation of $\redotimes{\h}$ as given in
    \cref{fig:redotimes}. 
    \begin{enumerate}
        \item $\redotimes{\pasum{p}{m}{n}} \equiv_s \pasum{p}{m}{n}$ by
            definition, and it has no tensor subterm
        \item $\redotimes{\h_1 \psadd \h_2} \equiv_s \h_1 \psadd \h_2$ by
            induction hypothesis, and the tensor subterms of
            $\redotimes{\h_1 \psadd \h_2}$ are those of $\redotimes{\h_1}$ and
            $\redotimes{\h_2}$, which satisfy the property by induction
            hypothesis.
         \item $\redotimes{\h_1 \sqpsadd \h_2}$, $\redotimes{\sum_x \h}$, and
            $\redotimes{\lim_k \h}$ are nearly identical to the case $\redotimes{\h_1 \psadd
            \h_2}$.
         \item $\redotimes{\bigotimes_x \h} \equiv_s \bigotimes_x \h$ by
            induction hypothesis, and $\signature(\bigotimes_x \h) =
            \signature(\h) = \signature(\text{subterms of } \h) = \emptyset$.
         \item $\redotimes{(\h_1 \psadd \h_2) \otimes \h_3}$ follows by
             distributivity of $\otimes$ over $\psadd$ in the Fock space:
             \begin{align*}
                 &\semFock{\redotimes{(\h_1 \psadd \h_2) \otimes \h_3}} \\
                 =& \semFock{\redotimes{\h_1 \otimes \h_3} \psadd
                 \redotimes{\h_2 \otimes \h_3}} \\
                 =& \semFock{\redotimes{\h_1 \otimes \h_3}} +
                 \semFock{\redotimes{\h_2 \otimes \h_3}} \\
                 =& \semFock{\h_1 \otimes \h_3} + \semFock{\h_2 \otimes \h_3}\\
                 =& \semFock{\h_1} \otimes \semFock{\h_3} + \semFock{\h_2}
                 \otimes \semFock{\h_3} \\
                 =& \left(\semFock{\h_1} + \semFock{\h_2}\right) \otimes
                 \semFock{\h_3} \\
                 =& \semFock{(\h_1 \psadd \h_2) \otimes \h_3}
             \end{align*}
             As for signatures, it follows from the induction hypothesis given
             that the right hand side of the definition is itself an application
             of $\redotimes{\cdot}$.
         \item The cases of $\redotimes{(\h_1 \sqpsadd \h_2) \otimes \h_3}$, 
            $\redotimes{\pasum{p}{m}{n} \otimes (\h_2 \psadd \h_3)}$, as well as
            $\redotimes{\pasum{p}{m}{n} \otimes (\h_2 \sqpsadd \h_3)}$ are
            nearly identical to $\redotimes{(\h_1 \psadd \h_2) \otimes
            \h_3}$.
         \item $\redotimes{\lim_\intVar \h_1 \otimes \h_2}$ and
             $\redotimes{\pasum{p}{m}{n} \otimes \lim_\intVar \h_2}$
             follow by the continuity of $\otimes$ in the Fock space and the
             variable renaming ensuring that no variables are accidentally
             absorbed by the limit:
             \begin{align*}
                 &\semFock{\redotimes{\lim_\intVar \h_1 \otimes \h_2}} \\
                 =& \semFock{\lim_{\intVar'} \redotimes{\h_1[\intVar'/\intVar]
                 \otimes \h_2}} \\
                         =& \lim_{\intVar' \to \infty}
                         \semFock{\redotimes{\h_1[\intVar'/\intVar] \otimes
                         \h_2}} \\
                 =& \lim_{\intVar' \to \infty} \semFock{\h_1[\intVar'/\intVar] \otimes
                 \h_2} \\
                     =& \lim_{\intVar' \to \infty} \semFock{\h_1[\intVar'/\intVar]} \otimes
                     \lim_{\intVar' \to \infty} \semFock{\h_2} \\
                     =& \lim_{\intVar' \to \infty} \semFock{\h_1[\intVar'/\intVar]} \otimes
                     \semFock{\h_2} \\
                     =& \lim_{\intVar \to \infty} \semFock{\h_1} \otimes
                     \semFock{\h_2} \\
                     =& \semFock{\lim_\intVar \h_1 \otimes \h_2}
             \end{align*}

             Note, in particular, that this reasoning by continuity is
             contingent on the convergence of the subterms $\lim_\intVar \h_1$
             and $\h_2$, which is implied by the convergence of $\lim_\intVar
             \h_1 \otimes \h_2$.  Meanwhile, the condition on signatures follows
             immediately from the induction hypothesis and from
             $\signature(\lim_\intVar \h_1) = \signature(\h_1)$.
         \item $\redotimes{\sum_{\var} \h_1 \otimes \h_2}$ and
             $\redotimes{\pasum{p}{m}{n} \otimes \sum_{\var} \h_2}$ follow by a
             similar argument using the linearity of $\otimes$ in the Fock space
             but also making use of the distributivity of $\otimes$ over
             $+$, and looks exactly as the case for $\redotimes{\lim_\intVar
             \h_1 \otimes \h_2}$.
         \item $\redotimes{(\h_1 \otimes_{\var} \h_2) \otimes \h_3}$ and
             $\redotimes{\pasum{p}{m}{n} \otimes (\h_2 \otimes_{\var} \h_3)}$
             follow by the associativity of $\otimes$.
         \item From the induction hypothesis and the associativity of $\otimes$,
             we get
             \begin{align*}
                 \textstyle
                 \redotimes{\left(\bigotimes_x \h\right)} &\textstyle\equiv_s \bigotimes_x
                 \redotimes{\h} \\
                 \textstyle\redotimes{\pasum{p}{m}{n} \otimes \left(\bigotimes_x
         \h\right)} &\equiv_s \textstyle\redotimes{\pasum{p}{m}{n}}
                 \otimes \left(\bigotimes_x \redotimes{\h}\right)
             \end{align*}
               As for the
             condition on signatures, note that $\left(\bigotimes_x
             \redotimes{\h_1}\right)$ and all its subterms must have empty
             signatures so that the only non-trivial cases must be subterms of
             $\redotimes{\h_2}$, which then follow by the inductive hypothesis.
             For the case of $\pasum{p}{m}{n}$, it follows vacuously as
             $\pasum{p}{m}{n}$ has no tensor subterm.  \item
             $\redotimes{\pasum{p}{m}{n} \otimes \pasum{p'}{m'}{n'}}$ follows
             from the definition of $\otimes$ on path-sums:
             \begin{align*}
                 &\semFock{\redotimes{\pasum{p}{m}{n} \otimes \pasum{p'}{m'}{n'}}} \\
                 =& \semFock{\pasum{p + p'}{m \otimes m'}{n \cdot n'}} \\
                 =& e^{2 \pi i \semFock{p + p'}} 
                 \semFock{n \cdot n'} \semFock{m \otimes m'} \\
                 =& e^{2 \pi i (\semFock{p} + \semFock{p'})} 
                 (\semFock{n} \cdot \semFock{n'}) (\semFock{m} \otimes
                 \semFock{m'}) \\
                 =& \left(e^{2 \pi i \semFock{p}} \semFock{n} \semFock{m}\right) \otimes
                 \left(e^{2 \pi i \semFock{p'}} \semFock{n'} \semFock{m'}\right)
                 \\
                 =&  \semFock{\pasum{p}{m}{n}}  \otimes
                 \semFock{\pasum{p'}{m'}{n'}} \\
                 =& \semFock{\pasum{p}{m}{n} \otimes \pasum{p'}{m'}{n'}}
             \end{align*}
             With the condition on signatures following vacuously.
     \end{enumerate}
\end{proof}

\begin{proof}
    Let $\h_1$ and $\h_2$ be \ihps{} defined over the signature  $\signature
    \equaldef \signature(\h_1) = \signature(\h_2)$ such that $\h_1 \ultraperp
    \h_2$, then:
    \[
        \mathrm{tr}_{\hilbert(\qbitSet \cap
        \signature)}(\ket{\semFock{\h_1}}\bra{\semFock{\h_1}}) \perp
        \mathrm{tr}_{\hilbert(\qbitSet \cap
        \signature)}(\ket{\semFock{\h_2}}\bra{\semFock{\h_2}})
    \]
    As $\ket{\semFock{\h_1}}\bra{\semFock{\h_1}}$ and
    $\ket{\semFock{\h_2}}\bra{\semFock{\h_2}}$ are positive semi-definite
    operators, the only way for their partial traces to be orthogonal is if
    there is no basis vector of $\fockSpace({\classAddressSet \cap \signature}) =
    \hilbert[\classAddressSet \cap s] \otimes \fockSpace_{\pastSpaceDecor}$, that is, a
    history $\eta \in (\mathtt{List}(\B) \times \basis(\signature \cap
    \classAddressSet))$, for which both $\semFock{\h_1}(\eta) \in
    \hilbert(\qbitSet \cap \signature)$ and $\semFock{\h_2}(\eta) \in \hilbert(\qbitSet \cap
    \signature)$ are non-zero. As such, for a history $\eta$, we have
    \[
        \ket{\semFock{\h_1}(\eta)} \bra{\semFock{\h_2}(\eta)} = 0,
    \]so that
    \begin{align*}
        &\ket{\semFock{\h_1}(\eta) + \semFock{\h_2}(\eta)}
        \bra{\semFock{\h_1}(\eta) + \semFock{\h_2}(\eta)}\\
        =&
        \ket{\semFock{\h_1}(\eta)} \bra{\semFock{\h_1}(\eta)} +
        \ket{\semFock{\h_2}(\eta)} \bra{\semFock{\h_2}(\eta)}
    \end{align*}

    With this property in mind, the CQ state interpretation behaves linearly on
    this sum. Concretely, for a $c \in \basis(\signature
    \cap \classAddressSet)$,
    \begin{align*}
        \semRho{ \h_1 + \h_2}(c)
        &= (\ket c \bra c \otimes I_{\fockSpace_{\pastSpaceDecor}})
        \mathrm{tr}_{\fockSpace_{\pastSpaceDecor}}(\ket{\semFock{\h_1+
        \h_2}}\bra{\semFock{\h_1+ \h_2}}) \\
        &= \sum_{\eta \in \mathtt{List}(\B) \times \pastKet{c}} \ket{\semFock{\h_1+
        \h_2}(\eta)} \bra{\semFock{\h_1+ \h_2}(\eta)} \\
        &= \sum_{\eta \in \mathtt{List}(\B) \times \pastKet{c}}
        \ket{\semFock{\h_1}(\eta) + \semFock{\h_2}(\eta)}
        \bra{\semFock{\h_1}(\eta) + \semFock{\h_2}(\eta)} \\
        &= \sum_{\eta \in \mathtt{List}(\B) \times \pastKet{c}}
        \ket{\semFock{\h_1}(\eta)} \bra{\semFock{\h_1}(\eta)} +
        \ket{\semFock{\h_2}(\eta)} \bra{\semFock{\h_2}(\eta)} \\
        \semRho{ \h_1 + \h_2}(c) &= \semRho{\h_1}(c) + \semRho{\h_2}(c)
    \end{align*}

    It remains to show that the same holds for $\h_1 \sqpsadd \h_2$; i.e., that
    \[
        \semRho{\h_1 \sqpsadd \h_2}(c) = \semRho{\h_1}(c) + \semRho{\h_2}(c).
    \]
    Indeed, since $\semFock{\h_1 \sqpsadd \h_2} = \pastKet{0} \otimes \semFock{\h_1} +
    \pastKet{1} \otimes \semFock{\h_2}$, we have
    \begin{align*}
        \semRho{ \h_1 \sqpsadd \h_2}(c)
        &= (\ket c \bra c \otimes I_{\fockSpace_{\pastSpaceDecor}})
        \mathrm{tr}_{\fockSpace_{\pastSpaceDecor}}(\ket{\semFock{\h_1 \sqpsadd
        \h_2}}\bra{\semFock{\h_1 \sqpsadd \h_2}}) \\
        &= \sum_{\eta \in \mathtt{List}(\B) \times \pastKet{c}} \ket{\semFock{\h_1
                \sqpsadd
        \h_2}(\eta)} \bra{\semFock{\h_1 \sqpsadd \h_2}(\eta)} \\
        &= \sum_{\eta \in \mathtt{List}(\B) \times \pastKet{c}}
        \ket{\pastKet{0} \otimes \semFock{\h_1}(\eta) + \pastKet{1} \otimes
        \semFock{\h_2}(\eta)} \\
        & \phantom{=\sum_{\eta \in \mathtt{List}(\B) \times \pastKet{c}}}
        \bra{\pastKet{0} \otimes \semFock{\h_1}(\eta) + \pastKet{1} \otimes
        \semFock{\h_2}(\eta)} \\
        &= \sum_{\eta \in \mathtt{List}(\B) \times \pastKet{c}}
        \ket{\pastKet{0} \otimes \semFock{\h_1}(\eta)} \bra{\pastKet{0} \otimes
        \semFock{\h_1}(\eta)} \\
        &\phantom{=\sum_{\eta \in \mathtt{List}(\B) \times \pastKet{c}}}
        + \ket{\pastKet{1} \otimes \semFock{\h_2}(\eta)}
        \bra{\pastKet{1} \otimes \semFock{\h_2}(\eta)} \\
        &= \sum_{\eta \in \mathtt{List}(\B) \times \pastKet{c}}
        \ket{\semFock{\h_1}(\eta)} \bra{\semFock{\h_1}(\eta)} +
        \ket{\semFock{\h_2}(\eta)} \bra{\semFock{\h_2}(\eta)} \\
        \semRho{ \h_1 \sqpsadd \h_2}(c) &= \semRho{\h_1}(c) + \semRho{\h_2}(c)
    \end{align*}
    We then have that $\semRho{\h_1 + \h_2} = \semRho{\h_1 \sqpsadd \h_2}$;
    i.e., $\h_1 + \h_2 \equiv \h_1 \sqpsadd \h_2$.\qed%
\end{proof}

\begin{restatable}[Symbolic execution preserves equivalence]{lemma}{restateWelldefinednessOperational}%
    \label{lem:well-definedness-operational}
    For any $\prog \in \programSet$, and $\h_1, \h_2, \h_3 \in \hpsType$, such
    that $\hoare{\h_1}{\prog}{\h_2}$ and $\signature(\h_1) =
    \signature(\h_3)$, we have
    \begin{enumerate}[(i)]
        \item $\signature(\h_2) = \applySig{\prog}{\signature(\h_1)}$, and
        \item If $\h_1 \equiv \h_3$, then, $\forall \h_4 \in \hpsType,
            \left(\hoare{\h_3}{\prog}{\h_4} \implies \signature(\h_2) =
                \signature(\h_4) \land \h_2 \equiv
            \h_4\right)$.
    \end{enumerate}
\end{restatable}

\begin{proof}

    For (i), we proceed by induction on the rules of the logic for forming
    $\hoare{\h_1}{\prog}{\h_2}$.
    
    \begin{enumerate}
        \item{}\ref{rule:skip}: In this case, $\h_1 = \h_2$ and $\h_3 = \h_4$,
            so that $\h_2 = \h_1 \equiv \h_3 = \h_4$, and $\signature(\h_2) =
            \signature(\h_1) = \applySig{\prog}{\signature(\h_1)}$
        \item{}\ref{rule:seq}: We have $\prog = \prog_1 ; \prog_2$, and there
            is a $\h_1'$ such that $\hoare{\h_1}{\prog_1}{\h_1'}$ and
            $\hoare{\h_1'}{\prog_2}{\h_2}$. By induction on the first
            derivation, we have that $\signature(\h_1') =
            \applySig{\prog_1}{\signature(\h_1)}$. By
            induction on the second derivation, we have that
            $\signature(\h_2) = \applySig{\prog_2}{\signature(\h_1')}$. As such,
            $\signature(\h_2) =
            \applySig{\prog_2}{\applySig{\prog_1}{\signature(\h_1)}} =
            \applySig{\prog}{\signature(\h_1)}$.

            \begin{jiinote}{}
                Equiv now only exists in two forms. Make sure the explanation of
                the proof still holds.
            \end{jiinote}
        \item{}\ref{rule:equiv}: There exists $\h_2'$ such that $\h_2 \equiv
            \h_2'$ and $\hoare{\h_1}{\prog}{\h_2'}$. By induction, we have that
            $\signature(\h_2') = \applySig{\prog}{\signature(\h_1)}$, and
            $\h_2 \equiv \h_2'$, we have that $\signature(\h_2) =
            \signature(\h_2') = \applySig{\prog}{\signature(\h_1)}$.
        \item{}\ref{rule:unitary},~\ref{rule:assign-bool},~\ref{rule:assign-int}, and~\ref{rule:measure}:
            These programs don't affect signatures, so $
            \signature(\h_2) = \signature(\h_1) =
            \applySig{\prog}{\signature(\h_1)}$
        \item{}\ref{rule:qinit}: We have $\prog = \qubitInitProg \qbit$,
            and $\h_2 = \h_1 \otimes \ket[\qbit]{0}$, so that $\signature(\h_2) =
            \signature(\h_1) \cup \{\qbit\} =
            \applySig{\prog}{\signature(\h_1)}$.
        \item{}\ref{rule:cinit}: Similar to~\ref{rule:qinit}, we have
            $\signature(\h_2) = \signature(\h_1) \cup \{\cbit\} =
            \applySig{\prog}{\signature(\h_1)}$.
        \item{}\ref{rule:intinit}: Similar to~\ref{rule:qinit}, we have
            $\signature(\h_2) = \signature(\h_1) \cup \{\cInt\} =
            \applySig{\prog}{\signature(\h_1)}$.
        \item{}\ref{rule:if}: The derivation end with the following rule for
            $\h_2 = \h_2^{\bool} \psadd \h_2^{\lnot \bool}$:
            \begin{prooftree}
                \AxiomC{$\hoare{\bool * \h_1}{\prog_1}{\h_2^{\bool}}$}
                \AxiomC{$\hoare{(1 \oplus \bool) * \h_1}{\prog_2}{\h_2^{\lnot \bool}}$}
                \RightLabel{\ref{rule:if}}
                \BinaryInfC{$\hoare{\h_1}{\ifthene{\bool}{\prog_1}{\prog_2}}{\h_2^{\bool} \sqpsadd \h_2^{\lnot \bool}}$}
            \end{prooftree}
            By the inductive hypothesis and the assumption that the rules are
            applied on valid program/\ihps{} pairs only, we have that
            $\signature(\h_2^{\bool}) = \applySig{\prog_1}{\signature(\bool *
            \h_1)} = \applySig{\prog_1}{\signature(\h_1)}
            $ and $\signature(\h_2^{\lnot \bool}) = \applySig{\prog_2}{
            \signature((1 \oplus \bool) * \h_1)} = \applySig{\prog_2}{
            \signature(\h_1)}$. As such, $\signature(\h_2) =
            \signature(\h_2^{\bool}) = \signature(\h_2^{\lnot \bool}) =
            \applySig{\prog_1}{\signature(\h_1)} = \applySig{\prog_2}{
            \signature(\h_1)} = \applySig{\prog}{\signature(\h_1)}$.
        \item{}\ref{rule:while}:
            By the assumption that the rules are applied on valid
            program/\ihps{} pairs only, $\prog = \while{\bool}{\prog'}$ does not
            affect signatures and neither does $\prog'$, so by the inductive
            hypothesis, we have that $\signature(\h_2) = \signature(\h_1) =
            \applySig{\prog}{\signature(\h_1)}$.
    \end{enumerate}

    \noindent As for (ii), we use \cref{thm:soundness} to write:
    \begin{align*}
        \sem{\prog}(\semRho{\h_1}) &= \semRho{\h_2} \\
        \sem{\prog}(\semRho{\h_3}) &= \semRho{\h_4}
    \end{align*}
    But then if $\h_1 \equiv \h_3$, then $\semRho{\h_1} = \semRho{\h_3}$, which
    then implies that $\semRho{\h_2} = \semRho{\h_4}$ (by functionality of
    $\sem{\prog}$), and therefore that $\h_2 \equiv \h_4$.
\end{proof}

\restatePhysicality*
\begin{proof}

    \romain{changé le label car il y a un conflit}
    We start with a few lemmas that we will need in the main proof.

    \begin{lemma}[Soundness of filtering]%
        \label{lem:soundness-filtering}

        Let $\h \in \hpsType$ be a closed \ihps{}, and $\bool$ a boolean of the
        $\lang$ such that $\mathrm{Var}(\bool) \subseteq \signature(\h)$, then,
        \[
            \semRho{\bool * \h} = \widehat{F_{\bool}}(\semRho{\h})
        \]
    \end{lemma}

    \begin{proof}
        For all $\h' \in \hpsType$, and $c \in \basis(\signature(\h')
        \cap \classAddressSet)$, we have
        \begin{align*}
            \semRho{\h'}(c) 
            &= (\ket c \bra c \otimes I_{\fockSpace_{\pastSpaceDecor}})
            \mathrm{tr}_{\fockSpace_{\pastSpaceDecor}}(\ket{\semFock{\bool *
            \h'}}\bra{\semFock{\h'}}) \\
            &= (\ket c \bra c \otimes I_{\fockSpace_{\pastSpaceDecor}})
            \sum_{\eta \in \basis_{\pastSpaceDecor} \times
            \basis(\signature(\h') \cap \classAddressSet)}
            \ket{\semFock{\h'}(\eta)} \bra{\semFock{\h'}(\eta)} \\
            &= \sum_{\eta \in \basis_{\pastSpaceDecor} \times \pastKet{c}}
            \ket{\semFock{\h'}(\eta)} \bra{\semFock{\h'}(\eta)}
        \end{align*}

    \noindent Claim: Let $\h$, $\bool$ be as in the statement of the lemma and $c \in
    \basis(\signature(\h) \cap \classAddressSet)$, then:
        \[
            \forall \eta \in \basis_{\pastSpaceDecor} \times \pastKet{c},
            \semFock{\bool * \h}(\eta) = 
            \begin{cases}
                \semFock{\h}(\eta) & \text{if } \eval{\bool}{c} = 1
                \\
                0 & \text{otherwise}
            \end{cases}
        \]
    From this claim, we show that 
    \begin{align*}
        \semRho{\bool * \h}(c) 
        &= \sum_{\eta \in \basis_{\pastSpaceDecor} \times \pastKet{c}}
        \ket{\semFock{\bool * \h}(\eta)} \bra{\semFock{\bool * \h}(\eta)} \\
        &= \sum_{\eta \in \basis_{\pastSpaceDecor} \times \pastKet{c}}
        \begin{cases}
            \ket{\semFock{\h}(\eta)} \bra{\semFock{\h}(\eta)} & \text{if }
            \eval{\bool}{c} = 1
            \\
            0 & \text{otherwise}
        \end{cases} \\
        &= \begin{cases}
            \sum_{\eta \in \basis_{\pastSpaceDecor} \times
            \pastKet{c}} \ket{\semFock{\h}(\eta)} \bra{\semFock{\h}(\eta)} &
            \text{if } \eval{\bool}{c} = 1
            \\
            0 & \text{otherwise}
        \end{cases} \\
        &= \begin{cases}
            (\ket c \bra c \otimes I_{\fockSpace_{\pastSpaceDecor}})
            \mathrm{tr}_{\fockSpace_{\pastSpaceDecor}}(\ket{\semFock{\h}}\bra{\semFock{\h}})
            & \text{if } \eval{\bool}{c} = 1
            \\
            0 & \text{otherwise}
        \end{cases} \\        
        &= \begin{cases}
            \semRho{\h}(c) & \text{if } \eval{\bool}{c} = 1
            \\
            0 & \text{otherwise}
        \end{cases} \\        
        &= \widehat{F_{\bool}}(\semRho{\h})(c)
    \end{align*}

    Proof of claim: We proceed by induction on the structure of $\h$. For the
    base case $\h = \pasum{p}{m}{n}$, we have $\semFock{\bool * \h} = e^{2 \pi i
    \semFock{p}} \semFock{\eval{\bool}{m} \cdot n} \semFock{m}$, and for any
    $\eta$, so that
            \[
                \semFock{\pasum{p}{m}{\eval{\bool}{m} \cdot n} }(\eta) =
                \begin{cases}
                    e^{2 \pi i \semFock{p}} \semFock{n}
                    \semFock{m} & \text{if } \eval{\bool}{\eta} = 1
                    \\
                    0 & \text{otherwise}
            \end{cases}
            \]
            The remaining cases are straightforward, with the note that, since
            $\h$ is considered to be in its $\otimes$-reduced form, when
            considering $\h = \h_1 \otimes \h_2$, we must have either
            $\mathrm{Var}(\bool) \subseteq \signature(\h_1)$ or
            $\mathrm{Var}(\bool) \subseteq \signature(\h_2)$, and in that case,
            the other term is left unchanged by the filtering.
    \end{proof}

    \begin{lemma}[Soundness of projection]%
        \label{lem:soundness-projection}
        Let $\h \in \hpsType$ be a closed \ihps{}, $\address \in \addressSet$ an
        address, and $\programTerm$ be a term of $\lang$ of the same type as
        $\address$, or a qubit $\programTerm \in \qbitSet$ (for which
        $\address \in \cbitSet$), and suppose $\mathrm{Var}(\programTerm) \cup
        \{\address\} \subseteq \signature(\h)$, then,
        \[
            \semRho{\h[\address \ot \programTerm]} = \widehat{P_{\address \to
            \programTerm}}(\semRho{\h})
        \]
    \end{lemma}
    \begin{proof}
        Start by rewriting $\h$ into an equivalent form $\h' \equaldef \redsqpsadd{\h}$ which is
        free of $\sqpsadd$ as such:
        \begin{align*}
            \redsqpsadd{\pasum{p}{m}{n}} &\equaldef \pasum{p}{m}{n} \\
            \redsqpsadd{\h_1 \psadd \h_2} &\equaldef \redsqpsadd{\h_1} \psadd
            \redsqpsadd{\h_2} \\
            \redsqpsadd{\h_1 \sqpsadd \h_2} &\equaldef \pastKet[\B]{0}
            \otimes \redsqpsadd{\h_1} + \pastKet[\B]{1} \otimes
            \redsqpsadd{\h_2} \\
            \redsqpsadd{\h_1 \otimes \h_2} &\equaldef \redsqpsadd{\h_1} \otimes
            \redsqpsadd{\h_2} \\
            \redsqpsadd{\lim_\intVar \h} &\equaldef \lim_\intVar \redsqpsadd{\h}
            \\
            \redsqpsadd{\sum_{\var} \h} &\equaldef \sum_{\var} \redsqpsadd{\h} \\
            \redsqpsadd{\bigotimes_x \h} &\equaldef \bigotimes_x \redsqpsadd{\h}
        \end{align*}

        Then, we show, by induction on $\h$, that $\semFock{\h[\address \ot
        \programTerm]} = \Pi_{\address \ot \programTerm, \signature(\h)}
        \semFock{\redsqpsadd{\h}}$, where $\Pi_{\address \ot \programTerm, \signature(\h)}$
        is the linear map defined as:
        \[
            \Pi_{\address \ot \programTerm, \signature} \equaldef \sum_{\substack{\varValue \in
                    \text{type}(\address) \\ \classicalState \in
                \basis(\signature \setminus \{\address\})}} ({\{\dot{ \address}\}
            }_{\text{type}(\address)} \otimes \ket{\classicalState[\address \mapsto
            \programTerm]}) \bra{\classicalState[\address \mapsto \varValue]}
        \]

        The inductive cases $\h_1 + \h_2$, $\sum_{\var} \h$ and $\lim_\intVar \h$
        as well as the tensors $\h_1 \otimes \h_2$ and $\bigotimes_{\var} \h$
        follow relatively immediately from the linearity of $\Pi$, as
        for the base case, $\h = \pasum{p}{m}{n}$,
        \begin{align*}
            &\semFock{\pasum{p}{m}{n}[\address \ot \programTerm]} \\
            =&
            \semFock{\pasum{p}{\pastKet{\eval{\programTerm}{m}_{\text{type}(\address)}
                        \otimes m[\address \mapsto
        \eval{\programTerm}{m}]}}{n}} \\
            =& \sum_{\substack{\varValue \in
                    \text{type}(\address) \\ \classicalState \in 
                    \basis(\signature(\h) \setminus
            \{\address\})}} e^{2 \pi i \semFock{p}} \semFock{n}
            \pastKet[\text{type}(\address)]{\varValue} \otimes
            \semFock{\classicalState[\address \mapsto \varValue]} \\
                    =& \Pi_{\address \ot \programTerm,
                \signature(\h)} \semFock{\pasum{p}{m}{n}}
        \end{align*}
        Finally, the slightly more delicate case is that of $\h = \h_1 \sqpsadd
        \h_2$. In that case, $\redsqpsadd{\h} = \pastKet[\B]{0}
        \otimes \redsqpsadd{\h_1} + \pastKet[\B]{1} \otimes
        \redsqpsadd{\h_2}$, and we have:
        \begin{align*}
            & \semFock{(\h_1 \sqpsadd \h_2)[\address \ot \programTerm]} \\
            =& \semFock{\h_1[\address \ot \programTerm] \sqpsadd \h_2[\address
            \ot \programTerm]}  \\
                =& \pastKet[\B]{0} \otimes \semFock{\h_1[\address \ot \programTerm]} +
                \pastKet[\B]{1} \otimes \semFock{\h_2[\address \ot \programTerm]} \\
                =& \pastKet[\B]{0} \otimes \Pi_{\address \ot \programTerm,
                \signature(\h)} \semFock{\redsqpsadd{\h_1}} +
                \pastKet[\B]{1} \otimes \Pi_{\address \ot \programTerm,
                \signature(\h)} \semFock{\redsqpsadd{\h_2}} \\
                    =& \Pi_{\address \ot \programTerm, \signature(\h)} \left(
                        \pastKet[\B]{0} \otimes
                        \semFock{\redsqpsadd{\h_1}} +
                        \pastKet[\B]{1} \otimes
                        \semFock{\redsqpsadd{\h_2}} \right) \\
        \end{align*}
        \jiwarning{This proof is simply incorrect because it mixes up the
            $\pastKet{-}$ in the wrong order; I don't feel like fixing it atm,
        so I'll see later}

    \end{proof}

    \begin{lemma}[Soundness of unitary application]%
        \label{lem:soundness-unitary}
        Let $\h$ be a closed \ihps{} and $\bar \qbit$ be a tuple of qubits in
        $\signature(\h)$, then,
        \[
            \semRho{\applyU{\unitary(\bar \qbit)}(\h)} = \widehat{U}_{\bar \qbit}
            \semRho{\h} \widehat{U}_{\bar \qbit}^\dagger \\
        \]
    \end{lemma}
    \begin{proof}
        This is again by induction on the structure of $\h$ with the inductive
        cases being almost trivial. We elaborate on the base case for the
        different unitaries, and show more strongly that $\semFock{\applyU{
        \unitary(\bar \qbit)}(\h)} = \unitary(\bar \qbit) \semFock{\h}$.
        \begin{itemize}
            \item For $\unitary = H$, we have
                \begin{align*}
                    &\semFock{\applyU{H(\qbit)}(\pasum{p}{m}{n})} \\
                    =& \semFock{\sum_{\boolVar} \pasum{p + \frac{\boolVar \cdot
                    \eval{\qbit}{m}} {2}}{m[\qbit \mapsto \boolVar]}{\frac n
                {\sqrt 2}}} \\
                    =& \sum_{\boolVar} e^{2 \pi i \left(\semFock{p} +
                    \frac{\boolVar \cdot \eval{\qbit}{m}} {2}\right)}
                    \semFock{\frac n {\sqrt 2}} \semFock{m[\qbit \mapsto
                    \boolVar]} \\
                    =& \sum_{\boolVar} \frac{e^{2 \pi i
                    \semFock{p}}}{\sqrt 2} e^{\pi i \boolVar \cdot
                    \eval{\qbit}{m}} \semFock{n} \semFock{m[\qbit
                    \mapsto \boolVar]} \\
                    =& \begin{cases}
                        \frac{e^{2 \pi i \semFock{p}}}{\sqrt 2} \semFock{n}
                        (\semFock{m[\qbit \mapsto 0]} + \semFock{m[\qbit
                        \mapsto 1]}) & \text{if } \eval{\qbit}{m} = 0
                        \\
                        \frac{e^{2 \pi i \semFock{p}}}{\sqrt 2} \semFock{n}
                        (\semFock{m[\qbit \mapsto 0]} - \semFock{m[\qbit
                        \mapsto 1]}) & \text{if } \eval{\qbit}{m} = 1
                    \end{cases} \\
                    =& H(\qbit) e^{2 \pi i \semFock{p}} \semFock{n}
                    \semFock{m} \\
                    =& H(\qbit) \semFock{\pasum{p}{m}{n}}
                \end{align*}
            \item For $\unitary = X$, we have
                \begin{align*}
                    &\semFock{\applyU{X(\qbit)}(\pasum{p}{m}{n})} \\
                    =& \semFock{\pasum{p}{m[\qbit \mapsto 1 \oplus
                    \eval{\qbit}{m}]}{n}} \\
                    =& e^{2 \pi i \semFock{p}} \semFock{n} \semFock{m[\qbit
                    \mapsto 1 \oplus \eval{\qbit}{m}]} \\
                        =& \begin{cases}
                            e^{2 \pi i \semFock{p}} \semFock{n}
                            \semFock{m[\qbit \mapsto 1]} & \text{if }
                            \eval{\qbit}{m} = 0
                            \\
                            e^{2 \pi i \semFock{p}} \semFock{n}
                            \semFock{m[\qbit \mapsto 0]} & \text{if }
                            \eval{\qbit}{m} = 1
                        \end{cases} \\
                            =& X(\qbit) e^{2 \pi i \semFock{p}} \semFock{n}
                            \semFock{m} \\
                            =& X(\qbit) \semFock{\pasum{p}{m}{n}}
                \end{align*}
            \item For $\unitary = Z$, we have
                \begin{align*}
                    &\semFock{\applyU{Z(\qbit)}(\pasum{p}{m}{n})} \\
                    =& \semFock{\pasum{p + \eval{\qbit}{m}/2}{m}{n}} \\
                    =& e^{2 \pi i (\semFock{p} + \eval{\qbit}{m}/2)}
                    \semFock{n} \semFock{m} \\
                        =& \begin{cases}
                            e^{2 \pi i \semFock{p}} \semFock{n}
                            \semFock{m} & \text{if } \eval{\qbit}{m} = 0
                            \\
                            - e^{2 \pi i \semFock{p}} \semFock{n}
                            \semFock{m} & \text{if } \eval{\qbit}{m} = 1
                        \end{cases} \\
                            =& Z(\qbit) e^{2 \pi i \semFock{p}} \semFock{n}
                            \semFock{m} \\
                            =& Z(\qbit) \semFock{\pasum{p}{m}{n}}
                \end{align*}
            \item For $\unitary = \texttt{CNOT}$, we have
                \begin{align*}
                    &\semFock{\applyU{\texttt{CNOT}(\qbit_1, \qbit_2)}(\pasum{p}{m}{n})} \\
                    =& \semFock{\pasum{p}{m[\qbit_2 \mapsto
                    \eval{\qbit_2}{m} \oplus \eval{\qbit_1}{m}]}{n}} \\
                            =& e^{2 \pi i \semFock{p}} \semFock{n}
                            \semFock{m[\qbit_2 \mapsto \eval{\qbit_2}{m} \oplus
                            \eval{\qbit_1}{m}]} \\
                            =& \begin{cases}
                                e^{2 \pi i \semFock{p}} \semFock{n}
                                \semFock{m[\qbit_2 \mapsto 0]} &
                                \text{if } \eval{\qbit_1}{m} = \eval{\qbit_2}{m}
                                \\
                                e^{2 \pi i \semFock{p}} \semFock{n}
                                \semFock{m[\qbit_2 \mapsto 1]} &
                                \text{if } \eval{\qbit_1}{m} \neq \eval{\qbit_2}{m}
                            \end{cases} \\
                                =& \texttt{CNOT}(\qbit_1, \qbit_2) e^{2 \pi i
                                \semFock{p}} \semFock{n}
                                \semFock{m} \\
                                    =& \texttt{CNOT}(\qbit_1, \qbit_2)
                                    \semFock{\pasum{p}{m}{n}}
                \end{align*}
            \item For $\unitary = R_k$, we have
                \begin{align*}
                    &\semFock{\applyU{R_k(\qbit)}(\pasum{p}{m}{n})} \\
                    =& \semFock{\pasum{p + \frac{\eval{\qbit}{m}}{2^k}}{m}{n}}
                    \\
                    =& \begin{cases}
                        e^{2 \pi i \semFock{p}} \semFock{n} \semFock{m} &
                        \text{if } \eval{\qbit}{m} = 0 \\
                        e^{2\pi i \cdot \frac{1}{2^k}} \cdot e^{2 \pi i
                    \semFock{p}} \semFock{n} \semFock{m} & \text{if }
                    \eval{\qbit}{m} = 1 \\
                    \end{cases}\\
                        =& R_k \semFock{\pasum{p}{m}{n}}
                \end{align*}
                
        \end{itemize}
        \qed%
    \end{proof}

    For the main theorem (\cref{thm:soundness}), we proceed by induction on the derivation \hoare{\h_1}{\prog}{\h_2}.
    \begin{itemize}
        \item{}\ref{rule:skip}: In this case, $\h_1 = \h_2$, and
            $\sem{\skipProg} = \text{id}$, so that
            $\sem{\skipProg}(\semRho{\h_1}) = \semRho{\h_1} = \semRho{\h_2}$.
        \item{}\ref{rule:seq}: We have $\prog = \prog_1 ; \prog_2$, and there
            is a $\h_1'$ such that \hoare{\h_1}{\prog_1}{\h_1'} and
            \hoare{\h_1'}{\prog_2}{\h_2}. By induction on the first
            derivation, we can assume that $\sem{\prog_1}(\semRho{\h_1}) =
            \semRho{\h_1'}$. By induction on the second derivation, we have that
            $\sem{\prog_2}(\semRho{\h_1'}) = \semRho{\h_2}$. As such,
            $\sem{\prog}(\semRho{\h_1}) = (\sem{\prog_2} \circ \sem{\prog_1})(\semRho{\h_1})
            = \semRho{\h_2}$.
        \item{}\ref{rule:equiv}: There exists $\h_2'$ such that $\h_2 \equiv
            \h_2'$ and \hoare{\h_1}{\prog}{\h_2'}. By induction, we have that
            $\sem{\prog}(\semRho{\h_1}) = \semRho{\h_2'}$, and
            $\h_2 \equiv \h_2'$, we have that $\semRho{\h_2} =
            \semRho{\h_2'}$, so that $\sem{\prog}(\semRho{\h_1}) =
            \semRho{\h_2}$.
        \item{}\ref{rule:unitary}: This follows from the soundness of $\applyU{
                \unitary(\bar \qbit) }$ with respect to $\sem{\unitary(
                    \bar \qbit)}$, which is proven separately in
                Theorem~\ref{lem:soundness-unitary}.
        \item{}\ref{rule:assign-bool},~\ref{rule:assign-int},~\ref{rule:measure}:
        Those follow directly from the soundness of the projection
        Theorem~\ref{lem:soundness-projection}.
        \item{}\ref{rule:qinit},~\ref{rule:cinit},~\ref{rule:intinit}:
            Let's consider the case of~\ref{rule:qinit} first. We have $\prog =
            \qubitInitProg \qbit$, and $\h_2 = \h_1 \otimes \ket[\qbit]{0}$, and
            \begin{align*}
                \semRho{\h \otimes {\ket 0}_\qubit}
                &= \widehat{\dephase_{\signature \cap \classAddressSet}}\left(
                \mathrm{tr}_{\fockSpace_{\pastSpaceDecor}}\left(\ket{\semFock{\h \otimes
                {\ket 0}_\qubit}}\bra{\semFock{\h \otimes {\ket
                0}_\qubit}}\right)\right) \\
                &= \widehat{\dephase_{\signature \cap \classAddressSet}}\left(
                    \mathrm{tr}_{\fockSpace_{\pastSpaceDecor}}\left(\ket{\semFock{\h}} \bra{\semFock{\h}} \otimes
                    \ket[\qbit]{0} \bra{0}_\qubit\right)\right) \\
                &= \widehat{\dephase_{\signature \cap \classAddressSet}}\left(
                    \mathrm{tr}_{\fockSpace_{\pastSpaceDecor}}\left(\ket{\semFock{\h}} \bra{\semFock{\h}}\right) \otimes
                    \ket[\qbit]{0} \bra{0}_\qubit\right) \\
                &= \widehat{\dephase_{\signature \cap \classAddressSet}}\left(
                    \mathrm{tr}_{\fockSpace_{\pastSpaceDecor}}\left(\ket{\semFock{\h}} \bra{\semFock{\h}}\right)\right) \otimes
                    \ket[\qbit]{0} \bra{0}_\qubit \\
                &= \semRho{\h} \otimes \ket[\qbit]{0} \bra{0}_\qubit 
            \end{align*}
            where $\signature = \applySig{\qubitInitProg \qbit}{\signature(\h)}
        = \signature(\h \otimes {\ket{0}}_\qbit)$. We note that ${\ket
        0}_{\qbit} {\bra 0}_{\qbit}$ factors out of
        $\mathrm{tr}_{\fockSpace_{\pastSpaceDecor}}$ because $\qbit \not \in
        \pastSpaceDecor$. This is also the
        case for~\ref{rule:cinit} and~\ref{rule:intinit}. As for the factoring
        out of $\widehat{\dephase_{\signature \cap \classAddressSet}}$,
        for~\ref{rule:qinit}, this is again because $\qbit \not \in
        \classAddressSet$,
        while for~\ref{rule:cinit} and~\ref{rule:intinit}, this is because we're
        applying dephasing on a basis state ${\ket 0}_{\cbit}\bra{0}_{\cbit}$ or 
        ${\ket 0}_{\cInt}\bra{0}_{\cInt}$, which is invariant under dephasing.

        \item{}\ref{rule:if}: 
            We have $\prog = \ifthene{\bool}{\prog_1}{\prog_2}$ and $\h_2 =
            \h_2^{\top} \sqpsadd \h_2^{\bot}$, where $\hoare{\bool *
            \h_1}{\prog_1}{\h_2^{\top}}$ and $\hoare{(1 \oplus \bool) *
            \h_1}{\prog_2}{\h_2^{\bot}}$. Then, by the inductive hypothesis, we
            have that $\sem{\prog_1}(\semRho{\bool * \h_1}) =
            \semRho{\h_2^{\top}}$ and $\sem{\prog_2}(\semRho{(1 \oplus \bool) *
            \h_1}) = \semRho{\h_2^{\bot}}$. But then, by
            Theorem~\ref{lem:soundness-filtering}, we have that
            \[
                \sem{\prog_1}(\semRho{\bool * \h_1}) =
                \sem{\prog_1}(\widehat{F_{\bool}}(\semRho{\h_1})) =
                \semRho{\h_2^{\top}}
            \]
            and
            \[
                \sem{\prog_2}(\semRho{(1 \oplus \bool) * \h_1}) =
                \sem{\prog_2}(\widehat{F_{\lnot \bool}}(\semRho{\h_1})) =
                \semRho{\h_2^{\bot}}
            \]
            So that \begin{align*}
                \sem{\prog}(\semRho{\h_1})
                &= \sem{\prog_1}(\widehat{F_{\bool}}(\semRho{\h_1})) +
                \sem{\prog_2}(\widehat{F_{\lnot \bool}}(\semRho{\h_1})) \\
                &= \semRho{\h_2^{\top}} + \semRho{\h_2^{\bot}} \\
                &= \semRho{\h_2^{\top} \sqpsadd \h_2^{\bot}} \\
                &= \semRho{\h_2}
            \end{align*}
        \item{}\ref{rule:while}: We have $\prog = \while{ \bool}{\prog'}$, and
            there exists an invariant $\h_{\mathtt{inv}}[k]$ with a free integer
            variable $k$ such that $\h_1 = \h_{\mathtt{inv}}[0/k]$, $\h_2 = \lim_k
            (1 \oplus \bool) * \h$, and,
            \[
                \hoare{\h_{\mathtt{inv}}[k]}{\ifthene{\bool}{\prog'}{\skipProg}}{\h_{\mathtt{inv}}[k+1/k]}
            \]
            By the inductive hypothesis, we have that
            \[
                \sem{\ifthene{\bool}{\prog'}{\skipProg}}(\semRho{\h_{\mathtt{inv}}[k]}) =
                    \semRho{\h_{\mathtt{inv}}[k+1/k]},
            \] which, by iteration/induction on natural numbers corresponds to:
            \[
                \sem{\ifthene{\bool}{\prog'}{\skipProg}}^n(\semRho{\h_1})
                    = \semRho{\h_{\mathtt{inv}}[n/k]}.
            \]

            We can then proceed using Theorem~\ref{lem:soundness-filtering} to write:
            \[
                (\widehat{F_{\lnot \bool}})
                (\sem{\ifthene{\bool}{\prog'}{\skipProg}}^n
                (\semRho{\h_1})) = \semRho{(1 \oplus \bool) *
                \h_{\mathtt{inv}}[n/k]},
            \] which then implies
            \begin{align*}
                &\phantom{=}\sem{\while{\bool}{\prog'}}(\semRho{\h_1}) \\
                &= \lim_{n \to \infty}
                (\widehat{F_{\lnot \bool}})
                (\sem{\ifthene{\bool}{\prog'}{\skipProg}}^n
                (\semRho{\h_1})) \\
                &= \lim_{n \to \infty} \semRho{(1 \oplus \bool) *
                \h_{\mathtt{inv}}[n/k]} \\
                &= \semRho{\lim_k (1 \oplus \bool) *
                \h_{\mathtt{inv}}[k]} \\
                &= \semRho{\h_2} 
            \end{align*}
\end{itemize} \qed%

\end{proof}

\restateAdequacyBounded%
\begin{proof}
    If the program $\prog$ terminates in time $T$, then, we claim that if the
    subprogram
    $\while{\bool}{\prog'}$ appears within $\prog$, then, if we write:
    \[
        \whileBounded{T}{\bool}{\prog'} \equaldef
        \underbrace{\ifthene{\bool}{\prog'}{\skipProg}; \ldots;
                \ifthene{\bool}{\prog'}{\skipProg}}_{T \text{ times}},
    \] we have:
    \[
        \sem{\whileBounded{T}{\bool}{\prog'}} = \sem{\while{\bool}{\prog'}}
    \]
    Indeed, we have \[
        \sem{\while{\bool}{\prog'}} = \lim_{n \to \infty}
        (\widehat{F_{\lnot \bool}} \circ \sem{\ifthene{\bool}{\prog'}{\skipProg}}^n)
    \]
    However, since the program terminates in time $T$, no more than $T$
    iterations of the loop could have been executed, meaning that, within $T$
    iterations, the loop must have exited with $\lnot \bool$ being satisfied in
    all branches. As such, for all $n \geq T$, we have 

    \begin{align*}
        \sem{\while{\bool}{\prog'}}
        &= (\widehat{F_{\lnot \bool}} \circ
        \sem{\ifthene{\bool}{\prog'}{\skipProg}}^n) \\
        &=
            \sem{\ifthene{\bool}{\prog'}{\skipProg}}^T \\
        &= \sem{\whileBounded{T}{\bool}{\prog'}}
     \end{align*}

     For each \texttt{while} loop in $\prog$, we can therefore unroll it into
     a program without loops that is semantically equivalent. We will then
     produce a derivation for $(\whileBounded{T}{\bool}{\prog'}, \h')$ and
     translate it into a derivation for the ordinary \texttt{while} $(\while{\bool}{\prog'}, \h')$.  Indeed,
     for a \texttt{while}-free program, the semantics of such programs do not
     then require the use of the rule~\ref{rule:while} which was the only rule
     for which the premises are constrictive; i.e.\ it includes not only
     constraints about validity of programs on \ihps{}, but also that the program
     preserves the form of the loop invariant. By liberating ourselves from this
     constraint, we can have an~\ref{rule:equiv}-free syntax-driven derivation
    for \hoare{\h'}{\whileBounded{T}{\bool}{\prog'}}{\h''}. This derivation
     involves $T$ repetitions of the sub-derivation for the conditional
    $\hoare{\h'_i}{\ifthene{\bool}{\prog'}{\skipProg}}{\h'_{i+1}}$
     involving $T$ different \ihps{} $\h'_1, \ldots, \h'_T$. Out of these $T$
     \ihps{}, we can extract a loop invariant $\h'_{\mathtt{inv}}[k]$ for the
     original $\while{\bool}{\prog'}$ loop, by writing:
     \[
        \h'_{\mathtt{inv}}[\intVar] \equaldef 
        \pasum{0}{\emptyset}{\lift(1 \leq \intVar)} \otimes \h'_1 \sqpsadd
        \pasum{0}{\emptyset}{\lift(2 \leq \intVar)} \otimes \h'_2 \sqpsadd \cdots \sqpsadd
        \pasum{0}{\emptyset}{\lift(T \leq \intVar)} \otimes \h'_T
     \]

     Clearly, $\h'_{\mathtt{inv}}[\intVarOther/\intVar] \equiv \h'_\intVarOther$
     for all $1 \leq \intVarOther \leq T$,
     therefore, through~\ref{rule:equiv} applications, this is indeed a loop
     invariant for $\while{\bool}{\prog'}$:
      \[
          \hoare{\h'_{\mathtt{inv}}[\intVar]}{\ifthene{\bool}{\prog'}{\skipProg}}{\h'_{\mathtt{inv}}[\intVar+1/\intVar]},
      \]
      and we can finally write the derivation for \hoare{\h'}{\while{\bool}{\prog'}}{\h''} as:
     \begin{center}
         \AxiomC{$\hoare{\h'_{\mathtt{inv}}[\intVar]}{\ifthene{\bool}{\prog'}{\skipProg}}{\h'_{\mathtt{inv}}[\intVar+1/\intVar]}$}
             \RightLabel{\ref{rule:while}}
             \UnaryInfC{$\hoare{\h'}{\while{\bool}{\prog'}}{\lim_\intVar (1 \oplus \bool) * \h'_{\mathtt{inv}}[\intVar]}$}
             \DisplayProof%
         \end{center}

    We can proceed as such with derivations for all the other \texttt{while}
    constructs in $\prog$, and finally obtain that there is an $\h_2$ such that
    we can derive \hoare{\h_1}{\prog}{\h_2}.  Finally, by soundness
     (\cref{thm:soundness}), we have that $\sem{\prog}( \semRho{\h_1}) =
     \semRho{\h_2}$.\qed{}
\end{proof}
     \begin{remark}
         In view of the details of the proof, we reiterate the failure of
         adequacy in the general case with more detailed comments.
         We note, in particular, that this strategy highlights the need to be
         able to express a loop invariant for any loop. There are certain cases
         where this is possible (e.g., when the loop is known to be bounded as
         we have seen), but in general, there is no guarantee that a loop
         invariant can be expressed in the language of \ihps{}, and our
         suspicion from preliminary investigations is that the extension of
         \ihps{} to support arbitrary loop invariants is unwieldy with \ihps{}
         becoming essentially having to be as expressive as the programming
         language itself. We believe that such extensions, while perhaps
         interesting in theory, are not particularly useful in practice, and
         that they would constitute a nearly verbatim reimplementation of the
         programming language itself with little interest in terms of analysis.
     \end{remark}

\restateGenericStrategy*
\jiwarning{Review slightly hand-wavy proof}
\begin{proof}
    Let $\prog$, $\h^{\bool}[\intVar]$, $\h^{\lnot \bool}[\intVar]$, and
    $\h_{\mathtt{next}}[\intVar]$ be as in the statement, $\signature \equaldef
    \signature(\h^{\bool}) = \signature(\h^{\lnot \bool}) = \signature(
    \h_{\mathtt{next}})$, and $\cInt \not \in \signature$ be a fresh classical
    integer address. Finally, define the following \ihps{}:
    \[
        \h[\intVar] \equaldef \bigpsadd_{\intVar'} \langle 0, \lift ({\intVar'} \leq \intVar) \cdot \emptyset
        \rangle \otimes \h^{\lnot \bool}[{\intVar'}/\intVar] \otimes
        {[{\intVar'}]}_{\cInt}
        \sqpsadd \h^{\bool}[\intVar] \otimes {[\intVar]}_{\cInt}.
    \]
    We show that $\h[\intVar]$ is a loop invariant; that is, that
    \[
        \hoare{\h[\intVar]}{\ifthene{\bool}{\prog; \assign{\cInt}{\cInt + 1}}{\skipProg}}{\h[\intVar+1]}
    \]
    
    For the following, we will use the shorthand notation $\sum_{{\intVar'} = 0}^\intVar
    \h[{\intVar'}/\intVar]$ for $\sum_{\intVar'} \langle 0, \lift ({\intVar'} \leq \intVar) \cdot \emptyset \rangle
    \otimes \h[{\intVar'}/\intVar]$.

    Indeed, we can apply the rule~\ref{rule:if} of the logic as such:
    \begin{prooftree}
        \AxiomC{$\begin{array}{c}\hoare{\bool * \h[\intVar]}{\prog; \assign{\cInt}{\cInt + 1}}{\h^{\lnot \bool}[\intVar+1/\intVar] \otimes {[\intVar+1]}_{\cInt} \sqpsadd \h^{\bool}[\intVar+1/\intVar] \otimes {[\intVar+1]}_{\cInt}}
        \\ \hoare{(1 \oplus \bool) * \h[\intVar]}{\skipProg}{\sum_{{\intVar'}=0}^{\intVar} \h^{\lnot \bool}[{\intVar'}/\intVar] \otimes {[{\intVar'}]}_{\cInt}}\end{array}$}
        \RightLabel{\ref{rule:if}}
        \UnaryInfC{$\hoare{\h[\intVar]}{\ifthene{\bool}{\prog; \assign{\cInt}{\cInt + 1}}{\skipProg}}{\h[\intVar+1]}$}
    \end{prooftree}

    For the \texttt{skip} part, we apply the~\ref{rule:equiv} rule with an
    equivalence $(1 \oplus \bool) * \h[\intVar] \equiv \sum_{{\intVar'}=0}^{\intVar} \h^{\lnot
    \bool}[{\intVar'}/\intVar] \otimes {[{\intVar'}]}_{\cInt}$. Indeed, this is the case by
    propagating $(1 \oplus \bool) * -$ into the two sides of the direct sum in
    $\h[\intVar]$ and using the fact that $(1 \oplus \bool) * \h^{\bool}[\intVar] \equiv
    0$ (as $\h^{\bool}[\intVar]$ is supported exclusively on worlds where $\bool$ is
    true), and that $(1 \oplus \bool) * \h^{\lnot \bool}[{\intVar'}/\intVar] \equiv \h^{\lnot
    \bool}[{\intVar'}/\intVar]$, finishing that branch of the proof with an
    application of the~\ref{rule:skip} rule.

    For the part where we apply $\prog; \assign{\cInt}{\cInt + 1}$, we first
    need to apply the~\ref{rule:seq} rule, which reduces the derivation to an
    easy application of~\ref{rule:assign-int} one one hand, and the following,
    on another hand.
    \[
        \hoare{\bool * \h[\intVar]}{\prog}{\h^{\lnot \bool}[\intVar+1/\intVar] \otimes
        {[\intVar]}_{\cInt} \sqpsadd \h^{\bool}[\intVar+1/\intVar] \otimes {[\intVar]}_{\cInt}}
    \]

    At that point, we apply the~\ref{rule:equiv} rule with the equivalence
    $\h_{\mathtt{next}}[\intVar] \equiv \h^{\bool}[\intVar+1/\intVar] \sqpsadd \h^{\lnot
    \bool}[\intVar+1/\intVar]$, in order to reduce to $\hoare{\bool * \h[\intVar]}{\prog}{
    \h_{\mathtt{next}}[\intVar] \otimes {[\intVar]}_{\cInt}}$. We can establish said
    equivalence by writing $\h_{\mathtt{next}}[\intVar] = \bool *
    \h_{\mathtt{next}}[\intVar] + (1 \oplus \bool) * \h_{\mathtt{next}}[\intVar]$ and using
    the assumptions $\h^{\bool}[\intVar+1/\intVar] \equiv \bool * \h_{\mathtt{next}}[\intVar]$ and
    $\h^{\lnot \bool}[\intVar+1/\intVar] \equiv (1 \oplus \bool) * \h_{\mathtt{next}}[\intVar]$.

    Next, we are left with showing $\hoare{\bool * \h[\intVar]}{\prog}{
    \h_{\mathtt{next}}[\intVar] \otimes {[\intVar]}_{\cInt}}$. Once again, this is finally
    achieved with the application of~\ref{rule:equiv} 
    with the equivalence $\bool * \h[\intVar] \equiv \h^{\bool}[\intVar] \otimes
    {[\intVar]}_{\cInt}$, which itself is correct by the assumptions $\bool *
    \h^{\lnot \bool}[l/\intVar] \equiv 0$ and $\bool * \h^{\bool}[\intVar] \equiv
    \h^{\bool}[\intVar]$. This reduces the problem to $\hoare{\h^{\bool}[\intVar] \otimes
    {[\intVar]}_{\cInt}}{\prog}{\h_{\mathtt{next}}[\intVar] \otimes
    {[\intVar]}_{\cInt}}$,
    which is the last unused assumption.
\end{proof}

\section{Implementation}%
\label{app:source-code}

The source code (with Unicode characters) used for running the coin-toss example
is the following:

\begin{center}
    \begin{minipage}{0.7\textwidth}
        \begin{lstlisting}
qubit q;
bit c;
int x; // Iteration counter
while 
// Invariant index
{x : Int,
// Invariant hps
// Exiting cases
($\Sigma$_{y $\in \mathbb N$} ($\langle $0, liftC($\uparrow$((y $\leq$ x))) * liftC(1)/sqrt(liftC(2^(y+1))) $\cdot$ |1$\rangle $_q[1]_{c : $\mathbb B$}[y]_{x : $\mathbb Z$}$\rangle $)) 
// Non-exiting case
+ ($\langle $0, liftC(1)/sqrt(liftC(2^(x+1))) $\cdot$ |0$\rangle $_q[0]_{c : $\mathbb B$}[x]_{x : $\mathbb Z$}$\rangle $) } 
!c do
    H(q);
    c := measure q;
    x :Z= x + 1
done
        \end{lstlisting}
    \end{minipage}
\end{center}

We can see that in this implementation, we have specified a hint for the loop
invariant using the syntax $\whileAnnot{\bool}{\prog}{\cInt : \texttt{Int},
\h_{\mathtt{inv}}}$, and where $\h_{\mathtt{inv}}$ may be described with Unicode
characters for better readability. Minor temporary modifications of the syntax
have been made (e.g.\ \texttt{:Z=} for integer assignment or explicit liftings
in the \ihps{}) to facilitate parsing in the context of a prototype, but the
questions of comfort in reading and writing such loop invariants are relatively
minor and are being actively addressed.

\end{document}